\documentclass[journal,final,letterpaper,twoside,twocolumn,10pt]{IEEEtran}
\usepackage{cite}

\ifCLASSINFOpdf
  \usepackage[pdftex]{graphicx}
  \graphicspath{{./figures/}}
\else
  \usepackage[dvips]{graphicx}
  \graphicspath{{./figures/}}
\fi

\usepackage{amssymb}
\usepackage{amsmath}
\usepackage{bm}

\usepackage{algorithmic}

\usepackage[colorlinks=false]{hyperref}
\hypersetup{
    pdftitle={Fast and Accurate Amplitude Demodulation of Wideband Signals},
    pdfauthor={Mantas Gabrielaitis},
    bookmarksopen=true,         
    colorlinks=true,            
    pdfpagemode=UseOutlines,
    allcolors=black
}

\usepackage{amsthm}
\theoremstyle{definition}

\theoremstyle{plain}
\newtheorem{proposition}{Proposition}[section]
\theoremstyle{remark}

\newcommand{\iD}{i_1 \cdots i_D}
\newcommand{\iDc}{i_1 \hspace{-1pt}\cdot\hspace{-2pt}\cdot\hspace{-2pt}\cdot\hspace{-1pt} i_D}

\newcommand{\In}{\mathcal{I}_{n}}
\newcommand{\Inw}{\mathcal{I}_{n}^{\varpi_{}}}
\newcommand{\Inwm}{\mathcal{I}_{n}^{\omega_{}}}
\newcommand{\Sw}{\mathcal{S}_{\varpi}}
\newcommand{\Swm}{\mathcal{S}_{\omega}}
\newcommand{\Swmbar}{\mathcal{S}_{\scalebox{0.67}{$\bm{\omega}$}}}
\newcommand{\Mwm}{\mathcal{M}_{\omega}}
\newcommand{\Mwmbar}{\mathcal{M}_{\bm{\omega}}}
\newcommand{\Sgeqz}{\mathcal{S}_{\raisebox{0.9pt}{\scalebox{0.5}{$\geq$}\raisebox{-0.95pt}{\scalebox{0.65}{$\mathbf{0}$}}}}}
\newcommand{\Sgeqzbar}{\mathcal{S}_{\raisebox{0.9pt}{\scalebox{0.5}{$\geq$}\raisebox{-0.95pt}{\scalebox{0.61}{$\mathbf{\bar{0}}$}}}}}
\newcommand{\Sgeqs}{\mathcal{S}_{\raisebox{1pt}{\scalebox{0.5}{$\geq\hspace*{-3pt}|$}\scalebox{0.72}{$\mathbf{s}$}\scalebox{0.5}{$|$}}}}
\newcommand{\Sgeqsbar}{\mathcal{S}_{\raisebox{1pt}{\scalebox{0.5}{$\geq\hspace*{-3pt}|$}\raisebox{-0.5pt}{\scalebox{0.72}{$\mathbf{\bar{s}}$}}\scalebox{0.5}{$|$}}}}
\newcommand{\Sgeqc}{\mathcal{S}_{\raisebox{1pt}{\raisebox{0.65pt}{\scalebox{0.58}{$\geq\hspace*{-3pt}|$}}\scalebox{0.72}{$\mathbf{\hat{c}}$}\raisebox{0.65pt}{\scalebox{0.58}{$|$}}}}}
\newcommand{\Sgeqss}{\mathcal{S}_{\raisebox{1pt}{\scalebox{0.4}{$\geq\hspace*{-3pt}|$}\scalebox{0.52}{$\mathbf{s}$}\scalebox{0.4}{$|$}}}}
\newcommand{\Sleqo}{\mathcal{S}_{\raisebox{1pt}{\scalebox{0.5}{$|..|\hspace*{-3pt}\leq\hspace*{0pt}$}}\scalebox{0.67}{$\mathbf{1}$}}}
\newcommand{\Sleqobar}{\mathcal{S}_{\raisebox{1pt}{\scalebox{0.5}{$|..|\hspace*{-3pt}\leq\hspace*{0pt}$}}\scalebox{0.67}{$\mathbf{\bar{1}}$}}}
\newcommand{\Sd}{\mathcal{S}_{\raisebox{1.0pt}{\scalebox{0.45}{\{}} \raisebox{0.27pt}{\scalebox{0.55}{1}} \raisebox{1.0pt}{\scalebox{0.45}{\}}},d}}
\newcommand{\Sdbar}{\mathcal{S}_{\raisebox{1.0pt}{\scalebox{0.45}{\{}} \raisebox{0.27pt}{\scalebox{0.55}{1}} \raisebox{1.0pt}{\scalebox{0.45}{\}}},\scalebox{0.70}{$\mathbf{d}$}}}
\newcommand{\Cd}{\mathcal{C}_{d}}
\newcommand{\Cdbar}{\mathcal{C}_{\mathbf{d}}}
\newcommand{\PropRef}[1]{{\color{black}\emph{\hyperref[#1]{Proposition~}\ref{#1}}}}
\begin{document}

\title{Fast and Accurate Amplitude Demodulation\\of Wideband Signals}

\author{Mantas~Gabrielaitis\thanks{M.~Gabrielaitis is with the Institute of Science and Technology Austria, 3400 Klosterneuburg, Austria (e-mail: \href{mailto:mantas.gabrielaitis@ist.ac.at}{mantas.gabrielaitis@ist.ac.at})}
}

\markboth{ACCEPTED FOR PUBLICATION IN IEEE TRANSACTIONS ON SIGNAL PROCESSING}%
{Gabrielaitis: Fast and Accurate Amplitude Demodulation of Wideband Signals}

\maketitle

\begin{abstract}
Amplitude demodulation is a classical operation used in signal processing. For a long time, its effective applications in practice have been limited to narrowband signals. In this work, we generalize amplitude demodulation to wideband signals. We pose demodulation as a recovery problem of an oversampled corrupted signal and introduce special iterative schemes belonging to the family of alternating projection algorithms to solve it. Sensibly chosen structural assumptions on the demodulation outputs allow us to reveal the high inferential accuracy of the method over a rich set of relevant signals. This new approach surpasses current state-of-the-art demodulation techniques apt to wideband signals in computational efficiency by up to many orders of magnitude with no sacrifice in quality. Such performance opens the door for applications of the amplitude demodulation procedure in new contexts. In particular, the new method makes online and large-scale offline data processing feasible, including the calculation of modulator-carrier pairs in higher dimensions and poor sampling conditions, independent of the signal bandwidth. We illustrate the utility and specifics of applications of the new method in practice by using natural speech and synthetic signals.
\end{abstract}



\section{Introduction}

\IEEEPARstart{A}{mplitude} demodulation refers to the decomposition of a signal into a product of a slow-varying modulator-envelope and a fast-varying carrier. First introduced in radio communications \cite{Vakman1998}, this procedure has found applications in data acquisition and processing related to a broad range of phenomena. Automatic speech recognition \cite{Kingsbury1998}, atomic force microscopy \cite{Ruppert2017}, ultrasound imaging \cite{Wachinger2012}, brainwave \cite{Ktonas1980}, seismic trace \cite{Taner1979}, and fingerprint \cite{Larkin2001} analyses are a few among many examples to mention.

Originally, amplitude demodulation was intended for use with signals built of locally sinusoidal, i.e., narrowband, carri-ers. Several classical approaches excel in this setting, with Gabor's analytic-signal (AS) method being a long-standing champion \cite{Gabor1946, Vakman1996}. Nonetheless, many relevant problems inevitably require demodulating signals that feature wideband carriers, typically of (quasi)-harmonic, (quasi)-random, or spike-train origin \cite{Wilson1991, Smith2002, Goswami2019, Lin2001, Duck2002, Gottlieb1970, Felblinger1997, Gill2005, Liu2016} (see Suppl.\,Mat.\,H for an overview). When applied to them, the classical techniques fail, misleadingly mixing the carrier and modulator information \cite{Sell2010, Sell2010b}.

For a long time, no consistent and accurate way of demodulating wideband signals was known. Typically, a proxy of the modulator would be obtained by rectifying and then low-pass filtering the signal. Different implementations of this procedure, each adapted for a specific signal class, were suggested (see, e.g., \cite{Libbey1994, Platt1998, Gill2005}). The estimates of signal modulators obtained in this way, however, are neither accurate nor consistent between different methods. The carriers and modulators are not appropriately separated either, i.e., they can be demodulated further by iterating the same procedure \cite{Turner2010}. Moreover, the carrier estimates are often unbounded, even in well-defined situations (see, e.g., \cite[Fig.\,3.1]{Turner2010}).

Recently, two promising demodulation approaches suitable to signals with arbitrary bandwidths have been formulated. Turner and Sahani shaped demodulation into a statistical inference problem \cite{Turner2010, Turner2011}. In this so-called probabilistic amplitude demodulation (PAD) approach, the modulator and carrier are inferred from the signal as latent variables of an appropriately selected statistical model. Mathematically, PAD defines a maximization of a posteriori probability, a high-dimensional nonlinear optimization task. In another work, Sell and Slaney chose a deterministic route to demodulation \cite{Sell2010}. In their linear-domain convex (LDC) approach, the modulator is described as a minimum-power signal with penalized high-frequency terms lying above the original waveform. This problem is convex and thus amenable to more efficient optimization methods than the PAD.

The PAD and LDC techniques separate the modulator and carrier information of various synthetic wideband signals with a high degree of accuracy \cite{Turner2010, Sell2010}. The principal weakness of these approaches is a huge associated computational burden, which impedes their use in practical situations (see Section~\ref{sec:Performance} for the performance evaluations). In particular, online or large-scale offline signal processing is out of reach for the PAD and LDC demodulations. Besides, derivations of these methods are guided more by high-level modulator or carrier properties and computational tractability rather than strict recovery conditions. Hence, the boundaries of their validity in the context of real-world signals are somewhat blurred.

In this work, we frame demodulation as a problem of modulator recovery from an unlabeled mix of its true and corrupted sample points. We show that, under some loose constraints on carriers and modulators, high-accuracy demodulation can be achieved through exact or approximate norm minimization. We introduce different versions of custom-made alternating projection algorithms and test them in numerical experiments to solve this task. The new approach is shown to be free of the performance limitations inherent to the PAD and LDC methods. In particular, it combines the computational economy of the classical AS technique with the capacity to recover a wide range of arbitrary-bandwidth signals. We reveal the power of the new approach in terms of efficiency, accuracy, consistency, and robustness to corrupted data through theoretical analysis and illustrate it using synthetic signals with known structure. The use of the new method in realistic online and offline settings is demonstrated by applying it to natural speech.

\section{Mathematical Formulation of the Problem \label{sec:Problem}}

In what follows, we assume the representation of a real-valued signal $s(t)$ formed by a finite collection of its values uniformly sampled over a limited time interval: $s_i \equiv s(t_i)$, $i \in \In = \{ 1,2,\ldots,n \}$. Thus, a realization of the signal, \mbox{$\mathbf{s} \equiv (s_1, s_2, \ldots, s_n)^T$}, is an element of an $n$-dimensional Euclidean space $\mathbb{R}^n$, i.e., a linear space equipped with the inner product $\langle \mathbf{s}^{(1)}, \mathbf{s}^{(2)}\rangle = \sum_{i=1}^n (s^{(1)}_i \cdot s^{(2)}_i)$, which induces the Euclidean norm $\| \mathbf{s} \|_2 = \sqrt{\langle \mathbf{s} , \mathbf{s} \rangle}$. We use modulo $n$ arithmetic for indexes of vector components in this work.

\subsection{Demodulation constraints \label{subsec:Problem-Set-Theo}}

The task of demodulation is to factorize a signal $\mathbf{s} \in \mathbb{R}^n$ into a modulator $\mathbf{m} \in \mathbb{R}^n$ and a carrier $\mathbf{c} \in \mathbb{R}^n$:
\begin{equation}
\mathbf{s} = \mathbf{m} \circ \mathbf{c},\label{eq:ProbForm1}
\end{equation}
where symbol $\circ$ denotes an elementwise product of two vectors. There exists an uncountable number of pairs of $\mathbf{m}$ and $\mathbf{c}$ that satisfy \eqref{eq:ProbForm1}. Thus, further constraints are needed to define its unique solution. It is precisely these constraints that give a distinct character to different demodulation methods and set the domain of their validity \cite{Vakman1996, Loughlin1996, Turner2010, Sell2010}.

In this work, we introduce the extra demodulation restrictions by imposing some general assumptions on $\mathbf{m}$ and $\mathbf{c}$.

We define feasible modulators as elements of a convex set
\begin{equation}
\Mwm = \Sgeqz \cap \Swm,\label{eq:ProbForm2_}
\end{equation}
where
\begin{align}
\Sgeqz &= \{\mathbf{x} \in \mathbb{R}^n: x_i \geq 0, \, i \in \In \}, \nonumber\\
\Swm &= \{\mathbf{x} \in \mathbb{R}^n: (\mathbf{F}\mathbf{x})_i=0, \, i \in (\In \setminus \Inwm) \}, \label{eq:ProbForm3_}\\
\Inwm &= \{i \in \In: i \leq \omega \} \cup \{i \in \In: i > n+1-\omega \}. \nonumber
\end{align}
In \eqref{eq:ProbForm3_}, $\mathbf{F}$ denotes the operator of the unitary discrete Fourier transform (DFT), and $(\ldots)_i$ marks the $i$-th component of the argument vector. Hence, in our framework, modulators are nonnegative low-pass signals whose rate of variation is limited by the cutoff frequency $\omega$ (with $1 \leq \omega \leq \lceil n/2 \rceil$), which parametrizes $\Mwm$. This is a formal definition of the classical modulator-envelope \cite{Vakman1998,Cohen1999}.

We declare feasible carriers as elements of a nonconvex set
\begin{equation}
\Cd = \Sleqo \cap \Sd,\label{eq:ProbForm4_}
\end{equation}
where
\begin{equation}
\begin{aligned}
\Sleqo &= \{\mathbf{x} \in \mathbb{R}^n: |x_i| \leq 1, \, i \in \In \}, \\
\Sd &= \big\{\mathbf{x} \in \mathbb{R}^n: {\textstyle(\forall i)\sum_{j=i}^{i+d-1} I_{\{1\}}(|x_j|) \geq 1,} \\ &\hspace{59pt}{\textstyle(\exists i)\sum_{j=i}^{i+d-1} I_{\{1\}}(|x_j|) = 1} \big\},
\end{aligned}
\label{eq:ProbForm5_}
\end{equation}
with $I_{\{1\}}$ being the indicator function of the singleton $\{1\}$. The set $\Sleqo$ implies the boundedness of $\mathbf{c}$ between $-1$ and $1$. This restriction follows from the standard notion that the time-dependent amplitude of an amplitude-modulated $\mathbf{s}$ is purely set by $\mathbf{m}$. Meanwhile, $\Sd$ fixes to $d$ the maximum gap between any two neighboring components of $\mathbf{c}$ whose absolute values are equal to $1$.\footnote{The requirement of the existence of at least one gap of length $d$ in the definition of $\Sd$ assures that $\mathcal{C}_{d_1} \hspace{-2.5pt} \cap \mathcal{C}_{d_2} \hspace{-1pt} = \hspace{-1pt} \emptyset$ if $d_1 \hspace{-1pt} \neq \hspace{-1pt} d_2$. Such parametrization of the carrier set allows specifying more definite demodulation conditions.} As shown next, this constraint allows formulating extensive demodulation guarantees while only moderately affecting the scope of relevant carriers. Bandwidth-wise, $\Cd$ covers the whole range, from zero (sinusoidal) to flat (random spike) bandwidth signals, and defines the qualifier ``wideband'' used in this work. Note that the bandwidth \mbox{of $\mathbf{c} \in \Cd$} is mostly determined not by $d$ but by the arrangement of the $|c_i|=1$ and other sample points.\footnote{For example, even $\mathcal{C}_1$, which features the most limited repertoire among all $\Cd$, has zero-bandwidth elements (consider the $\mathbf{c}$ with $c_i=(-1)^i$) and elements with approximately flat amplitude spectra (consider a $\mathbf{c}$ with $c_i$ randomly chosen from $\{-1,1\}$).} Instead, as we see next, $d$ decides whether a chosen $\mathbf{c} \in \Cd$ can be restored after modulation.

\subsection{Demodulation as modulator recovery \label{subsec:Problem-Formulation-B}}

Note that, assuming $\mathbf{c} \in \Cd$, $|\mathbf{s}|$ can be seen as a corrupted version of $\mathbf{m}$: $|s_i|=m_i$ when $|c_i|=1$, and $|s_i| \neq m_i$ otherwise.  Further, if $\mathbf{m}$ can be found from $|\mathbf{s}|$, $\mathbf{c}$ follows from \eqref{eq:ProbForm1} uniquely ($c_i = s_i / m_i$), except the sample points with $m_i=0$. The latter, if any, are sparse and can be typically interpolated from the neighboring points. Hence, in our case, demodulation is virtually a problem of reconstructing $\mathbf{m}$  from a mix of its \textit{true} ($i:|s_i|=m_i$) and \textit{corrupted} ($i:|s_i| \neq m_i$) sample points when the class of each point is unknown. This viewpoint is at the core of the developments that follow next.

\subsection{Modulator recovery through norm minimization \label{subsec:Problem-Formulation-C}}

Our approach to demodulation builds around the estimator
\begin{equation}
\mathbf{\hat{m}} = \underset{\mathbf{x} \in \Sgeqss \cap \Sw}{\arg\min} \| \mathbf{x} \|_2,
\label{eq:Recovery2_1_}
\end{equation}
where $\Sgeqs = \{\mathbf{x}\in\mathbb{R}^n: x_i \geq |s_i|, i \in \mathcal{I}_n \}$. Note that $\mathbf{m} \in $ $\Sgeqs \cap \Sw$ if $\varpi \geq \omega$. The restriction corresponding to $\Sgeqs$ assures that $\mathbf{x}$ does not fall below $\mathbf{m}$ at the true sample points, i.e., points where $m_i=|s_i|$. If, besides, the true sample points are spread densely enough, we expect the norm minimization to enforce $\hat{m}_i=m_i$ at these points. But then, $\mathbf{\hat{m}} = \mathbf{m}$ by the discrete sampling theorem. The foundation for this intuitive consideration is laid by the following results (see Suppl.\,Mat.\,B for the proofs).
\begin{proposition}
For almost every $\mathbf{m} \in \Mwm$, $\mathbf{\hat{m}} = \mathbf{m}$ only if $\varpi \geq \omega$, and \mbox{$\mathbf{c} \in \Cd$} with $n_s \equiv \sum_{i=1}^nI_{\{1\}}(|c_i|)\geq \varpi+\omega-1 \implies d \leq n-(\varpi+\omega-2)$.\footnote{In fact, as follows from the proof of this proposition in Suppl.\,Mat.\,B, the condition that $\mathbf{c} \in \Cd$ for at least some $d$ is necessary for strictly every $\mathbf{m}$.}
\label{prop:Recovery1_}
\end{proposition}
\begin{proposition}
Consider $\mathbf{m} \in \Mwm$ and $\mathbf{\tilde{c}} \in \mathcal{C}_{\tilde{d}}$ with $|\tilde{c}_i|=1$ for $i \in \mathcal{J}_n \subseteq \mathcal{I}_n$, and $\tilde{c}_i=0$ otherwise. If $\mathbf{\hat{m}} = \mathbf{m}$ holds for the $\mathbf{m}$ and $\mathbf{\tilde{c}}$, then it also holds for every pair made of the same $\mathbf{m}$ and any $\mathbf{c} \in \Cd$ with $d \leq \tilde{d}$ and $|c_i|=1$ for $i \in \mathcal{J}_n$.
\label{prop:Recovery4_}
\end{proposition}
\begin{proposition}
Assume $\mathbf{m} \in \Mwm$ and $\mathbf{c} \in \Cd$ with $\varpi \geq \omega$. If, additionally, there exist $d \in \mathcal{I}_n$ and $i \in \mathcal{I}_{d}$ such that $n_s \equiv (n/d) \in \mathbb{N}_+$, $n_s \geq \varpi+\omega-1$, and $|c_{i+(j-1)\cdot d}|=1$ for every $j\in \mathcal{I}_{n_s}$, then $\mathbf{\hat{m}} = \mathbf{m}$.
\label{prop:Recovery2_}
\end{proposition}

\PropRef{prop:Recovery1_} reveals the tight match of $\mathbf{\hat{m}}$ to $\Cd$: no $\mathbf{m} \in \Mwm$ can be inferred from $\mathbf{s}$ by $\mathbf{\hat{m}}$ precisely if $\mathbf{c} \notin \Cd$. It also establishes the central role of the presence of true sample points in the recovery: for almost every $\mathbf{m} \in \Mwm$, at least the number $\varpi + \omega-1$ of such points is needed. \PropRef{prop:Recovery4_} further consolidates the latter view by stating that the success of the exact recovery of an $\mathbf{m} \in \Mwm$ via $\mathbf{\hat{m}}$ is fully determined by the number and positions of the true sample points. In particular, if exact demodulation is possible for some $\mathbf{\tilde{c}}$ with $\tilde{c}_i \in \{0,1\}$, then it is possible for any $\mathbf{c}$ with $|c_i|=1$ at $i \in \{j:\tilde{c}_j=1\}$ independent of other sample points.

In \PropRef{prop:Recovery1_}, $\varpi \geq \omega$ and $n_s \geq \varpi + \omega - 1$ imply $n_s \geq 2\omega-1$, which is a sufficient condition for $\mathbf{m}$ recovery in the classical setup when all true sample points are known (see the remark below the proof of \textit{Proposition~A.1} in Suppl.\,Mat.\,A). Hence, the data corruption manifesting in our problem necessitates further constraints on the number or positions of true sample points. In particular, \PropRef{prop:Recovery2_} certifies a full recovery of $\mathbf{m}$ if $\varpi \geq \omega$, and there exists a (not necessarily known) subset of at least $\varpi+\omega-1$ regularly-spaced true sample points. The latter condition covers a wide range of practically relevant carriers, including: (1)~the classical $\sin(2\pi\nu\mathbf{t}+\phi)$ with $\nu \geq \omega$, (2)~harmonic signals, (3)~regular spike-trains of $|c_i|=1$. More generally, any (non)stationary time-series with regularly placed $|c_i|=1$ regardless of the remaining points are eligible.

In addition to the regularity of true sample points, \PropRef{prop:Recovery2_} requires $n/d$ to be an integer. Nevertheless, numerical experiments reveal that both of these conditions can be ignored without practically relevant consequences (see Suppl.\,Mat.\,C and Fig.\,10 there). In particular, we found that the discrepancy between $\mathbf{m}$ and $\mathbf{\hat{m}}$ is vanishing with an overwhelming probability for any $\mathbf{c} \in \Cd$ if $\varpi \geq \omega$, and $\lceil n/d \rceil \geq 2\varpi-1$. This result noticeably extends the scope of recovery conditions over the domain of practically relevant (quasi-)regular and stochastic carriers. Among the examples are nonstationary sinusoidal and harmonic signals and arbitrary spike-trains with the distance between neighboring spikes at or below $d$ points. Note that $n_s \geq \lceil n/d \rceil$ by the definition of $\Cd$. Hence the relaxation of the strict regularity condition on the $|c_i|=1$ sample points comes at the expense of a slightly tighter constraint on $n_s$ necessary for exact recovery of $\mathbf{m}$: compare $n_s \geq 2\varpi - 1$ vs. $n_s \geq \varpi + \omega -1$.\footnote{This statement is exact and is established as an intermediate result in the proof of \PropRef{prop:Recovery1_}.}

Another important generalization of the recovery conditions comes with the following inequality:
\begin{proposition}
Consider $\mathbf{m} \in \Mwm$ and $\mathbf{c} \in \Sleqo$. Take $n_s \geq 2\varpi-1$ sample points of $\mathbf{s} = \mathbf{m} \circ\mathbf{c}$ whose indexes are defined as entries of any chosen $\mathbf{r} \in \mathbb{N}_+^{n_s}$ with $r_{i+1}-r_{i} = n/n_s$ for every $i \in \mathcal{I}_{n_s}$. Then,
\label{prop:Recovery3_}
\begin{equation}
\textstyle
\|\mathbf{m}-\mathbf{\hat{m}}\|_2 / \|\mathbf{m}\|_2 \leq \sqrt{1-\sum_{i=1}^{n_s}s_{r_i}^2 / \sum_{i=1}^{n_s}m_{r_i}^2}.
\label{eq:Recovery3_1_}
\end{equation}
\end{proposition}
\noindent Hence, if one can find a sequence of at least $2\varpi-1$ regularly-spaced sample points with $|s_i|$ sufficiently close to $m_i$, then the relative recovery error is close to 0 in terms of \eqref{eq:Recovery3_1_}. This result endows $\mathbf{\hat{m}}$ with the stability to discrepancies from the recovery conditions discussed earlier. At the same time, it provides approximate recovery guarantees for a wider range of stochastic and (quasi-)regular carriers besides  those with fairly densely packed $|c_i|=1$ sample points. Due to the low-pass restriction on $\mathbf{m}$, \eqref{eq:Recovery3_1_} is expected to hold approximately for an irregular $\mathbf{r} \in \mathbb{N}_+^{n_s}$ with $r_{i+1}-r_i \leq \lceil n/n_s \rceil$ as well.

We finally note that, whereas $\omega$ and $d$ are fixed properties of $\mathbf{m}$ and $\mathbf{c}$, $\varpi$ is a control parameter that must be specified. An appropriate $\varpi$, which satisfies the recovery conditions formulated above, can only be selected by using prior knowledge on $\mathbf{m}$ and $\mathbf{c}$ or found in a supervised learning setup.

\subsection{Relaxation of the exact minimum-norm requirement \label{subsec:Problem-Formulation-D}}

The norm-minimizing property of $\mathbf{\hat{m}}$ in \eqref{eq:Recovery2_1_} is critical in formulating sharp recovery conditions. However, from a practical point of view, little would be lost if another estimator $\mathbf{\hat{m}}$ with only slightly larger than the minimum norm among all elements of $\Sgeqs \cap \Sw$ is used. Thus, we relax \eqref{eq:Recovery2_1_} to 
\begin{equation}
\begin{aligned}
\text{find}& \qquad \mathbf{\hat{m}} \in \Sgeqs \cap \Sw \\
\text{subject to}& \qquad \|\mathbf{\hat{m}} \|_2 \simeq \underset{\mathbf{x} \in \Sgeqss \cap \Sw}{\arg\min} \| \mathbf{x} \|_2
\end{aligned}
\label{eq:Recovery4_1_}
\end{equation}
To specify the otherwise ambiguous relation operator $\simeq$, we request that $\mathbf{\hat{m}}$ obtained through \eqref{eq:Recovery4_1_} recovers $\mathbf{m}$ exactly, i.e., is norm-minimizing, for sinusoidal, harmonic, and spike-train carriers covered by \PropRef{prop:Recovery2_}. As we see later, this restriction regularizes the numerical algorithms formulated in the present work for sufficiently accurate demodulation well beyond those three classes of $\mathbf{c}$. The advantage brought by the approximation is computational efficiency.

\subsection{Method of solution \label{subsec:Problem-Method}}

The algorithms that we introduce to solve \eqref{eq:Recovery2_1_} and \eqref{eq:Recovery4_1_} in this work fall in the domain of the so-called methods of alternating projections (APs). The defining feature of each AP method is an iterative calculation of a feasible point ($\mathbf{\hat{m}} \in \Sgeqs \cap \Sw$ in our case) via alternating metric projections of its current estimate onto the separate constraint sets. Initially proposed by von Neumann for two closed subspaces \cite{Neumann1951}, this approach was later extended to arbitrary closed convex sets of a Hilbert space (see \cite{Escalante2011} for a review). Various implementations of the AP algorithms exist, featuring different domains of application, rate of convergence, and additional requirements satisfied by the solutions \cite{Bauschke1996, Escalante2011}.

We provide a rigorous mathematical basis on which the AP algorithms for solving the demodulation problem rely in Suppl.\,Mat.\,D,\,E,\,F. For a practical comprehension of the material that follows next, it is sufficient to know that:
\begin{itemize}
\item The sets $\Sgeqs$ and $\Sw$ are closed and convex.
\item A metric projection, or simply a projection henceforth, of $\mathbf{z} \in \mathbb{R}^n$ onto a closed convex set $\mathcal{S} \subset \mathbb{R}^n$ is a unique $\mathbf{x_z} \hspace{-1.75pt} \in \hspace{-1.65pt}\mathcal{S}$ with the smallest distance, i.e., $\|\mathbf{x_z}\hspace{-1pt} - \hspace{-0.4pt} \mathbf{z}\|_2$, from $\mathbf{z}$.
\item The projections onto $\Sgeqs$ and $\Sw$ are respectively achieved by operators
\begin{align}
\mathbf{P}_{\Sgeqss}[\mathbf{z}] = |\mathbf{s}| + (\mathbf{z}-|\mathbf{s}|) \circ \theta (\mathbf{z}-|\mathbf{s}|)
\label{eq:MathPrel1x}
\end{align}
and
\begin{equation}
\mathbf{P}_{\Sw}[\mathbf{z}] = (\mathbf{F^{-1}} \, \mathbf{W}_\varpi \, \mathbf{F}) \, \mathbf{z}.\label{eq:MathPrel3x}
\end{equation}
Here, $\theta(\ldots)$ is the Heaviside step function evaluated elementwise. $\mathbf{W}_\varpi$ is a diagonal matrix such that $(W_{\varpi})_{ii} = 1$ if $i \in \Inw$, and $(W_{\varpi})_{ii} = 0$ otherwise.
\end{itemize}

To emphasize the nature of the underlying numerical algorithms, we name our new approach as AP demodulation.

\subsection{Relation to other problems and approaches}

\begin{figure*}
\centering
\includegraphics[width=1\textwidth]{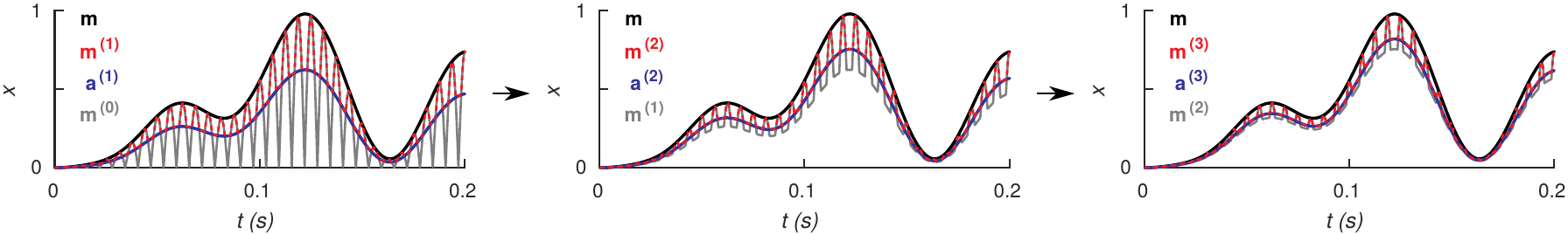}
\caption{The first three iterations of the AP-B algorithm applied to an amplitude-modulated sinusoidal signal. $\mathbf{m}$ stands for the real modulator.}
\label{fig:1}
\end{figure*}

Demodulation is a counterpart of a widely known and studied problem of blind deconvolution: $\mathbf{s} = \mathbf{m} \, \raisebox{0.7pt}{{$\scriptstyle\circledast$}} \, \mathbf{c}$. 
Indeed, both tasks admit the algebraic form of each other in the Fourier domain. Nevertheless, the properties of $\mathbf{m}$ and $\mathbf{c}$ inherent to practical instantiations of amplitude demodulation and blind deconvolution differ significantly. These differences predetermine the need for distinctive strategies to solve the respective tasks, as discussed next.

One of the most powerful convex-programming-based deconvolution approaches, introduced in \cite{Ahmed2014}, builds on the assumption that $\mathbf{m}$ and $\mathbf{c}$ belong to known low-dimensional subspaces. There, recovery of $\mathbf{m}$ and $\mathbf{c}$ is achieved by minimiz-ing the nuclear, atomic, $\ell_1$, or $\ell_{2,1}$ norms of their outer product (in the subspace representation) subject to linear measurement constraints of $\mathbf{s}$ \cite{Ahmed2014, Ling2015, Chi2016, Xie2019}. This scheme successfully solves many practically relevant blind deconvolution cases, such as image deblurring, multipath channel protection, and super-resolution microscopy \cite{Ahmed2014, Xie2019}. However, the low-dimension subspace assumption, a crucial prerequisite of the approach, renders it inapt to deal with realistic carriers in the amplitude demodulation context. Indeed, even a sinusoidal carrier with a fluctuating phase is hardly representable in this frame, not to mention more complex wideband signals met in practice. Moreover, the subspace model of $\mathbf{m}$ and $\mathbf{c}$ does not allow enforcing the amplitude contents to $\mathbf{m}$ exclusively.

Deconvolution problems have also been approached by using AP-like methods \cite{Oppenheim1981, Trussell1984, Kundur1998, Yang1994}. A general strategy of the existing algorithms is to achieve deconvolution by an iterative refinement of both $\mathbf{m}$ and $\mathbf{c}$ upon the requirement of exact \cite{Oppenheim1981, Kundur1998} or approximate \cite{Trussell1984, Yang1994} adherence to the defining equality $\mathbf{s} = \mathbf{m} \, \raisebox{0.7pt}{{$\scriptstyle\circledast$}} \, \mathbf{c}$ and the support region, intensity range, and spectrum constraints implied on $\mathbf{m}$ and $\mathbf{c}$ or $\mathbf{r} = \mathbf{s} - \mathbf{m} \, \raisebox{0.7pt}{{$\scriptstyle\circledast$}} \, \mathbf{c}$. These methods differ significantly between themselves. Each of them achieves satisfactory recovery by a judicious combination of specific constraint sets and the iterative scheme adjusted to specific classes of $\mathbf{m}$ and $\mathbf{c}$. The nonconvexity of $\Cd$ and the absence of efficient explicit projections onto this set makes the application of the known deconvolution methods unsuitable to amplitude demodulation. None of the current AP-like deconvolution methods allow assigning the amplitude contents to $\mathbf{m}$ purely either.

We next note that our formulation of the amplitude demodulation problem in Section~\ref{subsec:Problem-Formulation-B} reveals it as a generalization of the classical task of band-limited signal recovery from true sample points. An AP method known under the name Papoulis-Gerchberg and its variants were successfully applied in the latter setting (see \cite{Ferreira2001} for a review). The differences in the available information on the recoverable signal lead to distinct strategies in algorithmic approaches to these two problems. In particular, the Papoulis-Gerchberg methods rely entirely on known true data. Thus, they are impossible \mbox{to use for} demodulation purposes. The AP algorithms introduced in the present work can be applied in the classical setting. However, not using the available information about the true data makes them inferior to their classical counterparts unless the sample points are fairly uniformly spread, as discussed in Section~\ref{subsec:Problem-Formulation-C}.

The approach suggested in the present work also has some parallels with the LDC demodulation method by \cite{Sell2010}. There, \eqref{eq:ProbForm1} is accompanied by a constraint on the modulator $\mathbf{m}$ expressed as the solution of the quadratic programming problem
\begin{equation}
\begin{aligned}
\text{minimize}& \qquad \| \mathbf{w} \circ \mathbf{F} \mathbf{m} \|_2^2 + \|\mathbf{m}\|_2^2 \\
\text{subject to}& \qquad |s_i| \leq m_i \leq \max[\mathbf{s}] \qquad \forall i \in \In,
\end{aligned}
\label{eq:ProbForm4}
\end{equation}
where $\mathbf{w}$ denotes the weighting vector. \eqref{eq:ProbForm4} was introduced heuristically, trying to quantify the intuitive notion of the modulator-envelope as a signal wrapping $\mathbf{s}$ from above.

Practical applications suggest the LDC method defined by \eqref{eq:ProbForm1} and \eqref{eq:ProbForm4} being computationally most efficient and precise among all current techniques designed for demodulating signals unreachable to classical algorithms \cite{Turner2011, Sell2010}. Thus, we use it as a reference when evaluating the performance of the newly-formulated approach of the present work.

\section{Demodulation Algorithms \label{sec:Algorithms}}

In this section, we formulate three algorithms representing the core arsenal of the AP approach to demodulation. Simplicity, efficiency, and estimation accuracy of the algorithms are the main aspects under consideration. We refer the reader to Suppl.\,Mat.\,F for proofs of all propositions found here.

\subsection{AP-Basic \label{subsec:Algorithms1}}

We start with the simplest possible AP algorithm, therefore named ``AP-Basic'' (AP-B).


\vspace{6pt}
\hrule
\vspace{2pt}
\noindent\textbf{Algorithm:} AP-Basic (AP-B)
\vspace{1pt}
\hrule
\vspace{2.7pt}
\begin{algorithmic}[1]
\STATE \textbf{Set:} $N_{iter}$, $\epsilon_{tol}$
\vspace{0pt}
\STATE \textbf{Initialize:} $i=0$, $\epsilon^{(0)}=\|\mathbf{s}\|_2/\sqrt{n}$, $\mathbf{m}^{(0)} = |\mathbf{s}|$, $\mathbf{a}^{(0)} = \mathbf{0}$
\vspace{0pt}
\WHILE{$\epsilon^{(i)} > \epsilon_{tol}$ \AND $i<N_{iter}$}
\STATE $i = i + 1$
\vspace{0.92pt}
\STATE $\mathbf{a}^{(i)} = \mathbf{P}_{\Sw}[\mathbf{m}^{(i-1)}]$
\vspace{0.92pt}
\STATE $\mathbf{m}^{(i)} = \mathbf{P}_{\Sgeqss}[\mathbf{a}^{(i)}]$
\vspace{0.92pt}
\STATE $\epsilon^{(i)} = \| \mathbf{m}^{(i)} - \mathbf{a}^{(i)}\|_2 / \sqrt{n}$
\ENDWHILE
\vspace{0pt}
\STATE \textbf{Finalize:} $\mathbf{\hat{m}}=\mathbf{m}^{(i)}$
\vspace{2.9pt}
\hrule
\end{algorithmic}
\vspace{7pt}
Here, $N_{iter}$ stands for the maximum number of algorithm iterations. $\epsilon^{(i)}$ is the infeasibility error at the $i$-th iteration, which is used to control the termination of the algorithm. Specifically, $\epsilon^{(i)}$ measures the distance of the modulator estimate $\mathbf{m}^{(i)} \in \Sgeqs$ from $\Sw$ and sets a lower bound on the convergence error: $\epsilon^{(i)} \leq \| \mathbf{m}^{(i-1)} - \mathbf{m}^\dagger \|_2 / \sqrt{n}$ (see Suppl.\,Mat.\,G). The iterative process is stopped when $\epsilon^{(i)}$ drops to the level of a predetermined threshold $\epsilon_{tol} > 0$ or below. $\epsilon_{tol} \leq 0$ would force the completion of all $N_{iter}$ iterations of the algorithm. $\mathbf{\hat{m}}$ denotes the final estimate of the modulator. $\mathbf{\hat{m}}$ arbitrarily close to $\Sgeqs \cap \Sw$ can be reached if $N_{iter}$ is sufficiently large:


\begin{proposition}
A sequence $\mathbf{m}^{(0)},\mathbf{m}^{(1)}, \ldots, \mathbf{m}^{(i)},\ldots$ formed by the AP-B algorithm for $\epsilon_{tol} = 0$ and $N_{iter} \to +\infty$ converges to some $\mathbf{m}^\dagger \in \Sgeqs \cap \Sw$. The convergence is geometric and monotonic, i.e., there exist $\gamma>0$ and $0<r<1$ such that $\| \mathbf{m}^{(i)} - \mathbf{m}^\dagger \|_2 \leq \gamma \cdot r^i$ and $\| \mathbf{m}^{(i+1)} - \mathbf{m}^\dagger \|_2 \leq \| \mathbf{m}^{(i)} - \mathbf{m}^\dagger \|_2$ for $i \geq 0$.
\label{prop:APBSolConv}
\end{proposition}

It can be shown by example that the AP-B does not always provide minimum-norm estimators $\mathbf{\hat{m}}$. However, it is expected to do so at least approximately if some conditions are met. We clarify this next with the help of Fig.\,\ref{fig:1}, which displays the first three iterations of the AP-B applied to an example signal.

First, note that the starting point $\mathbf{m}^{(0)}$ is elementwise not-higher than the real modulator $\mathbf{m}$ (black). $\mathbf{P}_{\Sw}$ maps $\mathbf{m}^{(0)}$ to $\mathbf{a}^{(1)}$, which, by definition of a metric projection, is its best mean-squared-error (MSE) approximation in $\Sw$ (blue). By the definition of $\Sw$, $\mathbf{a}^{(1)}$ is nearly constant over time windows shorter than $n/(2\pi\varpi)$ points. In general, the best constant MSE estimator of a sample of numbers is its average. Thus, as the best MSE estimator of $\mathbf{m}^{(0)}$ over $\Sw$, $\mathbf{a}^{(1)}$ approximates the local average of $\mathbf{m}^{(0)}$ values in a window of $\approx n/(2\pi\varpi)$ points at every moment. If $\varpi \geq \omega$, $\mathbf{c} \in \Sleqo$, and $\approx n/(2\pi\varpi)$ sample points are sufficient to average out the local variations of $\mathbf{c}$, $\mathbf{a}^{(1)}$ is supposed to be proportional to $\mathbf{m}$, at least roughly. The first iteration is completed by the projection of $\mathbf{a}^{(1)}$ back onto $\Sgeqs$ to obtain $\mathbf{m}^{(1)}$ (red).

Applying the same reasoning as above, we deduce that, with each iteration, $\mathbf{a}^{(i)}$, and thus $\mathbf{m}^{(i)}$, approaches $\mathbf{m}$ elementwise (see Fig.\,\ref{fig:1}). In general, $\mathbf{m}^{(i)}$ may exceed the level of the real modulator $\mathbf{m}$ over time windows longer than $\geq n/(2\pi\varpi)$ points for higher $i$ before $\mathbf{m}^\dagger \in \Sgeqs \cap \Sw$ is reached. However, as follows from the considerations of the previous paragraph, such segments of $\mathbf{m}^{(i)}$ would be approximately compatible with $\Sgeqs \cap \Sw$ and would not be considerably affected in subsequent iterations. Hence, $\mathbf{\hat{m}}$ obtained by the AP-B is expected to follow the true sample points of $\mathbf{m}$ tightly. If the number of these points is sufficient, then $\mathbf{\hat{m}} \simeq \mathbf{m}$ as well.

The basis for the above considerations is laid by the fact that they are exact for some important types of carriers:
\\

\begin{proposition}
Consider $\mathbf{m} \in \Mwm$ and $\mathbf{c} \in \Cd$ with $|c_j| = \sum_{k=1}^{n/\nu}(\tilde{c}_{\nu \cdot k} \cdot e^{\imath 2 \pi \nu (k-1) (j-1) / n})$, where $\tilde{c}_{\nu \cdot k} \in \mathbb{C}$ and $n / \nu \in \mathbb{N}$. If $\varpi \geq \omega$ and $\nu \geq \varpi + \omega - 1$, then  a sequence $\mathbf{m}^{(0)}, \mathbf{m}^{(1)}, \ldots, \mathbf{m}^{(i)},\ldots$ formed by the AP-B algorithm for $\epsilon_{tol}=0$ and $N_{iter} \to +\infty$ converges to $\mathbf{m}$.
\label{prop:APBSolConv_b_}

\end{proposition}

\noindent Among others, \PropRef{prop:APBSolConv_b_} encompasses the sinusoidal, harmonic, and regular spike-train carriers covered by \PropRef{prop:Recovery2_}. Thus, in these cases, AP-B satisfies the minimum-norm property, i.e., provides $\mathbf{\hat{m}}$ that converges to a solution of \eqref{eq:Recovery4_1_}. The condition $\nu \geq \varpi + \omega - 1$ in \PropRef{prop:APBSolConv_b_} plays the role of the inequality $n/d \geq \varpi + \omega - 1$ in \PropRef{prop:Recovery2_}.

\subsection{AP-Accelerated \label{subsec:Algorithms2}}

One of the potential weak points of AP algorithms based on pure alternating projections onto convex sets, like the \mbox{AP-B}, is relatively slow convergence \cite{Gubin1967, Youla1982, Franchetti1986}. Indeed, despite the geometric nature of the convergence, the actual number of iterations necessary to reach a specific error level may be arbitrarily large if the factor $r$ in $\| \mathbf{m}^{(i)} - \mathbf{m}^\dagger \|_2 \leq \gamma \cdot r^i$ is sufficiently close to 1. To address this issue, various accelerated AP schemes have been suggested for specific classes of the constraint sets \cite{Gubin1967, Gearhart1989, Bauschke2003}. Here, we propose a parameter-free accelerated version of the AP-B algorithm specifically designed for the demodulation problem. We refer to it as ``AP-Accelerated'' (\mbox{AP-A}).


\vspace{7pt}
\hrule
\vspace{2pt}
\noindent\textbf{Algorithm:} AP-Accelerated (AP-A)
\vspace{1pt}
\hrule
\vspace{3pt}
\begin{algorithmic}[1]
\STATE \textbf{Set:} $N_{iter}$, $\epsilon_{tol}$
\STATE \textbf{Initialize:} $i=0$, $\epsilon^{(0)}=\|\mathbf{s}\|_2/\sqrt{n}$, $\mathbf{m}^{(0)} = |\mathbf{s}|$, $\mathbf{a}^{(0)} = \mathbf{0}$
\WHILE{$\epsilon^{(i)} > \epsilon_{tol}$ \AND $i<N_{iter}$}
\STATE $i=i+1$
\vspace{1pt}
\STATE $\mathbf{b}^{(i)} = \mathbf{P}_{\Sw}[\mathbf{m}^{(i-1)} - \mathbf{a}^{(i-1)}]$
\vspace{1pt}
\STATE $\lambda = \| \mathbf{m}^{(i-1)} - \mathbf{a}^{(i-1)} \|_2^2 / \|\mathbf{b}^{(i)} \|_2^2$
\vspace{1pt}
\STATE $\mathbf{a}^{(i)} = \mathbf{a}^{(i-1)} + \lambda \cdot \mathbf{b}^{(i)}$
\vspace{1pt}
\STATE $\mathbf{m}^{(i)} = \mathbf{P}_{\Sgeqss}[\mathbf{a}^{(i)}]$
\vspace{1pt}
\STATE $\epsilon^{(i)} = \| \mathbf{m}^{(i)} - \mathbf{a}^{(i)} \|_2/\sqrt{n}$
\ENDWHILE
\STATE \textbf{Finalize:} $\mathbf{\hat{m}}=\mathbf{m}^{(i)}$
\vspace{3pt}
\hrule
\end{algorithmic}
\vspace{7pt}

\begin{proposition}
\label{prop:APASolConv}
A sequence $\mathbf{m}^{(0)},\mathbf{m}^{(1)}, \ldots, \mathbf{m}^{(i)},\ldots$ formed by the AP-A algorithm for $\epsilon_{tol} = 0$ and $N_{iter} \to +\infty$ converges to some $\mathbf{m}^\dagger \in \Sgeqs \cap \Sw$. The convergence is monotonic, i.e., $\| \mathbf{m}^{(i+1)} - \mathbf{m}^\dagger \|_2 \leq \| \mathbf{m}^{(i)} - \mathbf{m}^\dagger \|_2$ for $i \geq 0$.
\end{proposition}

Note that $\lambda>1$ except when $\mathbf{P}_{\Sw}$ is the identity operator, i.e., the trivial case of $\mathbf{m}=|\mathbf{s}|$. Indeed, it follows from the definition of $\mathbf{P}_{\Sw}$ [see \eqref{eq:MathPrel3x}] and the unitary property of $\mathbf{F}$ that $\|\mathbf{m}^{(i-1)}-\mathbf{a}^{(i-1)} \|_2^2 > \| \mathbf{P}_{\Sw}[\mathbf{m}^{(i-1)}-\mathbf{a}^{(i-1)}] \|_2^2=\|\mathbf{b}^{(i)}\|_2^2$, if $\mathbf{P}_{\Sw}$ is not the identity operator. It is easy to see that $(\mathbf{a}^{(i)}-\mathbf{a}^{(i-1)})=\lambda \cdot (\mathbf{P}_{\Sw}[\mathbf{m}^{(i-1)}]-\mathbf{a}^{(i-1)})$ in the above algorithm. Moreover, if $\lambda$ is fixed to 1 by force, the AP-A and AP-B algorithms become identical. Therefore, the AP-A produces increments from $\mathbf{a}^{(i-1)}$ to $\mathbf{a}^{(i)}$ that are scaled up compared with those that were obtained by applying the AP-B algorithm for the same iterations. 

To understand the working principle of the AP-A better, recall that $\mathbf{P}_{\Sw}[\mathbf{m}^{(i)}]$, and thus $\mathbf{a}^{(i)}$, are nearly constant over time windows consisting of $<n/(2\pi\varpi)$ points (see Section~\ref{subsec:Algorithms1}). For a semiquantitative analysis, we can assume that this holds exactly. Let us denote a segment of $(\mathbf{m}^{(i-1)}-\mathbf{a}^{(i-1)})$ restricted to such a window by $\mathbf{z}$. Then, $\mathbf{b}^{(i)}$ defined in the same window is just $(l^{-1} \cdot \sum_{j=1}^l z_j) \cdot \mathbf{1}$, and $\|\mathbf{m}^{(i-1)} - \mathbf{a}^{(i-1)} \|_2^2$ corresponds to $\sum_{j=1}^l z_j^2$, where, $l = \lfloor n/(2\pi\varpi) \rfloor$. Consequently, $\lambda \cdot \mathbf{b}^{(i)}$, i.e., the increment from $\mathbf{a}^{(i-1)}$ to $\mathbf{a}^{(i)}$, is given by $\big(\sum_{j=1}^l{z_j^2} / \sum_{j=1}^l{z_j}\big) \cdot \mathbf{1}$. It follows from $\mathbf{m}^{(i-1)}=\mathbf{P}_{\Sgeqss}[\mathbf{a}^{(i-1)}]$ that $\mathbf{z}$ is elementwise nonnegative. Therefore, $\big(\sum_{j=1}^l{z_j^2} / \sum_{j=1}^l{z_j}\big) \leq \max[\mathbf{z}]$. However, $\max[\mathbf{z}]$ corresponds to the difference between the real modulator and $\mathbf{a}^{(i-1)}$ in the considered time window, at least approximately, if $\lceil n/d \rceil \geq 2\varpi-1$. Thus, while up-scaling $\mathbf{a}^{(i)}-\mathbf{a}^{(i-1)}$ at each iteration to accelerate the convergence, the AP-A also ensures that $\mathbf{a}^{(i)}$ stays approximately within the bounds of the real modulator $\mathbf{m}$. This property ensures that $\mathbf{\hat{m}}$ tightly follows $\mathbf{m}$ if the same conditions as required by the AP-B are met.

We further note that $\big(\sum_{j=1}^l{z_j^2} / \sum_{j=1}^l{z_j}\big) = \max[\mathbf{z}]$, i.e., $\mathbf{a}^{(i)}$ reaches $\mathbf{m}$ in a single iteration, if all but one element of $\mathbf{z}$ are equal to zero. Importantly, approximately this situation is faced in reality, as illustrated in Fig.\,\ref{fig:1}. Specifically, with increased $i$, $(\mathbf{m}^{(i-1)}-\mathbf{a}^{(i-1)})$ becomes mainly flat with only a few separate elements considerably above 0 over time windows shorter than $n/(2\pi\varpi)$ points. For comparison, the analogous increment from $\mathbf{a}^{(i-1)}$ to $\mathbf{a}^{(i)}$ is moderate and equals only $\max[\mathbf{z}]/l$ in the case of the AP-B method. These considerations explain the substantial speed-up provided by the AP-A algorithm in practice. They also reveal that any additional acceleration steps in the AP-A would result in overscaled $\mathbf{\hat{m}}$, hence reducing the demodulation accuracy. 

The AP-A algorithm repeats the AP-B in terms of exact recovery guarantees of \PropRef{prop:APBSolConv_b_}:
\begin{proposition}
Consider $\mathbf{m} \in \Mwm$ and $\mathbf{c} \in \Cd$ with $|c_j| = \sum_{k=1}^{n/\nu}(\tilde{c}_{\nu \cdot k} \cdot e^{\imath 2 \pi \nu (k-1) (j-1) / n})$, where $\tilde{c}_{\nu \cdot k} \in \mathbb{C}$ and $n / \nu \in \mathbb{N}$. If $\varpi \geq \omega$ and $\nu \geq \varpi + \omega - 1$, then  a sequence $\mathbf{m}^{(0)}, \mathbf{m}^{(1)}, \ldots, \mathbf{m}^{(i)},\ldots$ formed by the AP-A algorithm for $\epsilon_{tol}=0$ and $N_{iter} \to +\infty$ converges to $\mathbf{m}$.
\label{prop:APASolConv_b_}
\end{proposition}
\noindent This result substantiates the semiquantitative argumentation of the AP-A convergence properties provided above and establishes the respective $\mathbf{\hat{m}}$ as a numerical solution of \eqref{eq:Recovery4_1_}.

\begin{figure*}
\centering
\includegraphics[width=1\textwidth]{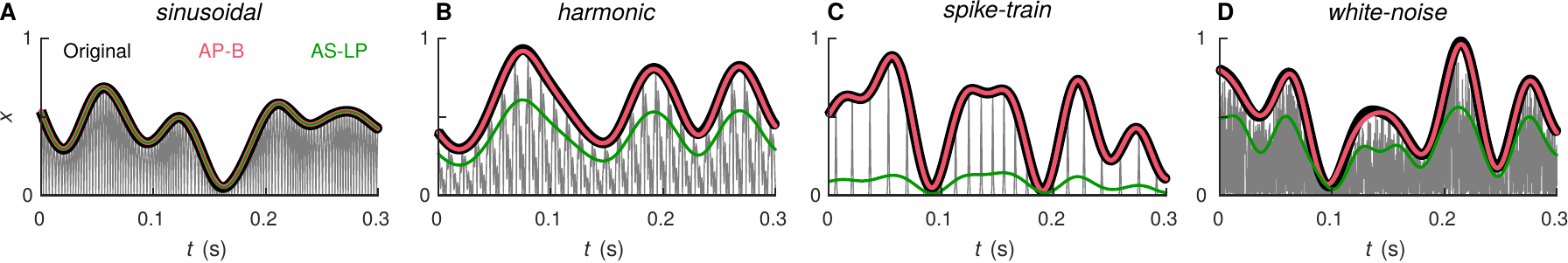}
\caption{Typical examples of the test signals (gray), featuring nonstationary sinusoidal (\textbf{A}), harmonic (\textbf{B}), spike-train (\textbf{C}), and stationary white-noise (\textbf{D}) carriers, and their modulators obtained by using the AP-B (red) and AS-LP (green) algorithms. The signals are represented by their absolute values here. The predefined modulators are shown in black.}
\label{fig:2}
\end{figure*}

\subsection{AP-Projected \label{subsec:Algorithms3}}

As argued above, the AP-A and AP-B algorithms produce modulator estimates that are expected to tightly follow the original $\mathbf{m}$ if the conditions analogous to those discussed in Section~\ref{subsec:Problem-Formulation-C} are met. These estimates, however, do not always hold the minimum-norm property \eqref{eq:Recovery2_1_}. A classical AP scheme that guarantees minimum-norm solutions is known under the name of Dykstra \cite{Dykstra1983, Boyle1986}. In particular, Dykstra's algorithm calculates the projection of a point in $\mathbb{R}^n$ onto the feasible set. Thus, by choosing an appropriate initial condition, the solution with a minimized norm can be obtained (see \PropRef{prop:APPSolConv} next and its proof in Suppl.\,Mat.\,F). We consider a version of this algorithm adapted to solve the demodulation problem and call it ``AP-Projected'' (AP-P).


\vspace{7pt}
\hrule
\vspace{2pt}
\noindent\textbf{Algorithm:} AP-Projected (AP-P)
\vspace{1pt}
\hrule
\vspace{3pt}
\begin{algorithmic}[1]
\STATE \textbf{Set:} $N_{iter}$, $\epsilon_{tol}$
\STATE \textbf{Initialize:} $i=0$, $\epsilon^{(0)}=\|\mathbf{s}\|_2/\sqrt{n}$, $\mathbf{m}^{(0)} = \mathbf{c}^{(0)} = |\mathbf{s}|$
\WHILE{$\epsilon^{(i)} > \epsilon_{tol}$ \AND $i<N_{iter}$}
\STATE$i=i+1$
\vspace{1pt}
\STATE$\mathbf{a}^{(i)} = \mathbf{P}_{\Sw}[\mathbf{m}^{(i-1)}]$
\vspace{1pt}
\STATE$\mathbf{m}^{(i)} = \mathbf{P}_{\Sgeqss}[\mathbf{a}^{(i)}-\mathbf{c}^{(i-1)}]$
\vspace{1pt}
\STATE$\mathbf{c}^{(i)} = \mathbf{m}^{(i)} - (\mathbf{a}^{(i)}-\mathbf{c}^{(i-1)})$
\vspace{1pt}
\STATE$\epsilon^{(i)} = \sqrt{(\| \mathbf{m}^{(i-1)}-\mathbf{a}^{(i)}\|_2^2 + \| \mathbf{m}^{(i)} - \mathbf{a}^{(i)}\|_2^2)/(2 \cdot n)}$
\ENDWHILE
\STATE \textbf{Finalize:} $\mathbf{\hat{m}}=\mathbf{m}^{(i)}$
\vspace{3pt}
\hrule
\end{algorithmic}
\vspace{7pt}

\begin{proposition}
A sequence $\mathbf{m}^{(0)},\mathbf{m}^{(1)}, \ldots, \mathbf{m}^{(i)},\ldots$ formed by the AP-P algorithm for $\epsilon_{tol}=0$ and $N_{iter} \to +\infty$ converges to a unique $\mathbf{m}^\dagger \in \Sgeqs \cap \Sw$ such that $\|\mathbf{m}^\dagger\|_2 \leq \|\mathbf{x}\|_2$ for every $\mathbf{x} \in \Sgeqs \cap \Sw$. The convergence is monotonic, i.e., $\| \mathbf{m}^{(i+1)} - \mathbf{m}^\dagger \|_2 \leq \| \mathbf{m}^{(i)} - \mathbf{m}^\dagger \|_2$ for $i \geq 0$.
\label{prop:APPSolConv}
\end{proposition}

The AP-P differs from the AP-B in that, before projecting a point onto $\Sgeqs$, an increment produced by the projection onto this set in the previous iteration is subtracted. This correction may cause the infeasibility error $\| \mathbf{m}^{(i)} - \mathbf{a}^{(i)}\|_2 / \sqrt{n}$ estimated after projecting onto $\Sgeqs$ to drop to zero intermittently before the final solution is reached, making it an inappropriate option as the stopping criterion. Hence, in contrast to the AP-B and AP-A algorithms, we defined the $\epsilon$ for the AP-P as a combination of the infeasibility errors evaluated after projecting onto both sets $\Sw$ and $\Sgeqs$ at each iteration (see line 8 above). This error measure is strictly positive and converges to zero when $N_{iter} \to +\infty$ \cite{Birgin2005}.

The understanding of the convergence rate of Dykstra's scheme is limited. It was shown that the convergence is geometric for an intersection of half-spaces \cite{Deutsch1994, Deutsch1995}. Nevertheless, no equivalent result exists for other convex sets. Moreover, it was demonstrated that the convergence rate of this algorithm may depend on the initial conditions and may be considerably slower than that of AP algorithms based on pure projections \cite{Bauschke1994}.

\subsection{Computational complexity \label{subsec:Algorithms4}}

Except for the projection operator $\mathbf{P}_{\Sw}$, each iteration of the three formulated AP algorithms relies on vector addition, scalar product, and value update. These are linear in the number of sample points. The operator $\mathbf{P}_{\Sw}$ can be easily implemented by using the direct and inverse fast Fourier transforms (FFTs) and setting the relevant elements of the signal to zero in the Fourier domain. The current state-of-the-art FFT algorithms have an $\mathcal{O}(n \log n)$ time complexity \cite{Duhamel1990}, which, thus, sets the overall time complexity of the AP algorithms introduced in this work. Our numerical experiments suggest that the convergence speed in terms of iteration number is independent of the signal length (see Suppl.\,Mat.\,M).

\section{Performance Tests\label{sec:Performance}}

To evaluate the AP algorithms introduced above, we compared their performance with the AS and LDC demodulation approaches when applied to infer the modulator of predefined synthetic test signals. The LDC approach was implemented by using two state-of-the-art quadratic programming solvers: Gurobi (v8.1.1) \cite{Gurobi2019} and OSQP (v0.6.0) \cite{Stellato2020}. The AS demodulation was achieved by using the FFT-based approach \cite{Marple1999}. In that case, we additionally low-pass filtered the obtained modulator estimate with $\mathbf{P}_{\Sw}$ to regularize it. We refer to this modified demodulation scheme as AS-LP.

\subsection{Test signals \label{subsec:Performance1}}

The test signals were composed as products of a modulator and a carrier: $\mathbf{s} = \mathbf{m} \circ \mathbf{c}$. Four types of $\mathbf{c}$, approximating basic building blocks of real-world signals, were used: nonstationary \textit{sinusoidal}, \textit{harmonic}, and \textit{spike-train}, as well as stationary \textit{white-noise} (see, respectively, (176), (180), (184), and (188) in Suppl.\,Mat.\,I). The former two were combined with modulators of nonstationary \textit{Gaussian} origin, while the latter two types of carriers were paired with the so-called \textit{maximally-uniformly distributed} modulators (see, respectively, (160)\,--\,(162) and (163)\,--\,(168) in Suppl.\,Mat.\,I).

In all cases, modulator and carrier pairs were selected to meet the core recoverability condition $\lceil n/d \rceil \geq 2\varpi-1$, at least approximately. The remaining parameters of $\mathbf{m}$ and $\mathbf{c}$ (see Suppl.\,Mat.\,I) were chosen to imitate realistic conditions as much as possible. For example, in all cases, signals were taken as segments of longer time series, and thus, were not $n$-periodic. The center frequencies of the sinusoidal and harmonic carriers were set so that only sample points with $|c_i| \approx 1$ rather than $|c_i| = 1$ were available.

\subsection{Performance evaluation \label{subsec:Performance12}}

\begin{figure*}
\centering
\includegraphics[width=1\textwidth]{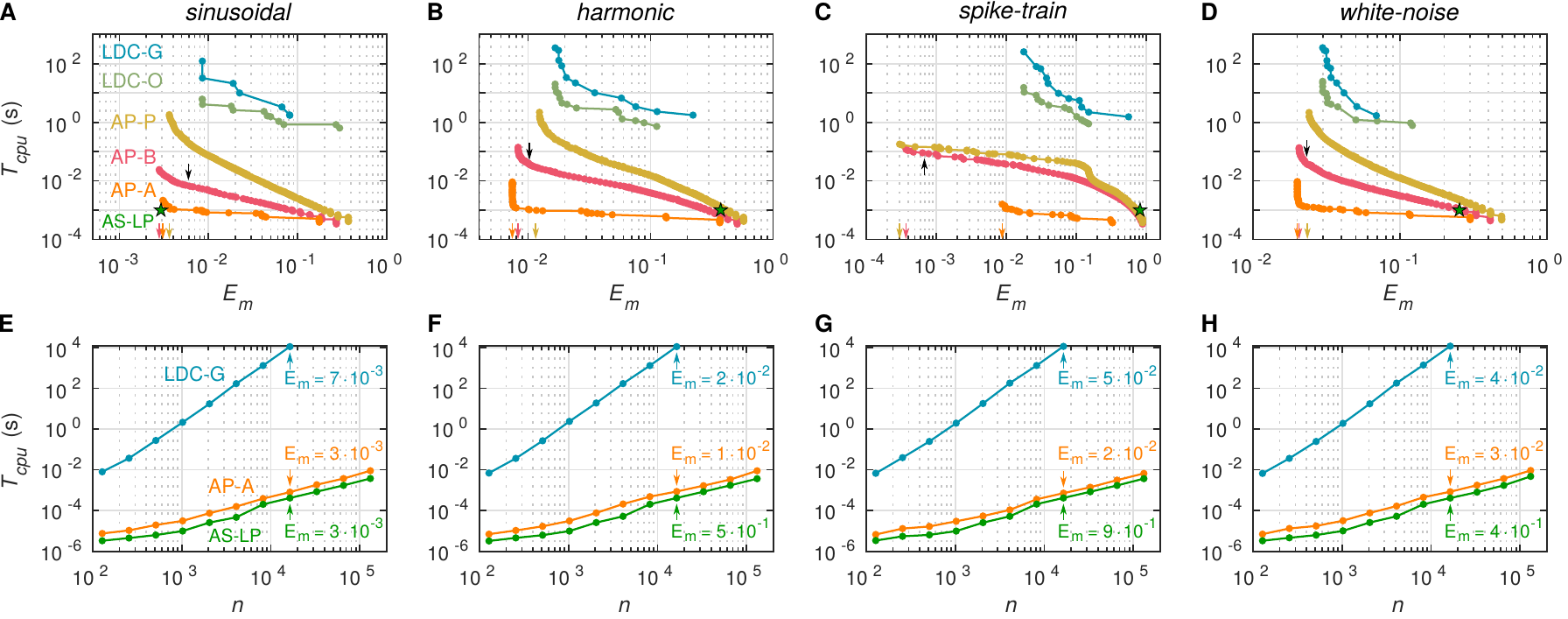}
\caption{Performance evaluation. \textbf{A}--\textbf{D}: Pareto fronts in the $(E_m,T_\mathrm{cpu})$ plane for different demodulation algorithms applied to the four different types of test signals when window splitting is used. Green stars mark the results of the AS-LP method. Color arrowheads point to the lower bounds on $E_m$ for the respective AP algorithms. Black arrowheads indicate demodulation error $E_m$ values of the AP-B algorithm calculated locally for signal windows shown in Fig.\,\ref{fig:2}. \textbf{E}--\textbf{H}: Dependence of the demodulation time $T_\mathrm{cpu}$ on the signal length $n$ at $\epsilon_{tol} = 10^{-3}$ when window splitting is not exploited. $E_m$ values in the legends correspond to demodulation results at $n=2^{14} \approx 1.6 \cdot 10^4$.}
\label{fig:3}
\end{figure*}

Demodulation performance was evaluated by using two complementary measures: 1) error of the modulator estimate, $E_m = \|\mathbf{m} - \mathbf{\hat{m}} \|_2 / \| \mathbf{m}\|_2$; and 2) execution time of the algorithm on the computer, $T_{\mathrm{cpu}}$. We evaluated the AP and LDC algorithms in the mode when $T_\mathrm{cpu}$ depends on the total number of sample points but not on the effective degrees of freedom. This choice made the results general, independent of a selected cutoff frequency $\varpi$. To insure against outliers, we averaged $E_m$ and $T_{\mathrm{cpu}}$ over ten independent signal realizations.

Execution of the AP and LDC algorithms is controlled by a~set of metaparameters whose choice influences the output. Therefore, we aimed for the Pareto fronts, not separate points, in the $(E_m, T_{\mathrm{cpu}})$ plane. Due to the computing speed limitations inherent to the LDC approach, we had to exploit signal decomposition into separate fragments for this analysis. In particular, signals were split into segments that were demodulated separately and then put together \cite{Sell2010}. This allowed achieving a linear growth in the computation time with the total length of the signal, and hence, speeding up the calculations. After identifying the optimal control-parameter combinations, we compared all methods by demodulating whole signals.

Sets of the demodulation control parameters that we considered for the Pareto optimality analysis, including those defining the signal splitting, are provided in Suppl.\,Mat.\,J. Details on the execution of the performance tests on a computer can be found in Suppl.\,Mat.\,K.

\subsection{Results \label{subsec:Performance4}}

Fig.\,\ref{fig:2} shows representative fragments of the test signals from all four classes (gray) and their modulator estimates obtained by using the AS-LP (green) and AP-B (red) algorithms. Whereas the \mbox{AP-B} allows obtaining high-quality estimates $\mathbf{\hat{m}}$ in all four cases, the AS-LP does so only for sinusoidal signals.

Results of the performance evaluation in the form of Pareto fronts in the $(E_m,T_\mathrm{cpu})$ plane for $n=2^{15}$ are displayed in Fig.\,\ref{fig:3}\,A--D. Panels E--H of the same figure show $T_\mathrm{cpu}$~vs.~$n$ relations derived by using no window splitting. A closer analysis of these data reveals the following:
\begin{enumerate}
\item
The AP algorithms feature lower bounds on the demodulation error $E_m$ than the LDC method (Fig.\,\ref{fig:3}\,A--D).
\item The AP algorithms are up to five orders of magnitude faster than their LDC counterparts for achieving the same $E_m$ when optimal signal window splitting is used (Fig.\,\ref{fig:3}\,A--D). The difference is even more pronounced when no window splitting is assumed (Fig.\,\ref{fig:3}\,E--H). For example, to process a 1\,s length signal sampled at 16\,kHz, the LDC needs $10^4$\,s of CPU time, in contrast to $10^{-3}$\,s taken by the AP-A. 
\item $T_\mathrm{cpu}$ varies substantially (up to three orders of magnitude) even between different AP algorithms (Fig.\,\ref{fig:3}\,A--D). The AP-A ranks as the fastest, and the AP-P as the slowest one for all tested signals.
\item Despite the differences in $T_\mathrm{cpu}$, all AP algorithms feature similar lower bounds on $E_m$, except the spike-train signals, when the AP-B and AP-P can noticeably surpass the AP-A on the relative scale (Fig.\,\ref{fig:3}\,A--D). Nevertheless, on the absolute scale, the AP-A still performs reasonably well.
\item For all tested signals, the AP-A algorithm outperforms the AS-LP-based demodulation in the sense that it can achieve the same or smaller errors with the same $T_\mathrm{cpu}$ (Fig.\,\ref{fig:3}\,A--D). Moreover, compared with the AS-LP, AP algorithms exhibit much lower bounds on $E_m$.
\item Even without the window splitting (when the highest demodulation accuracy is attained), the AP-A algorithm takes only 2--3 times longer than the AS-LP method (Fig.\,\ref{fig:3}\,E--H).
\end{enumerate}

We found that the decrease in $E_m$ along the Pareto fronts is mainly determined by the increase in the demodulation window size. In particular, the lower bounds on $E_m$ are achieved by the particular algorithms when the signal is demodulated using no window splitting. The relatively lower precision of the AP-A algorithm compared with AP-B and AP-P in the case of nonstationary spike-trains can be reduced to its acceleration mechanism. Indeed, in the AP-A, upscaling of iterates $\mathbf{a}^{(i)}$ is effectively based on the averaging of $(\mathbf{m}^{(i)}-\mathbf{a}^{(i)})$ over a window of length $\approx n/(2\pi\varpi)$ at each sample point. The precision of these estimates is more vulnerable to deviations from the exact recovery conditions for sparse carriers.

As can be expected, the high accuracy of modulator estimates achieved by the AP algorithms implies the high quality of carrier predictions $\mathbf{\hat{c}} = \mathbf{s} / \mathbf{\hat{m}}$ (see Suppl.\,Mat.\,L and Fig.\,13 therein). The AP approach leaves the AS-LP behind in terms of carrier estimation for all four signal types considered (see Suppl.\,Mat.\,L). When applicable, the inferred $\mathbf{\hat{c}}$ can be further frequency-demodulated by using dedicated techniques (see \cite{Vakman1998,Wiley1977}, and references given there).

The impressive performance of the AP-A algorithm in terms of $E_m$, $E_c$, and $T_\mathrm{cpu}$ makes it an ideal candidate for amplitude demodulation of a wide range of signals. Its AP-B and AP-P counterparts can be used instead if higher precision is needed in specific cases, as illustrated by the spike-train signals above.

\section{Convergence Tests \label{sec:Convergence}}

\begin{figure*}
\centering
\includegraphics[width=1\textwidth]{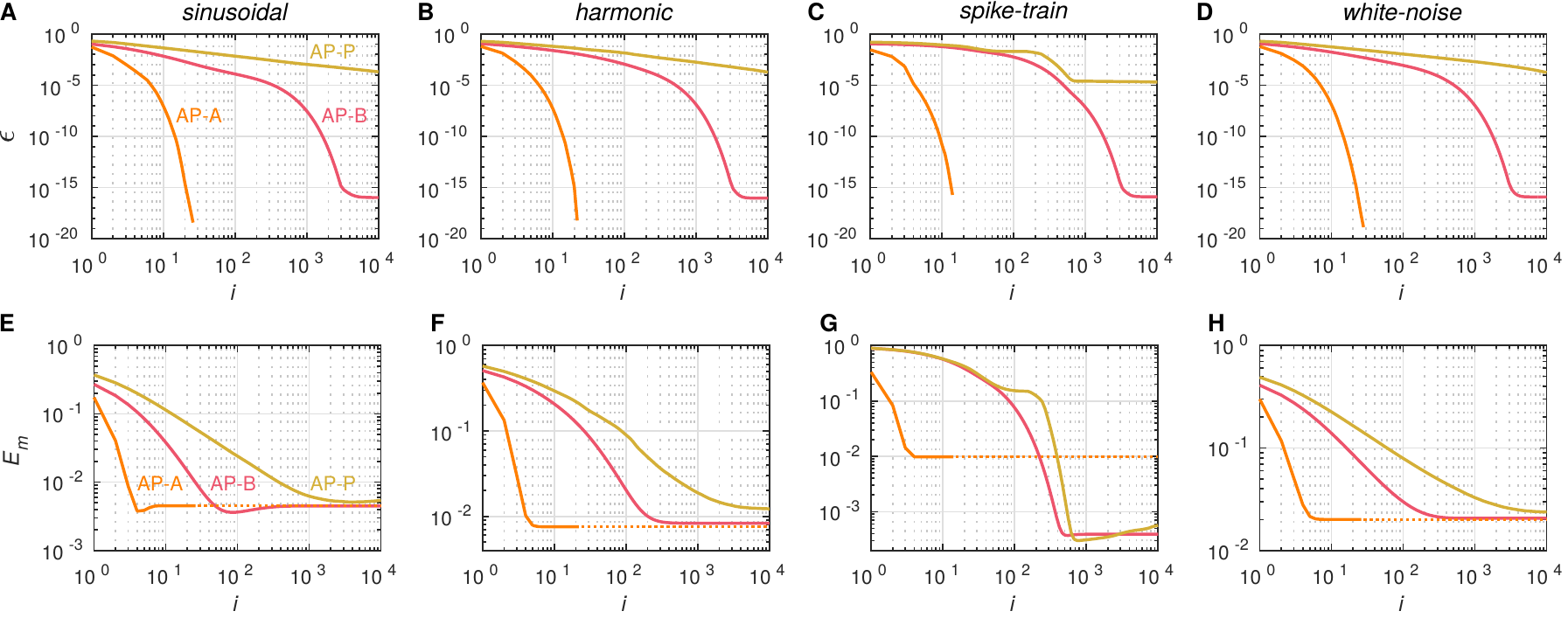}
\caption{Convergence analysis of the AP algorithms. \textbf{A}--\textbf{D}: Dependence of the infeasibility error $\epsilon$ on the iteration number $i$ for the AP-B, AP-A, and AP-P algorithms applied to the four different types of test signals with $n=2^{15}$ and no window splitting. \textbf{E}--\textbf{H}: Analogous plots to A--D made for the demodulation error $E_m$ instead of $\epsilon$. Dotted lines show hypothetical $E_m$ values that would be obtained if we continued the AP-A iterations after reaching the final solution.
}
\label{fig:4}
\end{figure*}

To clarify the differences between the $T_\mathrm{cpu}$ estimates of the three AP algorithms and understand the relationship between the demodulation and infeasibility errors, we performed a convergence analysis with the test signals from the previous section. The simulation results for fixed $n=2^{15}$ using no window splitting are summarized in Fig.\,\ref{fig:4}. A closer inspection uncovers the following:
\begin{enumerate}
\item The convergence rates in terms of both $\epsilon$ and $E_m$ parallel the differences in the computing speed of different AP algorithms. Among them, the fastest is the AP-A, which reaches any given $\epsilon$ or $E_m$ level with the smallest num-ber of iterations. The AP-P algorithm is the slowest one.
\item The AP-A algorithm converges in a finite number of iterations ($<30$) for all types of test signals studied. In particular, it requires only $\leq 5$ iterations to reach the plateau level of the demodulation error $E_m$. This fact explains the extraordinary computational efficiency of the AP-A documented in Section~\ref{sec:Performance}.
\item Differently from the convergence error $\|\mathbf{m}^{(i)}-\mathbf{m}^\dagger\|_2 / \sqrt{n}$, the dependence of $E_m^{\scalebox{0.7}{(}i\scalebox{0.7}{)}}$ on $i$ can be nonmonotonic if $\mathbf{m}^\dagger$ is not strictly equal to $\mathbf{m}$ (Fig.\,\ref{fig:4}\,E,\,G). Then, $E_m^{(i)}$ starts growing with increased $i$ after reaching the minimum point. However, this growth is mild and of no practical importance as long as $\mathbf{m}^\dagger \approx \mathbf{m}$, i.e., $\mathbf{\hat{m}} \approx \mathbf{m}$.
\end{enumerate}

The results shown in Fig.\,\ref{fig:4} represent only signals of fixed length ($n=2^{15}$ sample points). Additional simulations suggested no dependence on
$n$ (see Suppl.\,Mat.\,M).

\section{Robustness Tests \label{sec:Robustness}}

\begin{figure*}
\centering
\includegraphics[width=1\textwidth]{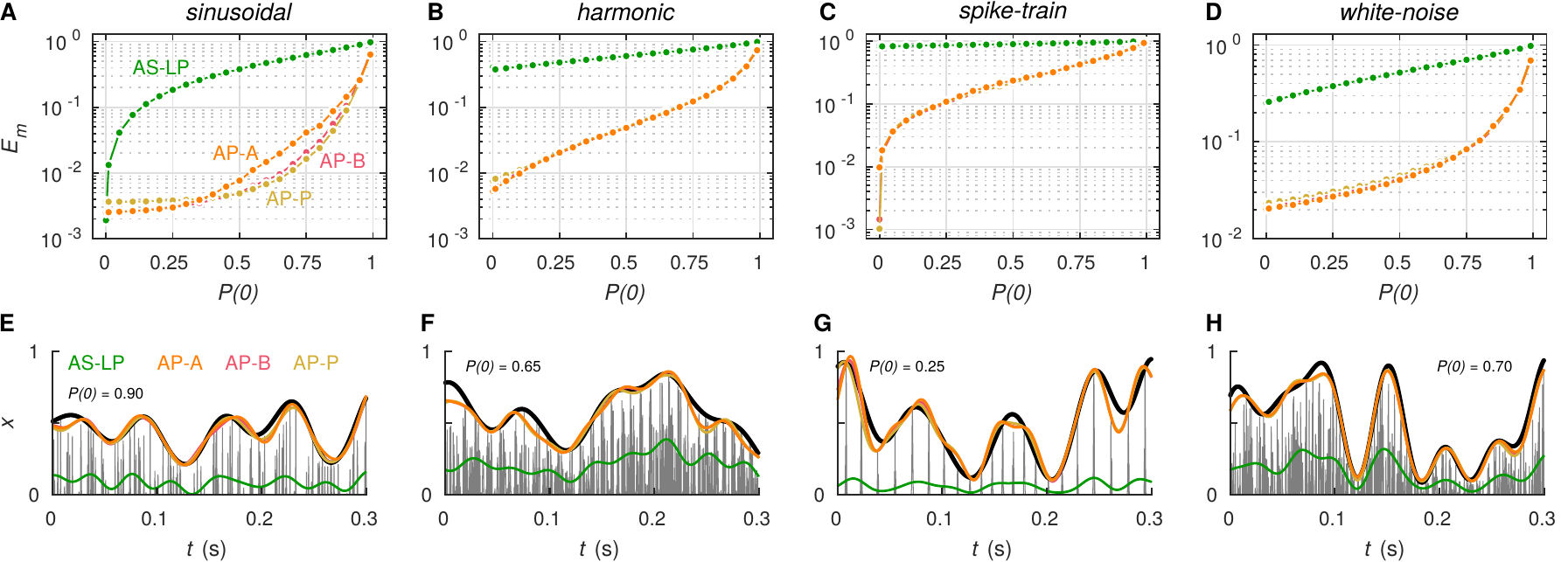}
\caption{Robustness evaluation. \textbf{A}--\textbf{D}: Dependence of the demodulation error $E_m$ on $P(0)$ (the probability of missing points) for the four types of test signals and different AP algorithms at $\epsilon_{tol}=10^{-4}$ (color coding). \textbf{E}--\textbf{H}: Representative examples of demodulation at various $P(0)$ levels for the test signals from A--D. Color code: gray -- the absolute-value signal, black -- the original modulator, color -- modulators inferred by different algorithms.}
\label{fig:5}
\end{figure*}

The pivotal condition for successfully separating the modulator-carrier information of a given signal by our approach is $\lceil n/d \rceil \geq 2\omega-1$. In practice, this requirement is not necessarily met. Hence, the choice of a particular demodulation method must be guided not only by the algorithmic efficiency but also robustness to deviations from the ideal recovery conditions. To shed light on this aspect, we considered demodulation of the test signals from Section~\ref{subsec:Performance1} corrupted by a multiplicative Bernoulli-$\{0,1\}$ noise. In this setup, sample points, including the decisive $|s_i|=m_i$, are eliminated with the probability of ``0'' elements in the noise ($P(0)$), effectively decreasing the value of $\lceil n/d \rceil$.

We found that all three AP algorithms considered in this work show a similar degree of robustness to increased $P(0)$ (see Fig.\,\ref{fig:5}). Only in the case of sinusoidal signals, the AP-A is slightly inferior to the AP-B and AP-P. Interestingly, the advantage of the AP-B and AP-P over the AP-A in the case of spike-train signals discussed in Section~\ref{subsec:Performance4} disappears in the presence of even small distortions (see Fig.\,\ref{fig:5}\,C). The differences in the $E_m$~vs.~$P(0)$ relations seen in Fig.\,\ref{fig:5}\,A--D are predetermined by different densities of $|c_i| \simeq 1$ points inherent to each carrier type. Analogous results to those shown in Fig.\,\ref{fig:5}\,A--D are obtained when considering carrier recovery via $\mathbf{\hat{c}} = \mathbf{s} / \mathbf{\hat{m}}$ (see Fig.\,14 in Suppl.\,Mat.\,L).

In contrast to the AP approach, the AS-based demodulation is highly vulnerable to missing sample points, and hence, to decreased $\lceil n/d \rceil$ (Fig.\,\ref{fig:5}\,A--D). Even for sinusoidal signals, which the AS and AS-LP are specially designed for, the zeroing of data points leads to a rapid decline in demodulation quality (Fig.\,\ref{fig:5}\,A,\,E).

The robustness to missing sample points endows the AP demodulation method with a highly valuable practical advantage. In particular, it can be exploited in real-world situations when: 1) the sampling rate is low; 2) some segments of the signal values are lost; 3) some sample points are corrupted by noise such that the level of these points can be reduced below the real modulator by low-pass filtering or explicitly identifying them. In this context, the PAD and LDC demodulations compare to the AP approach by construction \cite{Turner2010}.

\section{High-Level Properties \label{sec:HighLevel}}

As emphasized in Section~\ref{sec:Problem}, different demodulation methods can be derived by requesting adherence of the inferred modulators and carriers to a set of particular properties. Typically, various combinations that consist of a few out of many reasonable requirements are sufficient for unique demodulation formulations. However, some of these requirements are inconsistent with each other, making virtually all classical demodulation approaches fail to satisfy one or another essential condition \cite{Vakman1996, Loughlin1996, Turner2010}. For example, the AS demodulation method may return an unbounded modulator estimate for a bounded signal \cite{Loughlin1996}.

The AP approach formulated in this work is compatible with the following high-level requirements, which have crystallized as inseparable from the notion of proper amplitude demodulation with time \cite{Sell2010}, \cite[Section~3.5.2]{Turner2010}:

\begin{itemize}
\item \textbf{Boundedness:} The modulator and carrier of a bounded signal are bounded. In particular, it is required that $-\infty < \hat{m}_i < +\infty$ and $-1 \leq \hat{c}_i \leq 1$ for every $ i \in \In$. In the case of the AP approach, the boundedness of the modulator is guaranteed by the convergence of the AP algorithms. The boundedness of the carrier then follows from the constraint $\hat{m}_i \geq |s_i|$ and the fact that $\mathbf{\hat{c}} = \mathbf{s} \circ \mathbf{\hat{m}}^{-1}$.
\item \textbf{Scale covariance:} The modulator and carrier of a scaled signal are equal to the modulator and carrier obtained from the original signal and then scaled by the same amount. The adherence of the AP approach to this condition follows from two facts. First, projection operators $\mathbf{P}_{\Sgeqss}$ and $\mathbf{P}_{\Sw}$ are homogeneous with degree 1, i.e., $\mathbf{P}_{\mathcal{S}}[\alpha \cdot \mathbf{s}] = \alpha \cdot \mathbf{P}_{\mathcal{S}}[\mathbf{s}]$.  Second, each iteration of the AP algorithms can be expressed as a weighted sum of these projections with the weights independent of the scale.
\item \textbf{Smoothness:} The modulator of a bounded signal in its continuous-time representation is smooth. Because we use a discrete-time representation, this requirement has to be adjusted. In particular, let us denote by $\hat{m}'_t$ and $\hat{m}'_{t+\Delta t}$ the finite-difference approximations of the modulator's time-derivatives of any order at two subsequent time points: $t$ and $t + \Delta t$. Then, we require that, for any $\epsilon>0$, there exists a $\delta>0$ such that $|\hat{m}'_{t+\Delta t}-\hat{m}'_t| < \epsilon$ when $|\Delta t| < \delta$. The AP approach satisfies this requirement through the boundedness of the modulator and the bandwidth constraint set by $\Sw$ on it.
\item \textbf{Idempotence:} Information associated with the qualities of modulators and carriers is fully separated. Specifically, demodulation reapplied to an estimated modulator (carrier) must return the same modulator (carrier). The AP approach satisfies the idempotence requirement for the modulator exactly. Indeed, when any AP algorithm is applied to its final solution $\mathbf{\hat{m}} = \mathbf{m}^\dagger$, the latter is recognized as the final solution again after the first new iteration by construction. Regarding the carrier, the idempotence holds whenever the recovery conditions discussed in Section~\ref{subsec:Problem-Formulation-C} are met. That is because, in those cases, $\mathbf{\hat{c}}$ resulting from the first demodulation contains a sufficient number of $|\hat{c}_i|=1$ points to uniquely define the $\mathbf{\hat{m}}=\mathbf{1}$ as the norm-minimizing element of $\Sgeqc \cap \Sw$. If the recovery conditions are met only approximately, we expect no marked deviations from the idempotence condition (see Fig.\,16 in Suppl.\,Mat.\,N).
\end{itemize}
By fulfilling the above requirements, the AP approach parallels the methods of PAD and LDC demodulation \cite{Turner2010, Sell2010}. In this sense, all of them outperform the classical techniques.

\section{Demodulation of Speech Signals \label{sec:Speech}}

\begin{figure*}
\centering
\includegraphics[width=1\textwidth]{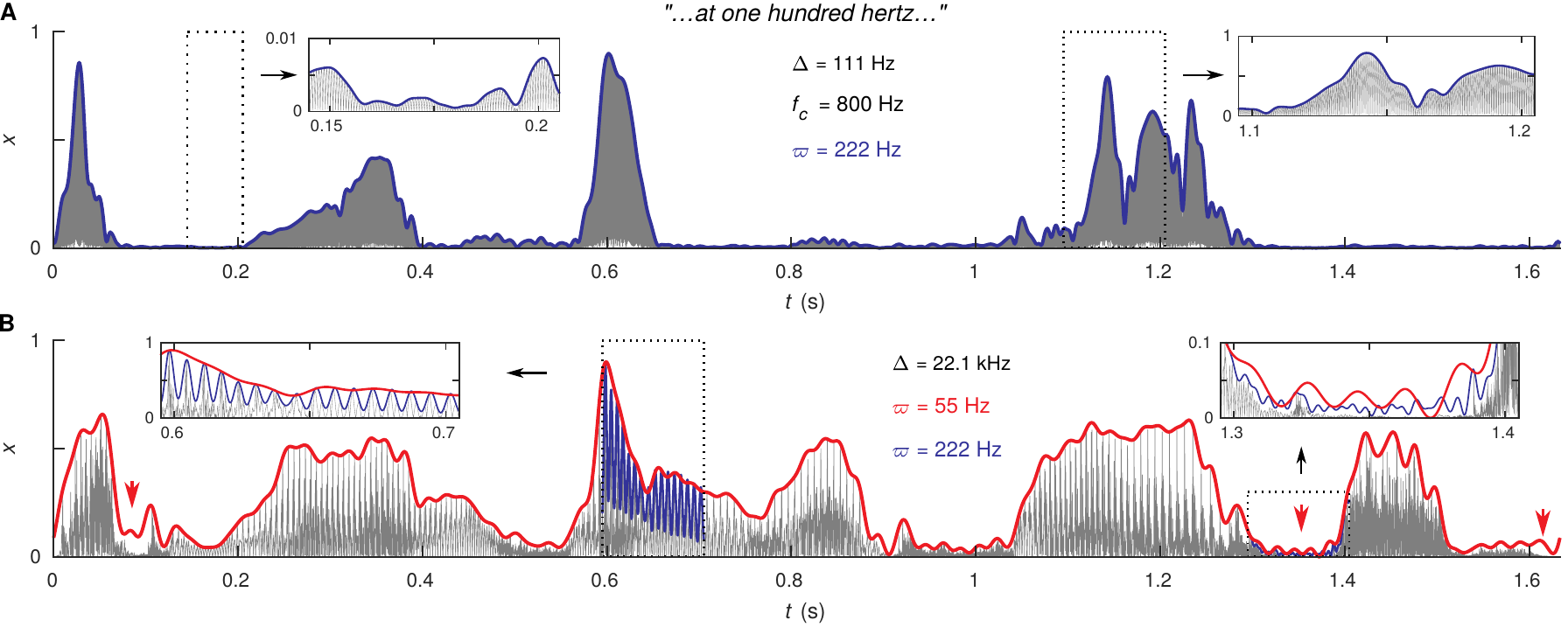}
\caption{Direct demodulation of speech signals. \textbf{A}: A band-pass-filtered signal of an utterance \textit{``\ldots at one hundred hertz \ldots''} by a female speaker; the signal was obtained with an equivalent rectangular bandwidth filter of the cochlea centered at $f_c=800~\mathrm{Hz}$ and $\Delta = 111~\mathrm{Hz}$. \textbf{B}: The original signal of the utterance used in panel A (full bandwidth of $22.1~\mathrm{kHz}$). In both panels, gray color marks the absolute-value version of the signals considered for demodulation. Blue and red lines show their modulators obtained by using the AP-A algorithm with the cutoff frequency $\omega$ set to, respectively, $55~\mathrm{Hz}$ and $222~\mathrm{Hz}$. Insets display time-expanded segments of the original window. Red arrowheads indicate the ringing artifacts of the modulator at some prolonged intervals of low signal levels. Source of the original signal: audio edition of \textit{The Economist} magazine, issue March 19th 2016, article ``Restoring lost memories.''}
\label{fig:6}
\end{figure*}

Amplitude demodulation is of central importance in various tasks of processing and analysis of speech signals. Application-wise, this procedure is used in hearing restoration \cite{Wilson1991, Wilson2008}, speech recognition \cite{Kingsbury1998, Wu2011, Lee2016}, and source separation \cite{Hu2004, Atlas2005}. On the theory side, amplitude demodulation is exploited in neurophysiological and psychophysical studies of auditory information processing in the brain \cite{Smith2002, Joris2004, Zeng2005, Goswami2019}. Depending on the problem, demodulation of either narrow subband \cite{Hu2004, Wu2011}, intermediate subband \cite{Smith2002, Wilson1991}, or whole wideband signal \cite{Shannon1995, Goswami2019} is needed. In all these cases, modulators and carriers convey the information about specific aspects of speech, e.g., semantic meaning, associated emotion, or speaker identity, that need to be extracted.

In this section, we apply the newly-introduced AP approach to speech demodulation to further demonstrate its potential. To represent the range of possible real-world situations, we consider two limiting signal types: 1) a narrow subband component of a signal obtained by a standard auditory ERB filter \cite{Glasberg1990}; and 2) the original wideband signal.

\subsection{Direct demodulation \label{subsec:Speech1}}

By construction, the output of auditory ERB filters occupies a frequency subband whose width $\Delta$ is much smaller than its center frequency $f_c$ \cite{Glasberg1990}. The resulting signal is an amplitude- and phase-modulated sinusoidal $\mathbf{s} = \mathbf{m} \circ \sin(2\pi f_c \mathbf{t} + \bm{\varphi})$, with most of the energies of $\mathbf{m}$ and $\bm{\varphi}$ residing in the frequency interval $[0, \Delta]$ \cite{Flanagan1980}. Hence, by the recovery conditions of the AP approach (see Section~\ref{subsec:Problem-Formulation-C}), setting the cutoff frequency $\varpi$ between $\Delta$ and $f_c$ necessarily results in accurate estimates $\mathbf{\hat{m}}$ and $\mathbf{\hat{c}}$. In particular, note that the local maximums of $|\mathbf{s}|$ correspond to the true sample points $|s_i|=m_i$. Thus, the high quality of demodulation is visually conveyed by a tight match of $\mathbf{\hat{m}}$ and $|\mathbf{s}|$ at these sample points. We illustrate our claims in Fig.\,\ref{fig:6}\,A, where a band-pass component of a female utterance \textit{``\ldots at one hundred hertz \ldots''} with $f_c=800\,\mathrm{Hz}$, $\Delta = 111\,\mathrm{Hz}$, and $\varpi=222\,\mathrm{Hz}$ is considered. Taking into account that the local maximum points $|s_i|$ are locally regular and that they correspond to $m_i$, we could exploit \emph{Proposition~A.2} in Suppl.\,Mat.\,A to find that $E_m \leq 8 \cdot 10^{-3}$.

Wideband speech signals are more challenging than their narrow subbands. They are built of temporarily structured segments of quasi-random and quasi-harmonic carriers, possibly featuring frequency glides \cite{Shoup1976}. These carriers are amplitude-modulated at different timescales, ranging between a hundred milliseconds and several seconds \cite{Turner2010, Keitel2018}. The power spectral density of the corresponding modulators is vanishingly small above 20\,Hz (see Fig.\,1 in \cite{Bosker2018}). Moreover, as we demonstrate in Suppl.\,Mat.\,O, the carrier components of natural speech signals align to the recoverability conditions of the AP approach for $\mathbf{m} \in \Mwm$ with $\omega$ up to at least $\sim 50\,\mathrm{Hz}$. Therefore, we expect appropriate performance from the AP algorithms in the setting of wideband speech.

Fig.\,\ref{fig:6}\,B displays demodulation results of the full-band version of the speech segment considered in Fig.\,\ref{fig:6}\,A by the AP-A algorithm with $\varpi=55~\mathrm{Hz}$. The obtained $\mathbf{\hat{m}}$ (red) envelops separate phonemes of the sound waveform tightly, indicating appropriate recovery of the true $\mathbf{m}$ (see Section~\ref{subsec:Speech2} next). However, intervals corresponding to prolonged transitions between phonemes or words are corrupted by ringing artifacts (marked by red arrowheads in Fig.\,\ref{fig:6}\,B), implying the necessity of higher frequency components to represent these transitions. Hence, although the power spectral density of the true $\mathbf{m}$ is very low above $20\,\mathrm{Hz}$, it sums to a noticeable contribution. Unfortunately, any attempt to cancel the artifacts by just increasing $\varpi$ fails by breaking the recovery conditions, as illustrated by the blue line in Fig.\,\ref{fig:6}\,B ($\varpi=222~\mathrm{Hz}$ there). No improvement is achieved by utilizing the AP-B, AP-P, or LDC algorithms either (data not shown).

\subsection{Demodulation using dynamic range compression \label{subsec:Speech2}}

\begin{figure*}
\centering
\includegraphics[width=1\textwidth]{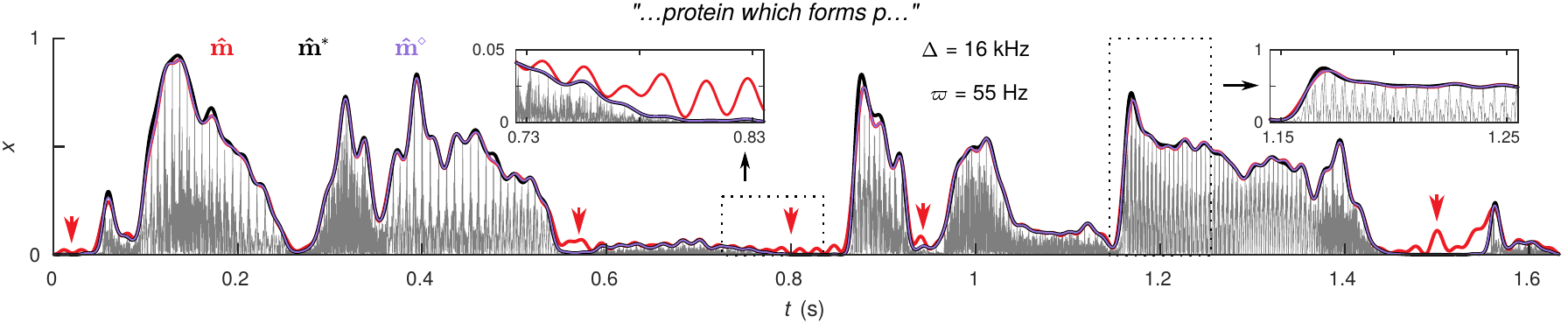}
\caption{Demodulation of speech signals using dynamic range compression. An audio signal of \textit{``\ldots protein which forms p \ldots''} uttered by a female speaker (full bandwidth of $16~\mathrm{kHz}$). Gray -- the absolute-value version of the signal considered for demodulation,  red -- its modulator obtained by using the AP-A algorithm with $\omega=55~\mathrm{Hz}$ and no compression (as in Fig.\,\ref{fig:6}), black -- modulator of the signal obtained when employing the dynamic range compression [see \eqref{eq:SpeechDemod2}], violet -- interpolation of the latter two [see \eqref{eq:SpeechDemod3}]. Red arrowheads indicate ringing artifacts of the modulator estimate. Source of the original signal: the same as Fig.\,\ref{fig:6}.}
\label{fig:7}
\end{figure*}

The aforementioned problem with modulator estimates of signals with sharp transitions to/from prolonged intervals of low-signal amplitude can be resolved by using a dynamic range compression. In particular, instead of demodulating the original signal $\mathbf{s}$ directly, we first apply a chosen AP algorithm to its compressed version:
\begin{equation}
\mathbf{\underline{s}} = \mathrm{sgn}(\mathbf{s}) \circ \mathbf{|s|}^{1/p}.\label{eq:SpeechDemod1}
\end{equation}
Here, $p \in (1, +\infty)$ controls the level of compression. The modulator estimate $\mathbf{\hat{m}^*}$ of $\mathbf{s}$ is then evaluated by inverse-transforming the modulator $\mathbf{\underline{\hat{m}}}$ of $\mathbf{\underline{s}}$:
\begin{equation}
\mathbf{\hat{m}^*} = \mathbf{\underline{\hat{m}}}^{p}.\label{eq:SpeechDemod2}
\end{equation}
The idea behind \eqref{eq:SpeechDemod1} is that the compression makes signals more uniform and, effectively, smooths their sharp changes responsible for ringing artifacts in the modulator estimates. These sharp changes are restored in the modulators without artifacts by the inverse transform \eqref{eq:SpeechDemod2}.

The expected effect of the compression procedure is illustrated in Fig.\,\ref{fig:7}, where signal demodulation of an utterance \textit{``\ldots protein which forms p\ldots''} is considered. Differently from the direct demodulation result $\mathbf{\hat{m}}$ (red line), the estimate $\mathbf{\hat{m}^*}$ obtained by using the compression with $p=3$ (black line) shows good alignment with $|\mathbf{s}|$ in the segments of both low and high intensity. To justify that this alignment really reflects the recovery of the true modulator, we performed additional tests where chimeric signals built of $\mathbf{\hat{m}^*}$ from Fig.\,\ref{fig:7} and natural speech carriers were demodulated (see Suppl.\,Mat.\,O). We found low demodulation errors, with $E_m$ ranging between $9 \cdot 10^{-3}$ and $5 \cdot 10^{-2}$ for different carrier components of speech signals (see Fig.\,18).

The compression level $p=3$ used above was adjusted by a trial and error for speech signals. In general, the gains in accuracy at low levels with increased $p$ comes at the expense of reduced precision of modulator estimated at high signal levels. Thus, a compromise between those two effects must be reached to find an optimal $p$. Moreover, the precision of the modulator estimates can be further increased by interpolating between $\mathbf{\hat{m}^*}$ (more accurate for low signal levels) and $\mathbf{\hat{m}}$ (more accurate for high signal levels). For example, the violet line in Fig.\,\ref{fig:7} shows a weighted average of the form
\begin{equation}
\mathbf{\hat{m}^\diamond} = \mathbf{\hat{m}} \circ \mathbf{w} + \mathbf{\hat{m}^*} \circ (1-\mathbf{w}),\label{eq:SpeechDemod3}
\end{equation}
where
\begin{equation}
\mathrm{w_i} = \bigg(\frac{1-e^{a \cdot (\mathrm{\hat{m}^*_i}/\max[\mathbf{\hat{m}^*}])}}{1+e^{a \cdot (\mathrm{\hat{m}^*_i}/\max[\mathbf{\hat{m}^*}])-b}}\bigg) \cdot \bigg( \frac{1-e^{a}}{1+e^{a-b}} \bigg)^{-1}
\label{eq:SpeechDemod4}
\end{equation}
for $i \in \In$, with $b=3$ and $a=10$. In general, an optimal interpolation between $\mathbf{\hat{m}}$ and $\mathbf{\hat{m}^*}$ can be learned by minimizing $\| \mathbf{\hat{m}^\diamond} \|_2^2$ over a chosen class of functions. Other compression models than \eqref{eq:SpeechDemod1}, e.g., $\mathbf{\underline{s}} = \mathrm{sgn}(\mathbf{s}) \circ \log (1+p \cdot |\mathbf{s}|)$, can be used to evaluate $\mathbf{\hat{m}^*}$ as well.

\subsection{Demodulation in real-time \label{subsec:RealTime}}

\begin{figure*}
\centering
\includegraphics[width=1\textwidth]{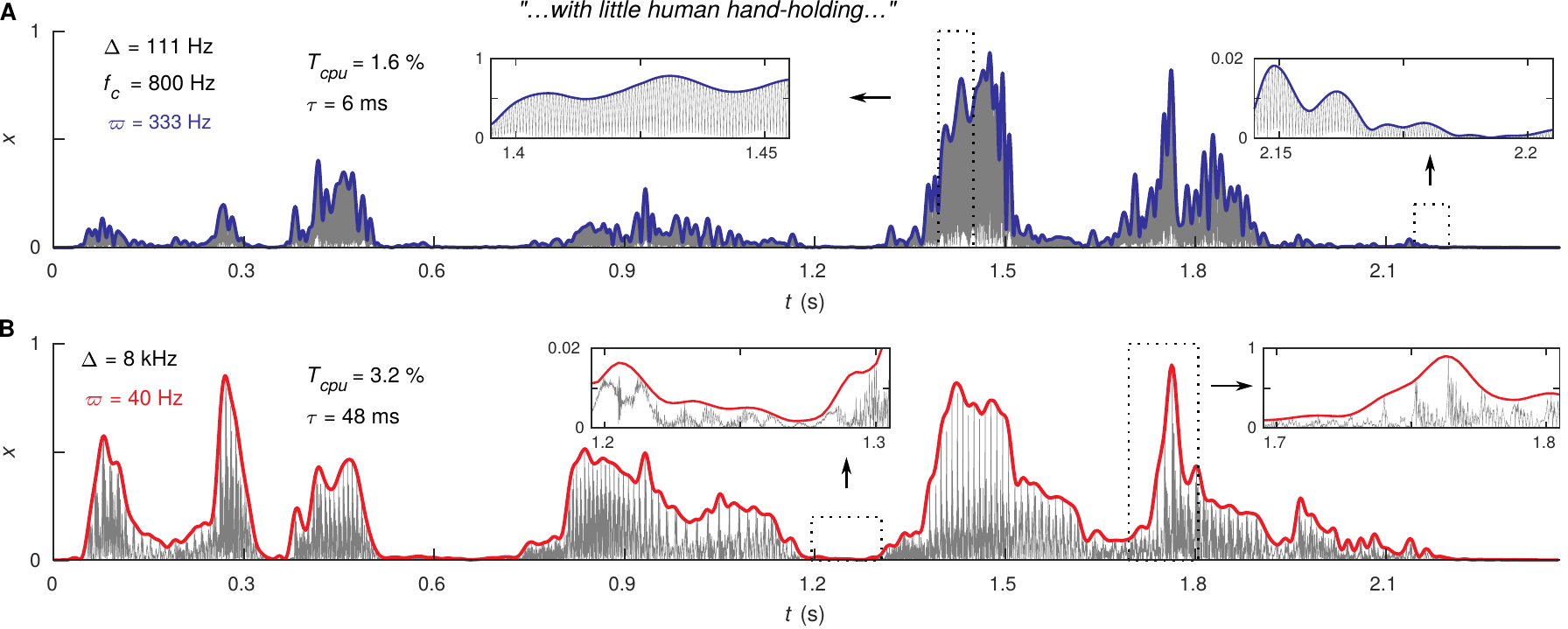}
\caption{Demodulation in real-time. \textbf{A}: A band-pass-filtered signal of an utterance \textit{``\ldots with little human hand-holding \ldots''} by a male speaker; the signal was obtained with an equivalent rectangular bandwidth filter of the cochlea centered at $800\,\mathrm{Hz}$ (bandwidth of $111~\mathrm{Hz}$). \textbf{B}: The original signal of the utterance used in panel A (full bandwidth of $8~\mathrm{kHz}$). In both panels, gray color marks the absolute-value version of the signals considered for demodulation. Blue and red lines show their modulators obtained by using the real-time version of the AP-A algorithm with the cutoff frequency $\omega$ set to, respectively, $333~\mathrm{Hz}$ and $40~\mathrm{Hz}$. Source of the original signal: audio edition of \textit{The Economist} magazine, issue March 19th 2016, article ``Artificial intelligence and Go.''}
\label{fig:8}
\end{figure*}

A number of amplitude demodulation applications, e.g., speech recognition \cite{Kingsbury1998}, ultrasound imaging \cite{Hoskins2019}, and cochlear prosthesis \cite{Wilson2008}, necessitate real-time processing. As we demonstrate below, the exceptional computational efficiency of the AP approach allows it to fulfill that requirement.

The nature of the task implies that online modulator estimates have to be generated by sequentially demodulating windowed segments $\mathbf{s}^{(j)}$ of a signal $\mathbf{s}$ at each updated sample point $j$ across time:
\begin{equation}
\mathbf{s}^{(j)}: s_i^{(j)} = w_i \cdot s_{j-k_l-1+i}, \qquad i \in \{ 1, 2, \ldots, k\}.\label{eq:RTDemod1}
\end{equation}
Here, $k$ is the number of sample points corresponding to the segment, and $k_l$ denotes the number of sample points of it that are to the left of the current point $j$. $w_i, w_2, \ldots, w_k$ are vector elements of the window function. The real-time modulator estimate $\hat{m}_j^\star$ at sample point $j$ is calculated as
\begin{equation}
\hat{m}_j^{\star} = \hat{ m}_{k_l+1}^{(j)},\label{eq:RTDemod2}
\end{equation}
where $\mathbf{\hat{m}}^{(j)}$ is a modulator estimate of $\mathbf{s}^{(j)}$.

It follows from the time-frequency uncertainty principle \cite{Gabor1946} that accurate evaluation of $\hat{m}_j^\star$ requires $\mathbf{s}^{(j)}$ with a duration of the order of the inverse of the effective bandwidth of the modulator, or longer. This condition sets the lower bounds on the segment length $k$ and sampling delay $k_\tau=k-k_l-1$ of $\mathbf{\hat{m}}^\star$. We found empirically that $k_\tau \approx 2 \cdot (f_s / \varpi)$ and $k \approx 4 \cdot (f_s / \varpi)$ are typically sufficient for accurate demodulation of wideband speech. These numbers are around two times smaller for narrow frequency band components of these signals. We know that $\varpi \geq 40~\mathrm{Hz}$ for the wideband speech and its subbands. Thus, delays $k_\tau \leq 50~\mathrm{ms}$ for estimating $\mathbf{\hat{m}}^\star$ are sufficient without a sacrifice in precision then. The main requirement for the window function in \eqref{eq:RTDemod1} is that it smoothly scales the signal to 0 at the boundaries, with no effect at the midst. We used a modified version of the Hann window for this purpose:
\begin{align}
w_i =\begin{cases} \sin^2\Big(\frac{\pi \cdot (i-1)}{2 \cdot k_l}\Big), & \mbox{$\quad 1 \leq i \leq k_l$}\\ 1, & \mbox{$\quad i = k_l+1$}\\
\cos^2\Big(\frac{\pi \cdot (i-k+k_\tau)}{2 \cdot k_\tau}\Big), & \mbox{$\quad k-k_\tau+1 \leq i \leq k$} \end{cases}.
\label{eq:RTDemod3}
\end{align}

Fig.\,\ref{fig:8} shows simulation results of real-time demodulation of a male utterance \textit{``\ldots with little human hand-holding \ldots''} (sampling rate $f_s = 16~\mathrm{kHz}$) based on the AP-A algorithm. There, demodulation was performed with $k=1536$ and $k_\tau=768$ ($\tau=48~\mathrm{ms}$) for the original signal (Fig.\,\ref{fig:8}\,B). Its subband component centered at 800~Hz (Fig.\,\ref{fig:8}\,A) was processed with $k=128$ and $k_\tau=65$ ($\tau=4~\mathrm{ms}$). In each case, $\hat{m}^\star_j$ was updated with the frequency of $10 \cdot \varpi$. The obtained estimates $\mathbf{\hat{m}^\star}$ are in very good agreement with $\mathbf{\hat{m}^*}$ derived by using offline demodulation of the whole signal, with $\|\mathbf{\hat{m}^\star}-\mathbf{\hat{m}^\diamond}\|_2 /  \| \mathbf{\hat{m}^\diamond} \|_2 < 0.02$. Importantly, they were achieved with modest CPU usage: $T_{\mathrm{cpu}}$ amounted to only 1.6\,\% (subband signal) and 3.2\,\% (wideband signal) of the time length of the demodulated signal on an Intel Core i7-7700 CPU run in single-thread mode. For comparison, these numbers were, respectively, $\sim 5 \cdot 10^3$ and $\sim 6 \cdot 10^4$ times higher for the LDC method.

An advantageous side effect of splitting the signal into small windows for demodulation is that it prevents the ringing artifacts (compare Fig.\,\ref{fig:8}\,B and Fig.\,\ref{fig:6}\,B). This is so because signal levels do not typically spread over different scales in a short time window. The window splitting also allows generalizing demodulation to situations when the cutoff frequency $\omega$ of the modulator varies strongly in time.

\section{Extensions and Generalizations \label{sec:Extensions}}

\subsection{Demodulation in higher dimensions \label{subsec:HigherDim}}

\begin{figure*}
\centering
\includegraphics[width=1\textwidth]{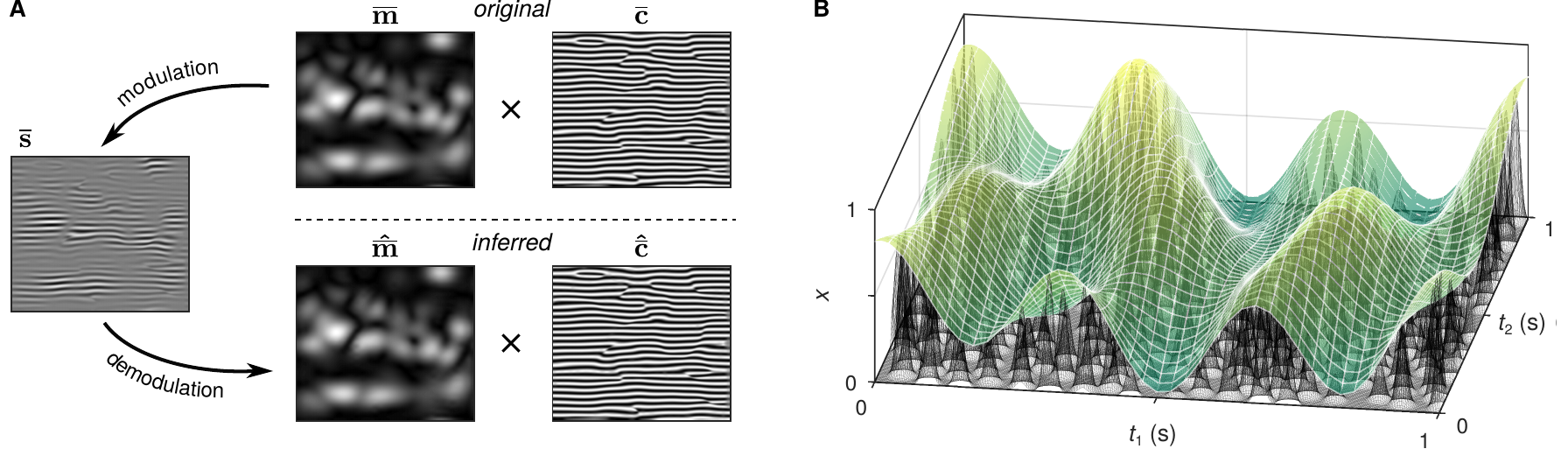}
\caption{Demodulation in 2D by using the AP-A algorithm. \textbf{A}: Synthetic fringe pattern of moderate bandwidth. Left -- the pattern to demodulate, upper right -- the original modulator and carrier, lower right -- the inferred modulator and carrier. \textbf{B}: Wideband signal consisting of randomly-placed spikes of finite width. Black grid -- the signal $\mathbf{\bar{s}}$ to demodulate, white grid -- the original modulator $\mathbf{\overline{m}}$ of the signal, color surface -- the  estimated modulator.}
\label{fig:9}
\end{figure*}

Amplitude demodulation has found successful applications beyond the setting of 1D signals. Several 2D extensions of the classical AS approach have been introduced and used for solving tasks in computer vision \cite{Felsberg2001, Larkin2001}, analysis of speech spectrograms \cite{Aragonda2015}, and biomedical imaging \cite{Seelamantula2012, Wachinger2012, Nadeau2014}. The AS framework has also been extended to calculate modulators and carriers for signals over graphs \cite{Venkitaraman2019}. These methods are limited to locally narrowband signals, which manifest visually as fringe patterns (see Fig.\,\ref{fig:9}\,A). This bandwidth restriction is evaded by a generalization of the AP approach to higher dimensions that we present next. The extension is immediate and follows from intuitive abstractions of the constraint sets introduced in Section~\ref{sec:Problem}.

Consider a $D$-dimensional signal $s(t_1,t_2,\ldots,t_D)$. Its uniformly sampled version $\mathbf{\overline{s}}$ is an element of an $n$-dimensional Euclidean space $\mathbb{T}^{n_{}}_D$ of real-valued order $D$ tensors with $n = \prod_{i=1}^D n_i$ and the inner product $\langle \mathbf{\overline{s}}^{(1)}, \mathbf{\overline{s}}^{(2)}\rangle = \sum_{i_1=1}^{n_1} \cdots \sum_{i_D=1}^{n_D} (\bar{s}^{(1)}_{\iD} \cdot \bar{s}^{(2)}_{\iD})$. The respective $D$-dimensional DFT is given by
\begin{equation}
\mathbf{\overline{F}} = \mathbf{F}^{(1)} \otimes \mathbf{F}^{(2)} \otimes \cdots \otimes \mathbf{F}^{(D)}, \label{eq:HighDim1}
\end{equation}
where $\mathbf{F}^{(i)}$ is a unitary DFT defined over $\mathbb{R}^{n_i}$. Then, the analogs of the constraint sets $\Sgeqz$, $\Swm$, $\Sgeqs$, $\Sleqo$, and $\Sd$ from Section~\ref{sec:Problem} read as
\begin{align}
\Sgeqzbar &= \{\mathbf{\overline{x}} \in \mathbb{T}^{n_{}}_D: \overline{x}_{\iD} \geq 0, \, i_j \in \mathcal{I}_{n_j} \}, \nonumber\\
\Swmbar &= \{\mathbf{\overline{x}} \in \mathbb{T}^{n_{}}_D: (\mathbf{\overline{F}}\hspace{1pt}\mathbf{\overline{x}})_{\iD}=0, \, i_j \in (\mathcal{I}_{n_j} \setminus \mathcal{I}_{n_j}^{\omega_j}) \}, \label{eq:HighDim2}\\
\Sgeqsbar &= \{\mathbf{\overline{x}} \in \mathbb{T}^{n_{}}_D: \overline{x}_{\iD} \geq |\bar{s}_{\iD}|, \, i_j \in \mathcal{I}_{n_j} \}, \nonumber
\end{align}
and
\begin{equation}
\begin{aligned}
\Sleqobar &= \{\mathbf{\overline{x}} \in \mathbb{T}^{n_{}}_D: |x_{\iD}| \leq 1, \, i_j \in \mathcal{I}_{n_j} \}, \\
\Sdbar &= \big\{\mathbf{\overline{x}} \in \mathbb{T}^{n_{}}_D: (\forall\iDc)R(\iDc,\mathbf{\overline{x}},\mathbf{d}) \geq 1, \\ 
& \hspace{60pt} (\exists \iDc) R(\iDc,\mathbf{\overline{x}},\mathbf{d})=1 \},
\end{aligned}
\label{eq:HighDim3}
\end{equation}
where
\begin{equation}
\begin{aligned}
\textstyle
\hspace{-3pt} R(\iDc,\mathbf{\overline{x}},\mathbf{d}) = \textstyle \sum_{j_1 \geq i_1} \hspace{-2pt}\cdot\hspace{-2pt}\cdot\hspace{-2pt}\cdot \sum_{j_D \geq i_D} \big[ I_{\{1\}}(|\overline{x}_{j_1 \cdots j_D}|) \\
\textstyle \cdot \theta \big(1-\sum_{k=1}^D(i_k-j_k)^2/d_k^2 \big) \big] - I_{\{1\}}(|\overline{x}_{i_1 \cdots i_D}|).
\end{aligned}
\label{eq:HighDim4}
\end{equation}
Simply substituting \eqref{eq:HighDim2}\,--\,\eqref{eq:HighDim3} for their $D=1$ versions in \eqref{eq:ProbForm2_}, \eqref{eq:ProbForm4_}, \eqref{eq:Recovery2_1_}, \eqref{eq:Recovery4_1_}, \eqref{eq:MathPrel1x}, and \eqref{eq:MathPrel3x} generalizes the modulator $\Mwm$ and carrier $\Cd$ sets as well as the modulator estimator $\mathbf{\hat{m}}$ and the respective AP algorithms. In particular, an $\mathbf{\overline{m}} \in \Mwmbar$ is a nonnegative signal with a low-pass rectangular spectrum set by $\bm{\omega}=(\omega_1, \ldots, \omega_D)$ along each of the $D$ dimensions in the DFT domain. A $\mathbf{\overline{c}} \in \Cdbar$ is a signal bounded between $-1$ and $1$ with the $|\bar{c}_{\iD}|=1$ sample points packed sufficiently densely, as implied by $\mathbf{d} = (d_1, \dots, d_D)$.

Without providing formal proofs, we state that all propositions and assertions of Sections~\ref{sec:Problem} and \ref{sec:Algorithms} about the modulator recoverability and convergence of the AP algorithms generalize to $D$-dimensional signals defined above. All quantitative conditions involving the parameters $\varpi$, $\omega$, $d$, $n$, and $n_s$ in the $D=1$ case are then replaced by elementwise conditions for $\varpi_i$, $\omega_i$, $d_i$, $n_i$, and $n_{s,i}$ at $i \in \mathcal{I}_D$.

Fig.\,\ref{fig:9} illustrates the potential of the AP-A algorithm with the help of two $D=2$ cases. Fig.\,\ref{fig:9}\,A shows successful demodulation results for a synthetic narrowband fringe pattern ($E_m = 1 \cdot 10^{-2}$, $E_c = 3 \cdot 10^{-2}$). Fig.\,\ref{fig:9}\,B displays high-accuracy demodulation of a wideband signal built of randomly-placed spikes of finite width as $\mathbf{\overline{c}}$ and a Gaussian random field with a rectangular amplitude spectrum as $\mathbf{\overline{m}}$. There, the white grid corresponds to the original $\mathbf{\overline{m}}$, while the color surface represents its estimate $\mathbf{\hat{\overline{m}}}$ ($E_m = 4 \cdot 10^{-3}$, $E_c = 1 \cdot 10^{-2}$).

The ability of the AP approach to deal with wideband signals allows it to cover a wider range of practically relevant situations. Among examples are nonlinear ultrasound imaging \cite{Wachinger2012, Duck2002}, speech processing \cite{Sell2010b,Aragonda2015}, and complicated cases of optical interference/diffraction setups \cite{Sanchez2010}. Moreover, it can also be of great use in time-critical imaging settings by providing high modulator estimation accuracy at low sampling rates of the signal (see, e.g., \cite{Nadeau2014,Zhou2015}).

The minimum number of sample points necessary to cover simultaneously for appropriate demodulation increases exponentially with $D$. Therefore, the computational advantage of the AP over the PAD and LDC demodulation approaches is even more pronounced in higher dimensions. In fact, if evaluated by using the FFT method, $\overline{\mathbf{F}}$ features an $\mathcal{O}\big(n\log n\big)$ computational time complexity. Hence, the time complexity of the AP algorithms is defined by the total number of sample points of the signal irrespective of its dimensionality.

\subsection{Generalized modulators and nonuniform sampling \label{subsec:DifferentMod}}

The demodulation approach formulated in the present work builds on the assumption that modulators are nonnegative elements of a low-pass DFT subspace of $\mathbb{T}^n_D$. However, as follows from the convergence proofs in Suppl.\,Mat.\,F, all of the introduced AP algorithms are bound to converge to an $\mathbf{\hat{m}} \in \Mwmbar$ and a $\mathbf{\hat{c}} \in \Cdbar$ independent of the origin of the linear subspace behind $\Mwmbar$. This naturally raises the question of whether the AP algorithms could recover true $\mathbf{m}$ and $\mathbf{c}$ under the generalized subspace assumption. Our preliminary experiments suggest a positive answer but subject to extra recovery conditions specific to a subspace of choice.

For example, consider a subset of $2\omega-1$ randomly chosen basis vectors of the DFT over $\mathbb{R}^n$. Denote the corresponding space as $\mathcal{F}_{\omega}$. It can be shown by example that a system resulting from random subsampling of the aforementioned vectors at $2\omega-1$ time points may be linearly dependent. If so, it then follows from the proof of \PropRef{prop:Recovery1_} that, in contrast to an $\mathbf{m} \in \Swm$, full recovery of an $\mathbf{m} \in \mathcal{F}_\omega$ necessitates more than $2\omega-1$ true sample points.

The problem of formulating modulator recovery conditions for different linear subspaces sets directions for future studies. If successful, these extensions would allow to:
\begin{enumerate}
\item broaden the concept of the amplitude modulator beyond the low-pass DFT signals,
\item loosen the constraints on the positioning of the $|c_i|=1$ sample points for recoverable carriers whenever a more compact representation of modulators is available,
\item encompass nonuniform sampling.
\end{enumerate}

While the above points are yet to be developed, the results of the present work already provide a strategy for an arbitrarily-accurate nonuniform sampling. Indeed, for any time grid $\mathbf{\tilde{t}} \in \mathbb{R}^{\tilde{n}}$, we can find a uniform grid $\mathbf{t} \in \mathbb{R}^n$ such that, for every $j \in \mathcal{I}_{\tilde{n}}$, there exists an $i \in \In$ with $|\tilde{t}_j - t_i|$ being arbitrarily small. We can then interpolate the original data $\mathbf{\tilde{s}} \in \mathbb{R}^{\tilde{n}}$ on the uniform grid $\mathbf{t}$ by
\begin{align}
s_i =\begin{cases} \tilde{s}_j, & \mbox{if $|\tilde{t}_j - t_i|=\min[|\tilde{t}_j \cdot \mathbf{1}-\mathbf{t}|]$}\\ 0, & \mbox{otherwise} \end{cases}, \quad~ i \in \In,
\label{eq:NonUni1}
\end{align}
to obtain an $\mathbf{s} \in \mathbb{R}^n$. The bandwidth constraint on $\mathbf{m}$ implies that all components of $\mathbf{s}$ corresponding to the true sample points of $\mathbf{\tilde{s}}$ are desirably close to the true sample points of the original signal if $n$ is large enough. Then, \PropRef{prop:Recovery4_} assures that modulator-carrier recovery is possible via \eqref{eq:Recovery2_1_} under the conditions discussed in Section~\ref{subsec:Problem-Formulation-C} for uniformly sampled signals. The described strategy requires increasing the effective dimensionality of the signal. However, this may still be more efficient than evaluating metric projections onto subspaces spanned by arbitrary nonuniform sampling basis vectors, which are not orthogonal in general.

\section{Conclusion \label{sec:Conclusion}}

In this paper, we have introduced a new approach to amplitude demodulation of arbitrary-bandwidth signals. We framed demodulation as a problem of modulator recovery from an unlabeled mix of its true and corrupted sample points. Taking this view, we showed that high-accuracy demodulation can be achieved via exact or approximate norm minimization of the modulator for a wide range of relevant signals. We formulated tailor-made alternating projection algorithms to achieve that in practice and tested them in a series of numerical experiments.

The generality and numerical efficiency of the new approach make it a preferred choice in many situations. In the context of narrowband signals, the new method outperforms the classical algorithms in terms of robustness to data distortions and compatibility with nonuniform sampling. When considering the demodulation of wideband signals, it surpasses the current state-of-the-art techniques in terms of computational efficiency by up to many orders of magnitude. Such performance enables practical applications of amplitude demodulation in previously inaccessible settings. Specifically, online and large-scale offline demodulation of wideband signals, signals in higher dimensions, and poorly-sampled signals become practically feasible. The algorithms underlying the new approach are simple and easy to implement on a computer.\footnote{The computer code for AP demodulation will be available at \href{https://github.com/mgabriel-lt/ap-demodulation}{\textit{https://github.com/mgabriel-lt/ap-demodulation}}.}


\section*{Acknowledgment}
\addcontentsline{toc}{section}{Acknowledgment}

The author thanks his colleagues K.~Husz\'{a}r and G.~Tka\v{c}ik for valuable discussions and comments on the manuscript.



\bibliographystyle{IEEEtran}
\bibliography{ms}

\end{document}


%
\title{{\Huge Supplementary Material for}\\\vspace{20pt}\textit{``Fast and Accurate Amplitude Demodulation\\of Wideband Signals''}\vspace{25pt}}
%
%
%

\author{Mantas~Gabrielaitis \\\texttt{\normalsize mantas.gabrielaitis@\{\href{mailto:mantas.gabrielaitis@ist.ac.at}{ist.ac.at},\,\href{mailto:mantas.gabrielaitis@gmail.com}{gmail.com}\}}}

%
%

\markboth{ACCEPTED FOR PUBLICATION IN IEEE TRANSACTIONS ON SIGNAL PROCESSING}%
{Supplementary Material for \textit{``Fast and Accurate Amplitude Demodulation\\of Wideband Signals''}}
%




\renewcommand\contentsname{{\large Contents}}

\maketitle

\pdfbookmark[section]{\contentsname}{toc}

\vspace{12pt}

\tableofcontents

\newpage

\markboth{Overview}%
{Overview}

\section*{\textbf{Overview}}
\addcontentsline{toc}{section}{Overview}
This document serves as a source of additional information to support the ideas and results introduced in the accompanying paper \textit{``Fast and Accurate Amplitude Demodulation of Wideband Signals.''}

Virtually, the material provided in this supplement can be divided into five blocks comprising, respectively, \textbf{Sections} \textbf{\ref*{sec:SMRecovery1}\,--\,\ref*{sec:SMRecovery3}}, \textbf{\ref*{sec:SMTheory0}\,--\,\ref*{sec:SMTheory3}}, \textbf{\ref*{sec:SMCarriers}\,--\,\ref*{sec:SMSynthetic}}, \textbf{\ref*{sec:SMPerformance}\,--\,\ref*{sec:SMProgrammatic}}, and \textbf{\ref*{sec:SMCarrierEstimates}\,--\,\ref*{sec:SMSpeech}}:
\begin{itemize}
\item The \textbf{\ref*{sec:SMRecovery1}\,--\,\ref*{sec:SMRecovery3}} block provides proofs of modulator recovery conditions.
\item The \textbf{\ref*{sec:SMTheory0}\,--\,\ref*{sec:SMTheory3}} block is concerned with mathematical aspects of the alternating projection (AP) algorithms of amplitude demodulation.
\item The \textbf{\ref*{sec:SMCarriers}\,--\,\ref*{sec:SMSynthetic}} block reviews the main types of amplitude-modulated wideband signals found in practice and defines synthetic modulators and carriers used for testing purposes in the present work.
\item The \textbf{\ref*{sec:SMPerformance}\,--\,\ref*{sec:SMProgrammatic}} block presents details on the numerical implementation of the AP and other relevant demodulation methods on a computer, as well as their benchmarking configurations.
\item The \textbf{\ref*{sec:SMCarrierEstimates}\,--\,\ref*{sec:SMSpeech}} block contains auxiliary simulation results and their discussion.
\end{itemize}

\vspace*{7pt}

\noindent A summary of each of the sections follows next to ease navigation through this document.

\vspace*{7pt}
\SecRef{sec:SMRecovery1}~establishes several auxiliary results that are exploited in the modulator recovery proofs next.

\SecRef{sec:SMRecovery2}~provides proofs for the modulator recovery conditions introduced in Section~II-C of the main text.

\SecRef{sec:SMRecovery3}~discusses the results and implementation of the numerical experiments performed to extend the \hspace*{58pt} modulator recovery conditions.

\vspace*{15pt}
\SecRef{sec:SMTheory0}~introduces the basic concepts of mathematical analysis necessary for the formulation and study of \hspace*{61pt}the properties of AP algorithms.

\SecRef{sec:SMTheory1}~formulates and proves relevant properties of the constraint sets of modulator estimates and defines \hspace*{60pt}operators that implement metric projections onto the modulator constraint sets used in this work.

\SecRef{sec:SMTheory2}~provides the convergence proofs of the AP algorithms introduced in the main text.

\SecRef{sec:SMTheory3}~derives a lower bound on the convergence error.

\vspace*{15pt}
\SecRef{sec:SMCarriers}~reviews the main types of amplitude-modulated wideband signals found in practice.

\SecRef{sec:SMSynthetic}~defines the modulators and carriers of synthetic signals used to test amplitude demodulation \hspace*{62pt}algorithms in the present work.

\vspace*{15pt}
\SecRef{sec:SMPerformance}~lists configurations of the execution control parameters employed for the performance analysis of \hspace*{59pt}the AP and LDC algorithms of amplitude demodulation.

\SecRef{sec:SMProgrammatic}~provides information on the implementation and execution of the demodulation algorithms on a \hspace*{62pt}computer that we used.

\pagebreak
\SecRef{sec:SMCarrierEstimates}~overviews the results of demodulation algorithm performance tests in terms of carrier estimates.

\SecRef{sec:SMConvergence}~presents additional simulation results on the dependence of convergence rates of the AP algorithms \hspace*{62pt}on the signal length.

\SecRef{sec:SMRedemodulation}~introduces simulation results of repetitive demodulation of carrier and modulator estimates obtained \hspace*{62pt}using the AP-A algorithm.

\SecRef{sec:SMSpeech}~presents the results of additional simulations that demonstrate the suitability  and consistency of \hspace*{63pt}the AP approach to demodulate wideband speech signals.

\newpage
\phantomsection

\markboth{RECOVERY CONDITIONS}%
{RECOVERY CONDITIONS}

\hspace{-11pt}\textbf{RECOVERY CONDITIONS}
\addcontentsline{toc}{section}{\textbf{RECOVERY CONDITIONS}}

\section{\textbf{Auxiliary Proofs} \label{sec:SMRecovery1}} 

Here, we establish two important properties of the unitary DFT basis vectors that are repetitively used in the proofs of propositions about the modulator recovery conditions in the next section.


\begin{proposition}
Consider a subset of DFT basis vectors $\{\mathbf{f}^{(k)}\}_{k \in \mathcal{I}_n^{\omega^{\sast}}}$. Assume a set of arbitrarily chosen $n_s$ sample points encoded by components of a vector $\mathbf{r} \in \mathbb{N}_+^{n_s}$, and introduce a linear transform $\Lr$ that maps every $\mathbf{x} \in \mathbb{R}^n$ to $[x_{r_1}, x_{r_2}, \ldots, x_{r_{n_s}}]^\mathrm{T} \in \mathbb{R}^{n_s}$. Then, a set $\{\Lr \mathbf{f}^{(k)}\}_{k \in \mathcal{I}_n^{\omega^{\sast}}}$ is linearly independent if and only if $n_s \geq 2\omega^{*}-1$.
\label{prop:Recovery01}
\end{proposition}

\begin{proof}
$ $

\underline{Necessity.} If $n_s < 2\omega^*-1$, the set of vectors $\{\Lr \mathbf{f}^{(k)}\}_{k\in\mathcal{I}_n^{\omega^{\sast}}}$ is linearly dependent because the number of linearly independent vectors cannot be higher than the number of components they consists of.

\underline{Sufficiency.} First, consider the case of $n_s = 2\omega^*-1$. Then, $\{\Lr\mathbf{f}^{(k)}\}_{k\in\mathcal{I}_n^{\omega^{\sast}}}$ is linearly independent if and only if the determinant of a matrix formed by concatenating all vectors from this set in an arbitrary order is not equal to zero (see, e.g., \cite[p.~13]{Zhang2011}). To show that the latter condition is satisfied in our case, define a matrix
\begin{equation}
\mathbf{M} = \Lr [\mathbf{f}^{(n-\omega^{\sast}+2)}, \mathbf{f}^{(n-\omega^{\sast}+3)}, \ldots, \mathbf{f}^{(n)}, \mathbf{f}^{(1)}, \mathbf{f}^{(2)}, \ldots, \mathbf{f}^{(\omega^{\sast}-1)}, \mathbf{f}^{(\omega^{\sast})}].
\label{eq:Recovery01_1}
\end{equation}
Taking into account that $\Lr \mathbf{f}^{(k)}$ can be written as $[(z^{r_1})^{k-1},(z^{r_2})^{k-1}, \ldots, (z^{r_{n_s}})^{k-1}]^\mathrm{T} / \sqrt{n}$, with $z=e^{\imath2\pi/n}$, and that $(z^{r_1})^{k}=(z^{r_1})^{k\,\mathrm{mod}\,n}$, $\mathbf{M}$ can be expressed as a product of a diagonal matrix and a Vandermonde matrix:
\begin{equation}
\mathbf{M} = \mathrm{diag}[\sqrt{n} \cdot \Lr \mathbf{f}^{(n-\omega^{\sast}+2)}] \, [\Lr \mathbf{f}^{(1)}, \Lr \mathbf{f}^{(2)}, \ldots, \Lr \mathbf{f}^{(2\omega^{\sast}-1)}].
\label{eq:Recovery01_2}
\end{equation}
Thus,
\begin{equation}
\begin{aligned}
\det\mathbf{M} &= \sqrt{n} \cdot \det\mathrm{diag}[\Lr \mathbf{f}^{(n-\omega^{\sast}+2)}] \cdot \det[\Lr \mathbf{f}^{(1)}, \Lr \mathbf{f}^{(2)}, \ldots, \Lr \mathbf{f}^{(2\omega^{\sast}-1)}]\\
&= \sqrt{n} \cdot \prod_{i=1}^{2\omega^{\sast}-1} \big( \Lr \mathbf{f}^{(n-\omega^{\sast}+2)} \big)_i \cdot \prod_{i=1}^{2\omega^{\sast}-1} \prod_{j=1}^{i-1} \big(\big(\Lr \mathbf{f}^{(2)})_i - (\Lr \mathbf{f}^{(2)} \big)_j \big)\\
&= \frac{1}{\sqrt{n}} \cdot \prod_{i=1}^{2\omega^{\sast}-1} \underbrace{e^{\imath2\pi(n-\omega^{\sast}+1)r_i/n}}_{\neq 0} \cdot \prod_{i=2}^{2\omega^{\sast}-1} \prod_{j=1}^{i-1} \underbrace{\big( e^{\imath2\pi r_i/n} - e^{\imath2\pi r_j/n} \big)}_{\neq 0}\\
&\neq 0,
\end{aligned}
\label{eq:Recovery01_3}
\end{equation}
which implies that the set $\{\Lr \mathbf{f}^{(k)}\}_{k\in\mathcal{I}_n^{\omega^{\sast}}}$ is linearly independent for $n_s=2\omega^*-1$. When writing the second equality above, we used the expression of the determinant of a Vandermonde matrix (see, e.g., \cite[p.~143]{Zhang2011}).

It follows from the definition of linear independence that extending each vector in the set by additional components cannot change the set from linearly independent to linearly dependent. Hence, the linear independence of $\{\Lr \mathbf{f}^{(k)}\}_{k\in\mathcal{I}_n^{\omega^{\sast}}}$ for $n_s=2\omega^*-1$ implies the linear independence of $\{\Lr \mathbf{f}^{(k)}\}_{k\in\mathcal{I}_n^{\omega^{\sast}}}$ for $n_s>2\omega^*-1$.
\end{proof}

\begin{remark}
The fact that $\{\Lr \mathbf{f}^{(k)}\}_{k \in \mathcal{I}_n^{\omega^{\sast}}}$ is linearly independent for $n_s \geq 2\omega^*-1$ means that the system of linear equations $\mathbf{L}_\mathbf{r}\mathbf{x} = \sum_{k\in\mathcal{I}_n^{\omega^{\sast}}} \big( a_k \cdot \Lr \mathbf{f}^{(k)} \big)$ has a unique solution for all $\mathbf{x} \in \Swmast$ if $n_s \geq 2\omega^*-1$. In other words, every $\mathbf{x} \in \Swmast$ can be recovered from its known sample points then.
\end{remark}


\begin{proposition}
Consider some $\mathbf{x}^{(1)} \in \Swmast$ and $\mathbf{x}^{(2)} \in \Swmast$. Assume a set of regularly spaced $n_s \geq 2\omega^*-1$ sample points encoded by components of a vector $\mathbf{r} \in \mathbb{N}_+^{n_s}$ such that $r_{i+1}-r_{i} = n/n_s$ for every $i \in \mathcal{I}_{n_s}$. Then,
\begin{equation}
\| \mathbf{x}^{(1)} \|_2 / \| \mathbf{x}^{(2)} \|_2 = \| \Lr \mathbf{x}^{(1)} \|_2 / \| \Lr \mathbf{x}^{(2)} \|_2.
\label{eq:Recovery02_1}
\end{equation}
\label{prop:Recovery02}
\end{proposition}

\begin{proof}
By the definition of $\Swmast$ (see Section~II-A in the main text),
\begin{equation}
\mathbf{x} = \sum_{k\in\mathcal{I}_n^{\omega^{\sast}}} \big( a_k \cdot \mathbf{f}^{(k)} \big), \quad \mathbf{x} \in \Swmast.
\label{eq:Recovery02_2}
\end{equation}
Moreover, it follows from the unitary property of the DFT matrix, $\langle \mathbf{f}^{(j)},\mathbf{f}^{(k)}\rangle=\delta_{j,k}$,\footnote{Here, and in the sequel, $\delta_{i,j}$ denotes the Kronecker delta.} that
\begin{equation}
\|\mathbf{x}\|_2^2 = \sum_{k\in\mathcal{I}_n^{\omega^{\sast}}} |a_k|^2.
\label{eq:Recovery02_3}
\end{equation}
Applying $\Lr$ to both sides of \eqref{eq:Recovery02_2} yields
\begin{equation}
\mathbf{L}_\mathbf{r}\mathbf{x} = \sum_{k\in\mathcal{I}_n^{\omega^{\sast}}} \big( a_k \cdot \Lr \mathbf{f}^{(k)} \big), \quad \mathbf{x} \in \Swmast,
\label{eq:Recovery02_4}
\end{equation}
where, taking into account that $r_{i+1}-r_{i} = n/n_s$,
\begin{equation}
\begin{aligned}
\Lr \mathbf{f}^{(k)} &= [e^{\imath 2\pi \cdot r_1 \cdot(k-1)/n}, e^{\imath 2\pi \cdot r_2 \cdot(k-1)/n}, \ldots, e^{\imath 2\pi \cdot r_{n_s} \cdot(k-1)/n}]^\mathrm{T} / \sqrt{n}\\
&= [1, e^{\imath 2\pi \cdot 1 \cdot(k-1)/n_s}, \ldots, e^{\imath 2\pi \cdot (n_s-1) \cdot(k-1)/n_s}]^\mathrm{T} / \sqrt{n_s} \cdot \big(e^{\imath 2\pi \cdot r_1 \cdot(k-1)/n} \cdot \sqrt{n_s/n}\big).
\end{aligned}
\label{eq:Recovery02_5}
\end{equation}
Note that $\Lr \mathbf{f}^{(k)}$ is the $k$-th column of the unitary $n_s \times n_s$ DFT matrix multiplied by a coefficient whose absolute value is equal to $\sqrt{n_s/n}$. Therefore, analogously to \eqref{eq:Recovery02_3},
\begin{equation}
\|\mathbf{L}_\mathbf{r}\mathbf{x}\|_2^2 = (n_s/n) \cdot \sum_{k\in\mathcal{I}_n^{\omega^{\sast}}} |a_k|^2,
\label{eq:Recovery02_6}
\end{equation}
as long as $n_s \geq 2\omega^*-1$.\footnote{If $n_s<2\omega^{\scalebox{0.8}{$*$}}-1$, some of the vectors $\Lr \mathbf{f}^{(k)}$ and $\Lr \mathbf{f}^{(j)}$ are identical for $k \neq j$, and hence, \eqref{eq:Recovery02_6} does not apply.} Consequently, $\|\mathbf{L}_\mathbf{r}\mathbf{x}\|_2 = \sqrt{n_s/n} \cdot \|\mathbf{x}\|_2$ for $\mathbf{x} \in \Swmast$, which implies \eqref{eq:Recovery02_1}.
\end{proof}

\section{\textbf{Recovery Proofs} \label{sec:SMRecovery2}}

In this section, we prove the modulator recovery conditions stated in the main text in the form of \textit{Propositions~II.1\,--\,II.4}. We repeat the original assertions from the main text for the sake of convenience. 


%
\begin{proposition_II_1}
For almost every $\mathbf{m} \in \Mwm$, $\mathbf{\hat{m}} = \mathbf{m}$ only if $\varpi \geq \omega$, and \mbox{$\mathbf{c} \in \Cd$} with $n_s \equiv \sum_{i=1}^nI_{\{1\}}(|c_i|) \geq \varpi+\omega-1 \implies d \leq n-(\varpi+\omega-2)$.\footnote{Here $\mathbf{\hat{m}}$ is as defined by (6) in the main text.}$^,\,$\footnote{As $d \leq n-(\varpi+\omega-2)$ is implied by $\sum_{i=1}^nI_{\{1\}}(|c_i|) \geq \varpi+\omega-1$, we do not refer to it explicitly in this proof.}
\end{proposition_II_1}

\begin{proof}
We prove the proposition by showing that, almost everywhere in $\Mwm$, $\mathbf{m} \neq \mathbf{\hat{m}}$ if $\mathbf{c} \notin \Cd$, or $\varpi < \omega$, or $n_s < \varpi + \omega - 1$. For the sake of convenience, we restate the definition of $\mathbf{\hat{m}}$ here:
\begin{equation}
\mathbf{\hat{m}} = \underset{\mathbf{x} \in \Sgeqss \cap \Sw}{\arg\min} \| \mathbf{x} \|_2.
\label{eq:Recovery2_1}
\end{equation}

If $\mathbf{c} \notin \Cd$, it means that either $|c_i|<1$, for every $i \in \mathcal{I}_n$, or there exists at least one $i \in \mathcal{I}_n$ such that $|c_i|>1$. In the former case, $|s_i|/m_i<1$ for every $i \in \mathcal{I}_n$. Hence, for $\alpha = \max \{|s_i|/m_i\}_{i \in \mathcal{I}_n}$, $\alpha \cdot \mathbf{m}$ belongs to $\Sgeqs \cap \Sw$ but has a smaller norm than $\mathbf{m}$, i.e., $\mathbf{m} \neq \mathbf{\hat{m}}$. In the latter case,
\begin{equation}
\bigg(\underset{\mathbf{x} \in \Sgeqss \cap \Sw}{\arg\min} \| \mathbf{x} \|_2\bigg)_i > m_i
\label{eq:Recovery2_10}
\end{equation}
for all $i$ corresponding to $|c_i|>1$ because $|s_i|>m_i$ then. Thus, $\mathbf{m} = \mathbf{\hat{m}}$ does not apply either. Therefore, $\mathbf{m} = \mathbf{\hat{m}}$ holds only if $\mathbf{c} \in \mathcal{C}_d$.

If $\varpi < \omega$, then the subset of modulators for which $\mathbf{m} = \mathbf{\hat{m}}$ is valid has the cardinality of $\mathbb{R}^{2\varpi-1}$, and hence, has zero volume in $\Mwm$, whose cardinality is that of $\mathbb{R}^{2\omega-1}$.

Next, assume that $\mathbf{c} \in \mathcal{C}_d$, and $\varpi \geq \omega$, but $n_s < \varpi+\omega-1$. Let us represent indexes of all sample points corresponding to $|c_i|=1$ by a vector $\mathbf{r} \in \mathbb{N}_+^{n_s}$. Then, analogously to \eqref{eq:Recovery02_4}, we have 
\begin{equation}
\Lr \mathbf{m} = \sum_{k\in\Inw} \big( a_k \cdot \Lr \mathbf{f}^{(k)} \big).
\label{eq:Recovery2_11}
\end{equation}
For the sake of convenience, let us redefine 
\begin{align}
\mathbf{f}^{(k)} = \begin{cases} \bm{\varphi}^{(1)}, & \mbox{$k=1$}\\ (\bm{\varphi}^{(2k-2)}+\imath \cdot \bm{\varphi}^{(2k-1)}) / \sqrt{2}, & \mbox{$2 \leq k \leq \varpi$}\\ (\bm{\varphi}^{(2(n-k)+2)}+\imath \cdot \bm{\varphi}^{(2(n-k)+3)}) / \sqrt{2}, & \mbox{$n-\varpi+2 \leq k=n$} \end{cases},
\label{eq:Recovery2_12}
\end{align}
and
\begin{align}
a_k = \begin{cases} \alpha_1, & \mbox{$k=1$}\\ (\alpha_{2k-2}+\imath \cdot \alpha_{2k-1}) / \sqrt{2}, & \mbox{$2 \leq k \leq \varpi$}\\ (\alpha_{2(n-k)+2}+\imath \cdot \alpha_{2(n-k)+3}) / \sqrt{2}, & \mbox{$n-\varpi+2 \leq k=n$} \end{cases}.
\label{eq:Recovery2_13}
\end{align}
Then, \eqref{eq:Recovery2_11} turns into
\begin{equation}
\Lr \mathbf{m} = \sum_{i=1}^{2\varpi-1} \big( \alpha_i \cdot \Lr \bm{\varphi}^{(i)} \big),
\label{eq:Recovery2_14}
\end{equation}

According to \PropRef{prop:Recovery01}, a set of $n_s$ vectors $\{ \Lr \mathbf{f}^{(1)}, \ldots, \Lr \mathbf{f}^{(\lceil (n_s+1)/2 \rceil}, \Lr \mathbf{f}^{(n-\lfloor (n_s-3)/2 \rfloor}, \ldots, \Lr \mathbf{f}_n\}$ is linearly independent. Hence, by \eqref{eq:Recovery2_12}, the same applies to $\{\Lr \bm{\varphi}^{(i)}\}_{i=1}^{n_s}$. Consequently, \eqref{eq:Recovery2_14} is an underdetermined system of linear equations defined by a full-rank matrix. We know from linear algebra that a general solution of such system is expressed as a sum of its separate solution $\bm{\alpha}^{(0)}$ and a solution of 
\begin{equation}
\mathbf{0} = \sum_{i=1}^{2\varpi-1} \big( \alpha_i \cdot \Lr \bm{\varphi}_i \big),
\label{eq:Recovery2_15}
\end{equation}
Solutions of \eqref{eq:Recovery2_15} form a $(2\varpi-1-n_s)$-dimensional subspace of $\mathbb{R}^{2\varpi-1}$. Thus, we can express the general solution of \eqref{eq:Recovery2_14} as
\begin{equation}
\bm{\alpha} = \bm{\alpha}^{(0)} + \sum_{i=1}^{2\varpi-1-n_s} z_i \bm{\rho}^{(i)}, \quad \mathbf{z} \in \mathbb{R}^{2\varpi-1-n_s},
\label{eq:Recovery2_16}
\end{equation}
where $\{\bm{\rho}^{(i)}\}_{i=1}^{2\varpi-1-n_s}$ is an orthonormal basis of the space of solutions of \eqref{eq:Recovery2_15}. Taking into account \eqref{eq:Recovery2_12}\,--\,\eqref{eq:Recovery2_13} as well as the linear independence of $\{ \mathbf{f}^{(k)}\}_{k \in \Inw}$ and $\{\bm{\rho}^{(i)}\}_{i=1}^{2\varpi-1-n_s}$, \eqref{eq:Recovery2_16} together with \eqref{eq:Recovery02_4} define a linear injective function that maps from $\mathbb{R}^{2\varpi-1-n_s}$ to $\Sw$:
\begin{equation}
f(\mathbf{z}) = \sum_{j=1}^{2\varpi-1} \bigg( \bm{\alpha}^{(0)} + \sum_{i=1}^{2\varpi-1-n_s} z_i \bm{\rho}^{(i)} \bigg)_j \, \bm{\varphi}^{(j)}, \quad \mathbf{z} \in \mathbb{R}^{2\varpi-1-n_s}.
\label{eq:Recovery2_17}
\end{equation}
The image of $f(\mathbf{z})$ is a subset of those elements of $\Sw$ that coincide with the true modulator $\mathbf{m}$ at entries $\mathbf{r} \in \mathbb{N}_+^{n_s}$. The injective nature of this function guarantees the existence of a unique $\mathbf{z}_\mathbf{m} \in \mathbb{R}^{2\varpi-1-n_s}$ such that $\mathbf{m} = f(\mathbf{z}_\mathbf{m})$.

Now, assume an $\mathbf{m} \in \Mwm^+ = \{\mathbf{x} \in \Mwm: m_i>0, i \in \mathcal{I}_n \}$, i.e., any feasible modulator whose all entries are strictly positive. Define an $\epsilon = \min \{m_i-|s_i|: (|c_i| < 1) \land (i \in \mathcal{I}_n) \}$. $\epsilon$ exists and is positive because $|c_i| = 1$ only for $n_s<n$ out of $n$ components of $\mathbf{c}$ by the assumption of the proposition, and $m_i-|s_i| = m_i(1-c_i)$. The linearity of $f(\mathbf{z})$ implies its continuity at every point of its domain. Hence, there exists an $\eta$ such that $\|f(\mathbf{z})-f(\mathbf{z}_\mathbf{m}) \|_2 < \epsilon$ whenever $\mathbf{z} \in \mathcal{H}_{\mathbf{z}_\mathbf{m},\eta} = \{\mathbf{z} \in \mathbb{R}^{2\varpi-1-n_s}: \|\mathbf{z}-\mathbf{z}_\mathbf{m}\|_2 < \eta \} $. On the other hand, $\|f(\mathbf{z})-f(\mathbf{z}_\mathbf{m})\|_2 = \sqrt{\sum_{i=1}^n[(f(\mathbf{z}))_i - m_i]^2} < \epsilon$ implies $|(f(\mathbf{z}))_i - m_i| < \epsilon$, and hence $(f(\mathbf{z}))_i > |s_i|$, for every $i \in \mathcal{I}_n$. Consequently, 
\begin{equation}
f(\mathbf{z}) \in (\Sgeqs \cap \Sw ), \quad \mathbf{z} \in \mathcal{H}_{\mathbf{z}_\mathbf{m},\eta}.
\label{eq:Recovery2_18}
\end{equation}
Furthermore, if $\mathbf{m} = \mathbf{\hat{m}}$, then $\|\mathbf{m}\|_2^2 < \| f(\mathbf{z}) \|_2^2$ for every $\mathbf{z} \in (\mathcal{H}_{\mathbf{z}_\mathbf{m},\eta} \setminus \mathbf{z}_\mathbf{m})$, i.e., $\mathbf{z}_\mathbf{m}$ is a strict local minimum point of
\begin{equation}
\begin{aligned}
\| f(\mathbf{z}) \|_2^2 &= \Bigg\| \sum_{j=1}^{2\varpi-1} \bigg( \bm{\alpha}^{(0)} + \sum_{i=1}^{2\varpi-1-n_s} z_i \bm{\rho}^{(i)} \bigg)_j \, \bm{\varphi}^{(j)} \Bigg\|_2^2 \\
&= \Bigg\| \sum_{j=1}^{2\varpi-1} \bigg( \bm{\alpha}^{(0)} + \sum_{i=1}^{2\varpi-1-n_s} z_i \bm{\rho}^{(i)} \bigg) \Bigg\|_2^2 \\
&= \| \bm{\alpha}^{(0)} \|_2^2 + \sum_{i=1}^{2\varpi-1-n_s} \big(z_i^2 + 2 \cdot z_i \cdot  \langle \bm{\alpha}^{(0)}, \bm{\rho}^{(i)} \rangle \big).
\end{aligned}
\label{eq:Recovery2_19}
\end{equation}
$\| f(\mathbf{z}) \|_2^2$ is a continuous, differentiable function with a positive-definite Hessian: $\partial^2 \| f(\mathbf{z}) \|_2^2 /\partial z_i \partial z_j = \delta_{i,j}$. Thus, it has a unique strict local minimum point defined by $\partial \| f(\mathbf{z}) \|_2^2 / \partial z_i=0$: \footnote{This strict local minimum point is the only local minimum point of this function.}
\begin{equation}
z_i^\dagger = - \langle \bm{\alpha}^{(0)}, \bm{\rho}^{(i)}\rangle, \quad i \in \mathcal{I}_{2\varpi-1-n_s}.
\label{eq:Recovery2_20}
\end{equation}

Remember that, without loss of generality, $\bm{\alpha}^{(0)}$ is a particular solution of \eqref{eq:Recovery2_14}. $\bm{\alpha}^{(0)}$ corresponding to $\mathbf{m}$, i.e., $\bm{\alpha}^{(0)} \equiv (\alpha_1^{(0)}, \alpha_2^{(0)}, \ldots, \alpha_{2\omega-1},0,\ldots,0)^\mathrm{T}$ with $\mathbf{m} = \sum_{i=1}^{2\omega-1}\alpha_i^{(0)}\bm{\varphi}^{(i)}$, is exactly such a solution. In this case, as follows from \eqref{eq:Recovery2_17}, $\mathbf{m} = f(\mathbf{z}^\dagger)$ if and only if $\mathbf{z}^\dagger = \mathbf{0}$.\footnote{Note that $\{ \bm{\rho}_i \}_{i=1}^{2\varpi-1-n_s}$ is linearly independent.} According to \eqref{eq:Recovery2_20}, that is equivalent to requiring $\bm{\alpha}^{(0)}$ to be a solution of the following homogeneous system of linear equations:
\begin{equation}
\langle \bm{\alpha}^{(0)}, \bm{\rho}^{(i)}\rangle = 0, \quad i \in \mathcal{I}_{2\varpi-1-n_s}.
\label{eq:Recovery2_21}
\end{equation}
Hence, the subset of $\mathcal{M}_{\omega}^+$ to which $\mathbf{m} = \mathbf{\hat{m}}$ applies has the same cardinality as $\mathbb{R}^{D}$, where $D$ is the dimension of the solution space of \eqref{eq:Recovery2_21}. Taking into account the linear independence of $\{ \bm{\rho}^{(i)} \}_{i=1}^{2\varpi-1-n_s}$, $D$ is equal to the difference between the number of elements of $\bm{\alpha}^{(0)}$ that are not identically equal to zero and the number of equations that are not trivially satisfied by any feasible $\bm{\alpha}^{(0)}$. The latter depends on $\mathbf{c}$, specifically, on the positions of sample points with $|c_i|=1$. To see this, consider two cases:
\begin{itemize}
\item A carrier with equidistantly-spaced true sample points: $|c_{i+(j-1) \cdot d}|=1$ for some $i \in \mathcal{I}_{d}$ and every $j \in \mathcal{I}_{n_s}$, where $d=n/n_s$. Then, it follows from \eqref{eq:Recovery02_5} that some of the elements of the system $\{\Lr \mathbf{f}^{(k)}\}_{k \in \Inw}$ are identical as long as $n_s < 2\varpi-1$.\footnote{Note that $n_s < \varpi + \omega-1$ and $\varpi \geq \omega$ imply $n_s <2\varpi-1$.} Specifically, we have
\begin{align}
\Lr \mathbf{f}^{(k)} = \begin{cases} \Lr \mathbf{f}^{(k+n-n_s)}, & \mbox{$\lfloor (n_s+4)/2 \rfloor \leq k \leq \varpi$}\\ \Lr \mathbf{f}^{(k+n_s-n)}, & \mbox{$n-\varpi+2 \leq k \leq n-\lfloor(n_s-1)/2\rfloor$} \end{cases}.
\label{eq:Recovery2_23}
\end{align}
Equivalently,
\begin{align}
\Lr \bm{\varphi}_{n_s+\chi_{n_s}-i} = (-1)^{\chi_i} \cdot \Lr \bm{\varphi}_{n_s-\chi_{n_s}+2 \chi_i + i}, \qquad 1 \leq i \leq 2\varpi-2+\chi_{n_s}-n_s
\label{eq:Recovery2_24}
\end{align}
and $\Lr \bm{\varphi}_{(n_s+1)} = \mathbf{0}$ when $\chi_{n_s}=0$, where, $\chi_z = z~\textrm{mod}~2$. Consequently,\footnote{Note that, as discussed before, $\{\bm{\varphi}_i\}_{i=1}^{n_s}$ is linearly independent.}
\begin{align}
(\bm{\rho}^{(i)})_j = \begin{cases} 1/\sqrt{2}, & \mbox{$j=n_s+\chi_{n_s}-i$}\\ (-1)^{\chi_i+1}/\sqrt{2}, & \mbox{$j=n_s-\chi_{n_s}+2 \chi_i+i$}\\
0, & \mbox{otherwise} \end{cases}, \qquad 1 \leq i \leq 2\varpi-2+\chi_{n_s}-n_s,
\label{eq:Recovery2_25}
\end{align}
and $\bm{\rho}^{(2\varpi-1-n_s)} = \mathbf{0}$ when $\chi_{n_s}=0$. Moreover, by our choice of $\bm{\alpha}^{(0)}$,
\begin{align}
\alpha_j^{(0)} = 0, \qquad 2\omega \leq j \leq 2\varpi-1.
\label{eq:Recovery2_26}
\end{align}
\eqref{eq:Recovery2_25} and \eqref{eq:Recovery2_26} together imply that \eqref{eq:Recovery2_21} holds for $i \in \mathcal{I}_{n_s-(2\omega-1)}$ independent of $\bm{\alpha}^{(0)}$ corresponding to the chosen $\mathbf{m} \in \Mwm^+$, and that \eqref{eq:Recovery2_21} holds for the remaining $i \in (\mathcal{I}_{2\varpi-1-n_s} \setminus \mathcal{I}_{n_s-(2\omega-1)})$ if and only if
\begin{align}
\alpha_j^{(0)} = 0, \qquad 2(n_s+1-\varpi) \leq j \leq 2\omega-1.
\label{eq:Recovery2_27}
\end{align}
Hence, \eqref{eq:Recovery2_21} applies only to the subset of $\Mwm^+$ with
\begin{equation}
\begin{aligned}
D &= \min \{(2\omega-1), ~(2\omega-1) - 2(\varpi + \omega - 1 - n_s) \} \\
  &= \min \{(2\omega-1), ~2(n_s+1-\varpi)-1 \}.
\label{eq:Recovery2_28}
\end{aligned}
\end{equation}
\eqref{eq:Recovery2_28} implies that $D < 2\omega-1$ as long as $n_s < \varpi + \omega-1$. Then, the subset of all $\mathbf{m} \in \Mwm^+$ that satisfy $\mathbf{m} = \mathbf{\hat{m}}$ has zero volume in $\Mwm^+$, which has the cardinality of $\mathbb{R}^{2\omega-1}$. Moreover, if $n_s < \varpi$, it follows from \eqref{eq:Recovery2_27} that \eqref{eq:Recovery2_21} has only the trivial solution $\bm{\alpha}^{(0)}=0$, which implies $\mathbf{m}=0$. The latter is infeasible in our case, and hence, $\mathbf{m} = \mathbf{\hat{m}}$ does not apply to any $\mathbf{m} \in \Mwm^+$. On the other hand, if $n_s \geq \varpi + \omega - 1$, then all equations of the \eqref{eq:Recovery2_21} system are satisfied independent of the actual feasible $\bm{\alpha}^{(0)}$. Consequently, $D = 2\omega - 1$.

\item An arbitrary $\mathbf{c} \in \Cd$ with unspecified structure. Then, none of the $2\varpi-1-n_s$ equations of the system \eqref{eq:Recovery2_21} are satisfied independent of the actual $\bm{\alpha}^{(0)}$.\footnote{One possible example of this kind is $\mathbf{c}$ with $|c_i|=1$ for $i \in \mathcal{I}_{n_s}$, and $|c_i|<1$ for $i \in (\In \setminus \mathcal{I}_{n_s})$.} Therefore,
\begin{equation}
\begin{aligned}
D &= (2\omega-1)-(2\varpi-1-n_s)\\
 &= n_s - 2(\varpi - \omega).
\end{aligned}
\label{eq:Recovery2_29}
\end{equation}
\eqref{eq:Recovery2_29} is valid only if $(2\omega-1) > (2\varpi-1-n_s)$. Otherwise, \eqref{eq:Recovery2_21} has only the trivial solution $\bm{\alpha}^{(0)} = \mathbf{0}$, which implies $\mathbf{m}=0$. The latter is infeasible in our case, and hence, $\mathbf{m} = \mathbf{\hat{m}}$ does not apply to any $\mathbf{m} \in \Mwm^+$. In any case, $D < (2\omega-1)$ as long as $n_s < (2\varpi-1)$, which follows from the assumption of the proposition that $n_s < (\varpi + \omega - 1)$ and $\varpi \geq \omega$. In fact, $(2\varpi-1) - (\varpi + \omega - 1) = (\varpi - \omega)$.
\end{itemize}

Compared with the general case, a smaller number of necessary sample points with $|c_i|=1$ is achieved in the first example due to the fact that the subset of vectors $\Lr \bm{\varphi}^{(k)}$ with $2\omega \leq k \leq 2\varpi-1$ is linearly dependent and that $\alpha^{(0)}_k$ are identically equal to zero in that range. Specifically, every $\bm{\varphi}^{(k)}$ in the range $2\omega \leq n_s-1+\chi_{n_s}$ is proportional to one of the $\bm{\varphi}^{(k)}$ in the range $n_s+2-\chi_{n_s} \leq 2\varpi-1$. Every such dependence reduces the necessary number of points with $|c_i|=1$ for modulator recovery by one. Further, it follows from \eqref{eq:Recovery02_5}, \eqref{eq:Recovery2_12}, and \PropRef{prop:Recovery01} that the sets $\{\Lr \bm{\varphi}^{(k)}\}_{k=2\omega}^{n_s-1+\chi_{n_s}}$ and $\{\Lr\bm{\varphi}^{(k)}\}_{k=n_s+2-\chi_{n_s}}^{2\varpi-1}$ are linearly independent. Hence, no further linear dependencies among $\{\Lr \bm{\varphi}^{(k)}\}_{k=2\omega}^{2\varpi-1}$ are possible in general. This means that $n_s = \varpi + \omega - 1$ is the absolute minimum of sample points with $|c_i|=1$ necessary for $\mathbf{m} = \mathbf{\hat{m}}$ to hold. The second example above illustrates that this number is surely higher for some $\mathbf{c} \in \Cd$.

The last three paragraphs demonstrate that the subset of $\Mwm^+$ to which $\mathbf{m} = \mathbf{\hat{m}}$ applies has zero volume in $\Mwm^+$ if $n_s < \varpi+\omega-1$. Taking into account that $(\Mwm \setminus \Mwm^+)$ has lower cardinality than $\Mwm^+$ ($\mathbb{R}^{2\omega-2}$ vs. $\mathbb{R}^{2\omega-1}$), we conclude that the set of all $\mathbf{m} \in \Mwm$ that satisfy $\mathbf{m} = \mathbf{\hat{m}}$ has zero volume in $\Mwm$ if $n_s < \varpi+\omega-1$.

\end{proof}


\begin{proposition_II_2}
Consider $\mathbf{m} \in \Mwm$ and $\mathbf{\tilde{c}} \in \mathcal{C}_{\tilde{d}}$ with $|\tilde{c}_i|=1$ for $i \in \mathcal{J}_n \subseteq \mathcal{I}_n$, and $\tilde{c}_i=0$ otherwise. If $\mathbf{\hat{m}} = \mathbf{m}$ holds for the $\mathbf{m}$ and $\mathbf{\tilde{c}}$, then it also holds for every pair made of the same $\mathbf{m}$ and any $\mathbf{c} \in \Cd$ with $d \leq \tilde{d}$ and $|c_i|=1$ for $i \in \mathcal{J}_n$.\footnote{\label{footnote10}Here $\mathbf{\hat{m}}$ is as defined by (6) in the main text.}
\label{prop:Recovery4}
\end{proposition_II_2}

\begin{proof}
Denote $\mathbf{\tilde{s}} = \mathbf{m} \circ \mathbf{\tilde{c}}$ and $\mathbf{s} = \mathbf{m} \circ \mathbf{c}$. Now, note that $|c_i| \geq |\tilde{c}_i|$ by the condition of the proposition, and hence, $|s_i| \geq |\tilde{s}_i|$. Therefore, $\Sgeqs \subseteq \Sgeqst$, which implies $(\Sgeqs \cap \Sw) \subseteq (\Sgeqst \cap \Sw)$, and, consequently, $\underset{\mathbf{x} \in \Sgeqss \cap \Sw}{\arg\min} \| \mathbf{x} \|_2 \geq \underset{\mathbf{x} \in \Sgeqsst \cap \Sw}{\arg\min} \| \mathbf{x} \|_2$. According to the proposition, $\underset{\mathbf{x} \in \Sgeqsst \cap \Sw}{\arg\min} \| \mathbf{x} \|_2 = \mathbf{m}$. On the other hand, by construction, $\mathbf{m} \in (\Sgeqs \cap \Sw)$. Thus, $\underset{\mathbf{x} \in \Sgeqss \cap \Sw}{\arg\min} \| \mathbf{x} \|_2 = \mathbf{m}$.
\end{proof}

\begin{remark}
It can be shown by example that the validity of  $\mathbf{\hat{m}} = \mathbf{m}^{(1)}$ for some $\mathbf{m}^{(1)} \in \Mwm$ and $\mathbf{\tilde{c}} \in \Cd$ with $|\tilde{c}_i|=1$, $i \in \mathcal{J}_n$, does not necessarily imply the validity of $\mathbf{\hat{m}} = \mathbf{m}^{(2)}$ for another $\mathbf{m}^{(2)} \in \Mwm$ and the same $\mathbf{\tilde{c}}$.
\end{remark}


\begin{proposition_II_3}
Assume $\mathbf{m} \in \Mwm$ and $\mathbf{c} \in \Cd$ with $\varpi \geq \omega$. If, additionally, there exist $d \in \mathcal{I}_n$ and $i \in \mathcal{I}_{d}$ such that $n_s \equiv (n/d) \in \mathbb{N}_+$, $n_s \geq \varpi+\omega-1$, and $|c_{i+(j-1)\cdot d}|=1$ for every $j\in \mathcal{I}_{n_s}$, then $\mathbf{\hat{m}} = \mathbf{m}$.$^{\ref{footnote10}}$
\end{proposition_II_3}

\begin{proof}

Here, we distinguish between two cases: $(\varpi+\omega-1) \leq n_s < (2\varpi-1)$ and $n_s \geq (2\varpi-1)$.

If $(\varpi+\omega-1) \leq n_s < (2\varpi-1)$, then it follows from the proof of \textit{Proposition~II.1} that $\mathbf{m}$ has the smallest norm among all elements of the image of $f(\mathbf{z})$ defined by \eqref{eq:Recovery2_17}, i.e, all elements $\mathbf{x} \in \Sw$ that satisfy $\Lr \mathbf{x} = \Lr \mathbf{m}$. Next, consider a $\mathbf{y} \in \Swm$ such that $(\Lr \mathbf{y})_i \geq (\Lr \mathbf{m})_i$ for every $i \in \mathcal{I}_{n_s}$ and $(\Lr \mathbf{y})_i > (\Lr \mathbf{m})_i$ for at least one $i \in \mathcal{I}_{n_s}$. Then, by \eqref{eq:Recovery02_1}, $\|\mathbf{y}\|_2 > \|\mathbf{m}\|_2$. Moreover, using the same argumentation as for $\mathbf{m}$, we see that $\mathbf{y}$ has the smallest norm among all elements $\mathbf{x} \in \Sw$ that satisfy $\Lr \mathbf{x} = \Lr \mathbf{y}$. Hence, $\mathbf{m}$ has smaller norm than any other element $\mathbf{x} \in \Sw$ that satisfies $\Lr \mathbf{x} \geq \Lr \mathbf{m}$. Moreover, $(\Sgeqs \cap \Sw) \subset \Sw$. Thus, we conclude that $\mathbf{\hat{m}} = \mathbf{m}$ holds.

If $n_s \geq (2\varpi-1)$, then it follows from \eqref{eq:Recovery02_1} of \PropRef{prop:Recovery02} that
\begin{equation}
\underset{\mathbf{x} \in \Sgeqss \cap \Sw}{\arg\min} \| \mathbf{x} \|_2 = \underset{\mathbf{x} \in \Sgeqss \cap \Sw}{\arg\min} \| \mathbf{L}_\mathbf{r}\mathbf{x} \|_2.
\label{eq:Recovery2_30}
\end{equation}
Further, the constraint set $\Sgeqs$ implies that
\begin{equation}
\|\mathbf{L}_\mathbf{r}\mathbf{x}\|_2 \geq \|\mathbf{L}_\mathbf{r}|\mathbf{s}|\|_2, \quad \mathbf{x} \in \Sgeqs \cap \Sw.
\label{eq:Recovery2_31}
\end{equation}
On the other hand, $|s_{r_i}| = |m_{r_i} \cdot c_{r_i}| = m_{r_i} \cdot |c_{r_i}| = m_{r_i}$ for $i \in \mathcal{I}_{n_s}$, i.e., $\mathbf{L}_\mathbf{r}|\mathbf{s}|= [m_{r_1}, m_{r_2}, \ldots, m_{r_{n_s}}]^\mathrm{T}$. According to \PropRef{prop:Recovery01}, $\{\Lr \mathbf{f}^{(k)}\}_{k \in \Inw}$ is linearly independent if $n_s \geq 2\varpi-1$. Therefore, \eqref{eq:Recovery02_4}\footnote{Note that $\omega^{\scalebox{0.8}{$*$}}$ in \eqref{eq:Recovery02_2}\,--\,\eqref{eq:Recovery02_4} stands for $\varpi$ here.} has a unique solution, which, by \eqref{eq:Recovery02_2}, means that, among all $\Sw$, there is a unique $\mathbf{x}=\mathbf{m}$ that satisfies $\mathbf{L}_\mathbf{r}\mathbf{x}= [m_{r_1}, m_{r_2}, \ldots, m_{r_{n_s}}]^\mathrm{T} = \mathbf{L}_\mathbf{r}|\mathbf{s}|$. Hence, by \eqref{eq:Recovery2_31}, $\underset{\mathbf{x} \in \Sgeqss \cap \Sw}{\arg\min} \| \mathbf{L}_\mathbf{r}\mathbf{x} \|_2 = \mathbf{m}$, and, by \eqref{eq:Recovery2_30}, $\underset{\mathbf{x} \in \Sgeqss \cap \Sw}{\arg\min} \| \mathbf{x} \|_2 = \mathbf{m}$.
\vspace{9pt}
\end{proof}


\begin{proposition_II_4}
Consider $\mathbf{m} \in \Mwm$ and $\mathbf{c} \in \Sleqo$. Take $n_s \geq 2\varpi-1$ sample points of $\mathbf{s} = \mathbf{m} \circ\mathbf{c}$ whose indexes are defined as entries of any chosen $\mathbf{r} \in \mathbb{N}_+^{n_s}$ with $r_{i+1}-r_{i} = n/n_s$ for every $i \in \mathcal{I}_{n_s}$. Then,\footnote{Here $\mathbf{\hat{m}}$ is as defined by (6) in the main text.}
\label{prop:Recovery3}
\begin{equation}
\textstyle
\|\mathbf{m}-\mathbf{\hat{m}}\|_2 / \|\mathbf{m}\|_2 \leq \sqrt{1-\sum_{i=1}^{n_s} s_{r_i}^2 / \sum_{i=1}^{n_s}m_{r_i}^2}.
\label{eq:Recovery3_1}
\end{equation}
\end{proposition_II_4}

\begin{proof}
For the sake of convenience, we will exploit the linear transformation $\mathbf{L}_\mathbf{r}$, already introduced in \PropRef{prop:Recovery01}, that maps every $\mathbf{x} \in \mathbb{R}^n$ to $[x_{r_1},x_{r_2}, \ldots, x_{r_{n_s}}]^\mathrm{T}$. Then, \eqref{eq:Recovery3_1} can be rewritten as
\begin{equation}
\|\mathbf{m}-\mathbf{\hat{m}}\|_2 / \|\mathbf{m}\|_2 \leq \sqrt{1-\|\mathbf{L}_\mathbf{r}\mathbf{s} \|_2^2 / \|\mathbf{L}_\mathbf{r}\mathbf{m} \|_2^2 }.
\label{eq:Recovery3_2}
\end{equation}
Note that 
\begin{equation}
\|\mathbf{m}\|_2^2-\|\mathbf{\hat{m}}\|_2^2 - \|\mathbf{m}-\mathbf{\hat{m}} \|_2^2 = 2 \cdot \|\mathbf{\hat{m}}\|_2 \cdot (\|\mathbf{m}\|_2- \|\mathbf{\hat{m}}\|_2).
\label{eq:Recovery3_3}
\end{equation}
Next, we have by construction that $\mathbf{m} \in (\Sgeqs \cap \Sw)$. Hence,
\begin{equation}
\|\mathbf{m}\|_2^2 \geq \|\mathbf{\hat{m}}\|_2^2,
\label{eq:Recovery3_4}
\end{equation}
which, together with \eqref{eq:Recovery3_3} implies $\|\mathbf{m}\|_2^2-\|\mathbf{\hat{m}}\|_2^2 - \|\mathbf{m}-\mathbf{\hat{m}} \|_2^2 \geq 0$, i.e., 
\begin{equation}
\|\mathbf{m}-\mathbf{\hat{m}}\|_2 / \|\mathbf{m}\|_2 \leq \sqrt{1-\|\mathbf{\hat{m}} \|_2^2 / \| \mathbf{m} \|_2^2 }.
\label{eq:Recovery3_5}
\end{equation}
On the other hand, by \PropRef{prop:Recovery02} , $\|\mathbf{\hat{m}} \|_2^2 / \| \mathbf{m} \|_2^2 = \|\mathbf{L}_\mathbf{r}\mathbf{\hat{m}} \|_2^2 / \| \mathbf{L}_\mathbf{r}\mathbf{m} \|_2^2$ if $n_s \geq 2\varpi-1$ and $r_{i+1} - r_{i} = n/n_s$ for every $i \in \mathcal{I}_{n_s}$. Thus,
\begin{equation}
\|\mathbf{m}-\mathbf{\hat{m}}\|_2 / \|\mathbf{m}\|_2 \leq \sqrt{1-\|\mathbf{L}_\mathbf{r}\mathbf{\hat{m}} \|_2^2 / \| \mathbf{L}_\mathbf{r}\mathbf{m} \|_2^2 }.
\label{eq:Recovery3_6}
\end{equation}
Finally, $\hat{m}_i \geq |s_i|$ for every $i \in \mathcal{I}_n$ because $\mathbf{m} \in \Sgeqs$, which means that
\begin{equation}
\|\mathbf{L}_\mathbf{r}\mathbf{\hat{m}} \|_2^2 \geq \| \mathbf{L}_\mathbf{r}\mathbf{s} \|_2^2.
\label{eq:Recovery3_7}
\end{equation}
Combining \eqref{eq:Recovery3_6} with \eqref{eq:Recovery3_7} leads to \eqref{eq:Recovery3_2}.
\end{proof}

\begin{remark}
Note that $\|\mathbf{m}\|_2^2 = \|\mathbf{\hat{m}}\|_2^2$ in \eqref{eq:Recovery3_4} if and only if $\mathbf{m}=\mathbf{\hat{m}}$ because $\Sgeqs \cap \Sw$ is convex and $\| \ldots\|_2^2$ is strictly convex. Therefore, the equality in \eqref{eq:Recovery3_1} holds if and only if $\mathbf{m}=\mathbf{\hat{m}}$, i.e., the modulator recovery is exact.
\end{remark}

\section{\textbf{Further Analysis: Numerical Experiments} \label{sec:SMRecovery3}}

Here, we present the results of numerical experiments used to extend the modulator recovery conditions to carriers with nonuniformly placed sample points $|c_i|=1$ in terms of the parameters $n$, $\omega$, $\varpi$, and $d$.

\subsection*{Setup}
The numerical experiments under consideration consist of the following steps.
\begin{enumerate}
\item[1.] $10^3$ pairs of $\mathbf{m}$ and $\mathbf{c}$ are generated by randomly sampling from $\Mwm$ and $\Cd$ for every feasible combination of $\omega$ and $d$ consistent with a chosen signal length $n$.
\item[2.] For every pair of $\mathbf{m}$ and $\mathbf{c}$ generated, $\mathbf{m}$ is inferred from the $\mathbf{s} = \mathbf{m} \circ \mathbf{c}$ via $\mathbf{\hat{m}}$ defined by (6). The latter is evaluated by using the AP-P algorithm, introduced in Section~III-C of the main text, with $\epsilon_{tol} = 10^{-14}$ and unlimited $N_{iter}$.
\item[3.] For every combination of the parameters $\omega$ and $d$, two estimates related to the recovery error are evaluated: 1) the average empirical error $\langle E_m \rangle$ and 2) the fraction of cases with vanishing error $P(E_m < \varepsilon)$, where $\varepsilon$ is a positive number arbitrarily close to zero. $P(E_m < \varepsilon)$ can be seen as the demodulation success rate for a given error threshold $\varepsilon$.
\end{enumerate}

A crucial aspect of the above experiments to producing informative data for our purposes is the way $\Mwm$ and $\Cd$ are sampled. For both of these sets, we exploited uniform sampling but with some additional constraints, as explained next.
\begin{itemize}
\item The cutoff frequency $\omega$, defining the modulator set, and the cutoff frequency $\varpi$, defining the estimator $\mathbf{\hat{m}}$, were fixed to be equal. This choice allowed us to considerably reduce the extent of relevant parameter combinations to be checked without loss of generality. Indeed, $\varpi \geq \omega$ is a necessary condition for a full recovery independent of $\mathbf{c} \in \Cd$ and $\Mwm \subset \Mw$ if $\varpi > \omega$. Hence, all recovery conditions applicable in the case of $\varpi = \omega$ hold for $\varpi > \omega$ as well.
\item Only the subset of pure spike-train carriers consisting of $c_i \in \{0,1\}$ sample points was considered among all possible $\mathbf{c} \in \Cd$. According to \textit{Proposition~II.2}, that is sufficient for identifying full recovery conditions without loss of generality.
\item Different elements of the pure spike-train subset of $\Cd$ may substantially differ in the number $n_s$ of sample points with $|c_i|=1$. In particular, $\lceil n/d \rceil \leq n_s \leq n-d+1$. We considered modulator reconstruction by uniformly sampling either from the sparsest ($n_s = \lceil n/d \rceil$) or the densest ($n_s = n - d +1$) subset of spike-train carriers. In view of \textit{Proposition~II.2}, pure spike-train carriers with $n_s = \lceil n/d \rceil$ have the tightest, and hence the most general, constraints for exact modulator recovery in terms of the parameters $\varpi$ and $d$.
\end{itemize}
The algorithms for sampling from the modulator and carrier sets specified above are presented in Section~\ref{sec:SMSynthetic}.

\subsection*{Results}

\begin{figure*}[ht]
\centering
\includegraphics[width=1\textwidth]{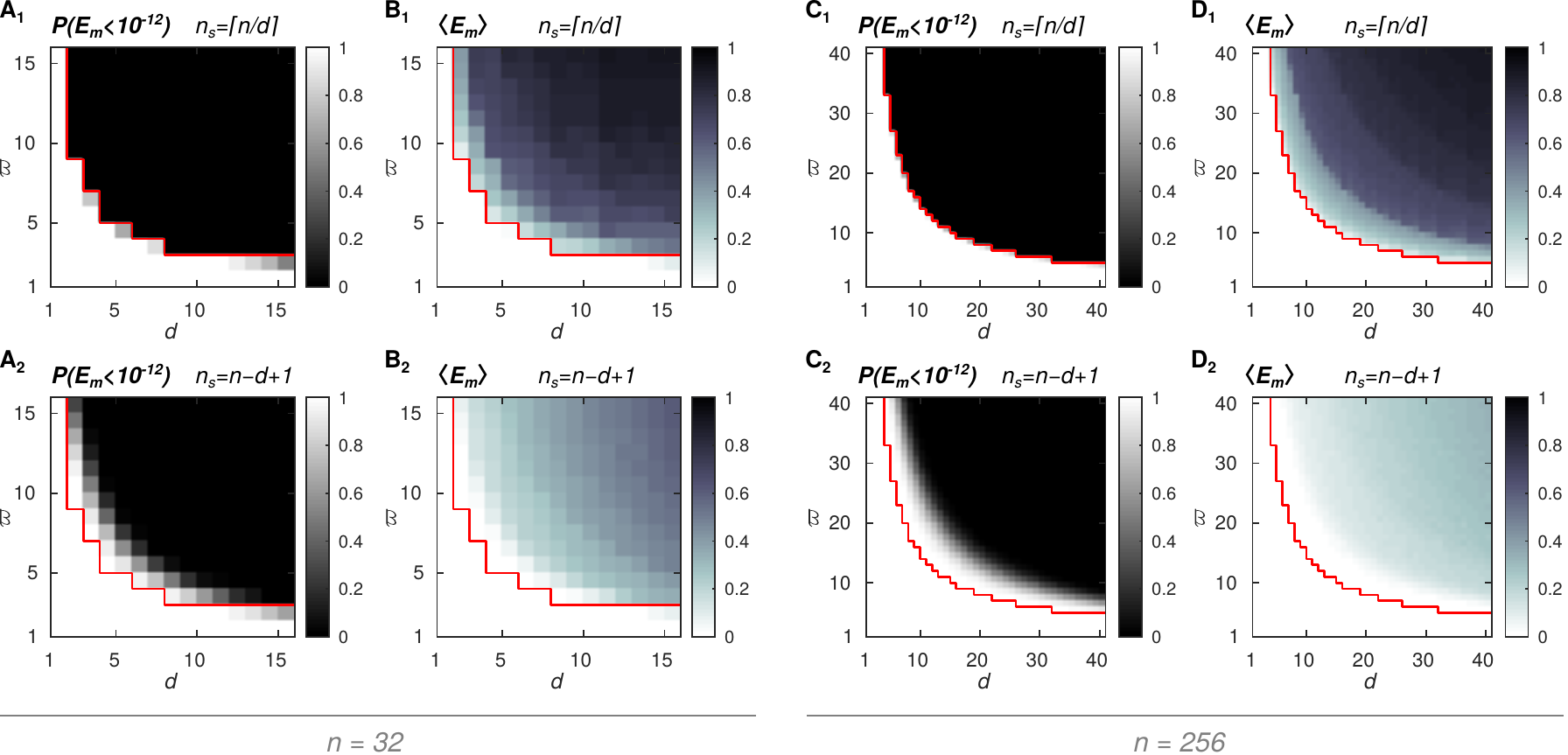}
\caption{\footnotesize Success rates and errors of modulator recovery for signals based on pure $c_i \in \{0,1\}$ spike-train carriers. $\mathbf{A_1}$--$\mathbf{A_2}$: color plots of the fraction ($P$) of modulator recovery cases with the recovery error $E_m$ lower than $10^{-12}$ for different combinations of $d$ and $\varpi$, and $n=32$; red lines plot the relation $\lceil n / d \rceil = 2\varpi-1$
. $\mathrm{A_1}$ displays the results for carriers with $n_s=\lceil n/d \rceil$ spikes, while $\mathrm{A_2}$ corresponds to carriers with $n_s=n-d+1$. $\mathbf{B_1}$--$\mathbf{B_2}$: the same as $\mathrm{A_1}$--$\mathrm{A_2}$, but with the average recovery error instead of the success rate over all modulator and carrier pairs shown for each combination of $d$ and $\varpi$. $\mathbf{C_1}$--$\mathbf{C_2}$: the same as $\mathrm{A_1}$--$\mathrm{A_2}$ except that $n=256$. $\mathbf{D_1}$--$\mathbf{D_2}$: the same as $\mathrm{B_1}$--$\mathrm{B_2}$ except that $n=256$.}
\label{fig:App1}
\end{figure*}

Fig.\,\ref{fig:App1}\,$\mathrm{A_1}$ displays color-plots of the fraction $P$ of the modulator recovery cases with $E_m<10^{-12}$ over the $(d,\varpi)$ plane for $n=32$ and $n_s = \lceil n/d \rceil$. 
In agreement with the necessary recovery condition discussed in Section~II-C, $P(E_m<10^{-12})$ is equal to 0 for all points with $\lceil n / d \rceil < 2\varpi-1$. More remarkably, for $\lceil n / d \rceil \geq 2\varpi-1$, $P(E_m<10^{-12})$ equals 1 except some boundary points, where $P$ varies between 0.65 and 1. However, even in the latter cases, the error is small, as follows from Fig.\,\ref{fig:App1}\,$\mathrm{B_1}$, which plots the average $\langle E_m \rangle$ over the $(d,\omega)$ plane. The maximum likelihood (ML) estimate of $P(E_m<10^{-12})$ for all tested modulator-carrier pairs that adhere to the necessary recovery condition $\lceil n/d \rceil \geq 2\varpi-1$ is 0.971; its $99\%$ confidence interval is $(0.969,0.972)$.

Increasing the number of carrier points with $|c_i|=1$ does not change the landscape of the recovery success rate considerably. Indeed, pushing $n_s$ to the maximum $n-d+1$ increases the $P$ values only at points in the immediate vicinity of the $\lceil n / d \rceil = 2\varpi-1$ boundary (see Fig.\,\ref{fig:App1}\,$\mathrm{A_2}$). Nevertheless, it has to be noted that the recovery errors are decreased by increasing $n_s$ on average (see Fig.\,\ref{fig:App1}\,$\mathrm{B_2}$). The ML estimate of $P(E_m<10^{-12})$ for all tested modulator-carrier pairs that adhere to the necessary recovery condition $\lceil n/d \rceil \geq 2\varpi-1$ is 0.987 in this case; its $99\%$ confidence interval is $(0.986,0.988)$.

We found an analogous picture while considering signals with different lengths $n$. One such example ($n=256$) is considered in Fig.\,\ref{fig:App1}\,$\mathrm{C_1}$\,--\,$\mathrm{D_2}$. The chances of the modulator recovery with vanishing error, i.e., $E_m < 10^{-12}$, are higher in this case. Specifically, the ML estimate of $P(E_m<10^{-12})$ is 0.9933, with the $99\%$ confidence interval $(0.9930,0.9936)$ when $n_s = \lceil n/d \rceil$. That can be explained by the smaller contribution of the boundary points of the relation $\lceil n / d \rceil = 2\varpi-1$ in the $(d,\varpi)$ plane to the total count. It is important to note that, in all cases discussed here, essentially the same results are obtained even if the error threshold $\varepsilon$ is increased to $10^{-3}$. This rejects any possibility of numerical inaccuracies affecting our conclusion.

\clearpage
\phantomsection

\markboth{AP ALGORITHMS}%
{AP ALGORITHMS}

\hspace{-11pt}\textbf{AP ALGORITHMS}
\addcontentsline{toc}{section}{AP ALGORITHMS}

\section{\textbf{Mathematical Preliminaries} \label{sec:SMTheory0}}

In this section, we introduce some basic concepts of mathematical analysis necessary for the formulation and assessment of the AP algorithms that we used to calculate modulator estimators defined by (6) and (8) in the main text.

\subsection*{\underline{Convex, interior, closed, and bounded sets}}

We start with definitions of a few basic attributes of sets in Euclidean spaces.


\begin{definition}
A set $\mathcal{S} \subseteq \mathbb{R}^n$ is said to be \textit{convex} if $\theta \cdot \mathbf{x} + (1-\theta) \cdot \mathbf{y} \in \mathcal{S}$ for all $\mathbf{x},\mathbf{y} \in \mathcal{S}$ and $\theta \in [0,1]$.
\label{def:ConvexSet}
\end{definition}


\begin{definition}
An element $\mathbf{x} \in \mathcal{S} \subseteq \mathbb{R}^n$ is said to be \textit{an interior point} of $\mathcal{S}$ if there exists an $\epsilon > 0$ such that $\{\mathbf{y} \in \mathbb{R}^n: \|\mathbf{x}-\mathbf{y} \|_2 < \epsilon \}\subset \mathcal{S}$.
\label{def:InteriorPoint}
\end{definition}

\begin{definition}
A set that consists of all interior points of $\mathcal{S}\subseteq \mathbb{R}^n$ is called \textit{the interior} of $\mathcal{S}$. We denote it by $\mathcal{S}^\circ$.
\end{definition}


\begin{definition}
An element $\mathbf{y} \in \mathbb{R}^n$ is said to be \textit{a contact point} of $\mathcal{S} \subseteq \mathbb{R}^n$ if, for any $\epsilon > 0$, there exists an $\mathbf{x} \in \mathcal{S}$ such that $\|\mathbf{x} - \mathbf{y} \|_2 < \epsilon$.
\label{def:ContactPoint}
\end{definition}

\begin{definition}
A set $\mathcal{S} \subseteq \mathbb{R}^n$ is said to be \textit{closed} if it is equal to the set of all its contact points.
\label{def:ClosedSet}
\end{definition}

Convexity and closedness of sets are preserved under intersection.

\begin{proposition}
The intersection $\mathcal{S}_1 \cap \mathcal{S}_2$ of two closed and convex sets $\mathcal{S}_1$ and $\mathcal{S}_2$ is closed and convex.
\label{prop:SetC1C2}
\end{proposition}

Another important characteristic of sets is their boundedness.


\begin{definition}
A set $\mathcal{S} \subset \mathbb{R}^n$ is said to be \textit{bounded} if there exists a $b \in \mathbb{R}$ such that $\| \mathbf{x}-\mathbf{y} \|_2 \leq b$ for all $\mathbf{x}, \mathbf{y} \in \mathcal{S}$.
\label{def:BoundedSet}
\end{definition}

\begin{definition}
A set $\mathcal{S} \subset \mathbb{R}^1=\mathbb{R}$ is said to be \textit{bounded from above} if there exists a $u \in \mathbb{R}$ such that  $\mathbf{x} \leq u$ for all $\mathbf{x}\in \mathcal{S}$. $u$ is called an \textit{upper bound} of $\mathcal{S}$.
\label{def:BoundedAboveSet}
\end{definition}

\begin{remark}
$\underline{u} \in \mathbb{R}$ is said to be the \textit{least upper bound} of $\mathcal{S} \subset \mathbb{R}$ if $\mathbf{x} \leq \underline{u}$ for all $\mathbf{x}\in \mathcal{S}$, and there exists a $\mathbf{y} \in \mathcal{S}$ for every $\epsilon > 0$ such that $\mathbf{y} > \underline{u} - \epsilon$. The least upper bound exists for any $\mathcal{S} \subset \mathbb{R}$ bounded from above due to the continuity of the real numbers.
\end{remark}

\begin{definition}
A set $\mathcal{S} \subset \mathbb{R}^1=\mathbb{R}$ is said to be \textit{bounded from below} if there exists an $l \in \mathbb{R}$ such that  $l \leq \mathbf{x}$ for all $\mathbf{x}\in \mathcal{S}$. $l$ is called a \textit{lower bound} of $\mathcal{S}$.
\label{def:BoundedBelowSet}
\end{definition}

\begin{remark}
$\bar{l} \in \mathbb{R}$ is said to be the \textit{greatest lower bound} of $\mathcal{S} \subset \mathbb{R}$ if $\bar{l} \leq \mathbf{x}$ for all $\mathbf{x}\in \mathcal{S}$, and there exists a $\mathbf{y} \in \mathcal{S}$ for every $\epsilon > 0$ such that $\mathbf{y} < \bar{l} + \epsilon$. Analogously to the least upper bound, any $\mathcal{S} \subset \mathbb{R}$ bounded from below has the greatest lower bound.
\end{remark}

\subsection*{\underline{Convergence of sequences}}


The following fundamental properties of infinite sequences of points in bounded subsets of Euclidean spaces play a critical role in the proofs of the convergence of the AP algorithms.

\begin{proposition}[Monotone Convergence Theorem]
Any monotonically decreasing sequence of real numbers
\begin{equation}
x^{(0)} \geq x^{(1)} \geq \ldots \geq x^{(i)} \geq \ldots
\end{equation}
that is bounded from below converges to its greatest lower bound $\bar{l}$, i.e., for every $\epsilon > 0$, there exists an $N(\epsilon)$ such that $|x^{(i)}-\bar{l}| < \epsilon$ whenever $i >N(\epsilon)$.
\label{prop:MonConvTh}
\end{proposition}

\begin{remark}
Analogously, any monotonically increasing sequence that is bounded from above converges to its least upper bound.
\end{remark}

\begin{definition}
Consider a sequence $\mathbf{x}^{(0)}, \mathbf{x}^{(1)}, \ldots, \mathbf{x}^{(i)}, \ldots$ in $\mathbb{R}^n$, i.e., $\mathbf{x}^{(i)} \in \mathbb{R}^n$ for every $i\geq 0$. Another sequence $\mathbf{x}^{(k_0)}, \mathbf{x}^{(k_1)}, \ldots, \mathbf{x}^{(k_i)}, \ldots$ in $\mathbb{R}^n$ generated by removing some of the elements of the original sequence is called a \textit{subsequence} of the latter. Note that $k_i > k_j$ for all $ i > j \geq 0$, and $k_i \geq i$ for every $ i \geq 0$ here.
\label{def:SubSeq}
\end{definition}

\begin{proposition}[Bolzano-Weierstrass Theorem]
Any bounded infinite sequence $\mathbf{x}^{(0)}, \mathbf{x}^{(1)}, \ldots, \mathbf{x}^{(i)}, \ldots$ in $\mathbb{R}^n$ has an infinite subsequence $\mathbf{x}^{(k_0)}, \mathbf{x}^{(k_1)}, \ldots, \mathbf{x}^{(k_i)}, \ldots$ that converges to a particular $\mathbf{x}^\dagger \in \mathbb{R}^n$, i.e., for any $\epsilon>0$, there exists an $N(\epsilon)$ such that $\|\mathbf{x}^{(k_i)}-\mathbf{x}^\dagger\|_2 < \epsilon$ whenever $ i > N(\epsilon)$.
\label{prop:BzWeTh}
\end{proposition}

\subsection*{\underline{Metric projections}}


The central operation around which AP algorithms are built is that of a metric projection.


\begin{definition}
An element $\mathbf{x_z}$ of a closed subset $\mathcal{S}$ of $\mathbb{R}^n$ is said to be \textit{a metric projection} of $\mathbf{z} \in \mathbb{R}^n$ onto $\mathcal{S}$ if $\| \mathbf{x_z} - \mathbf{z} \|_2 \leq \|\mathbf{x} - \mathbf{z} \|_2$ for all $\mathbf{x} \in \mathcal{S}$. We denote a transformation that assigns an $\mathbf{x_z} \in \mathcal{S}$ to every $\mathbf{z} \in \mathbb{R}^n$ by $\mathbf{P_\mathcal{S}}: \mathbb{R}^n \to \mathcal{S}$.
\label{def:ProjOp}
\end{definition}

\begin{remark}
$\mathbf{P_\mathcal{S}}$ generalizes the linear projection operator that assigns an element of a linear space to one of its subspaces (see, e.g., \cite[p.\,223]{Kolmogorov1975}). For the sake of brevity, we skip the qualifier ``metric'' and refer to $\mathbf{P}_\mathcal{S}$ as ``a projection'' in the sequel.
\end{remark}


\begin{proposition}[see Theoreom 5.11 in \cite{Cinlar2013}]
A projection of any element of $\mathbb{R}^n$ onto its closed convex subset $\mathcal{S}$ exists and is unique.
\label{prop:ProjOpUniq}
\end{proposition}

\subsection*{\underline{Inequalities}}

Two important inequalities that we use extensively in the convergence proofs of AP algorithms apply to Euclidean spaces.

\begin{proposition}[Triangle Inequality]
\begin{equation}
\|\mathbf{x}+\mathbf{y}\|_2 \leq \|\mathbf{x}\|_2 + \|\mathbf{y}\|_2 \qquad \forall\mathbf{x} \in \mathbb{R}^n, \forall\mathbf{y} \in \mathbb{R}^n.
\label{eq:TriIneq}
\end{equation}
%
\label{prop:TriIneq}
\end{proposition}

\begin{remark}
Another inequality relevant to us follows from \eqref{eq:TriIneq}. In particular, let us consider some $\mathbf{a},\mathbf{b},\mathbf{c} \in \mathbb{R}^{n}$. If we define $\mathbf{x}=\mathbf{a}-\mathbf{b}$ and $\mathbf{y}=\mathbf{b}-\mathbf{c}$, then \eqref{eq:TriIneq} implies $\|\mathbf{a}-\mathbf{c}\|_2 \leq \|\mathbf{a}-\mathbf{b}\|_2 + \|\mathbf{b}-\mathbf{c}\|_2$, i.e., $\|\mathbf{a}-\mathbf{b}\|_2 \geq \|\mathbf{c}-\mathbf{a}\|_2 - \|\mathbf{c}-\mathbf{b}\|_2$. On the other hand, setting $\mathbf{x}=\mathbf{c}-\mathbf{a}$ and $\mathbf{y}=\mathbf{a}-\mathbf{b}$, we obtain $\|\mathbf{c}-\mathbf{b}\|_2 \leq \|\mathbf{c}-\mathbf{a}\|_2 + \|\mathbf{a}-\mathbf{b}\|_2$, i.e., $\|\mathbf{a}-\mathbf{b}\|_2 \geq -(\|\mathbf{c}-\mathbf{a}\|_2 - \|\mathbf{c}-\mathbf{b}\|_2)$. Thus,
\begin{equation}
\|\mathbf{a}-\mathbf{b}\|_2 \geq \big| \|\mathbf{c}-\mathbf{a}\|_2 - \|\mathbf{c}-\mathbf{b}\|_2 \big| \qquad \forall\mathbf{a} \in \mathbb{R}^n, \forall\mathbf{b} \in \mathbb{R}^n, \forall\mathbf{c} \in \mathbb{R}^n.
\label{eq:TriIneq2}
\end{equation}
\end{remark}


\begin{proposition}[Containing-Half-Space Inequality, see Theoreom 5.13 in \cite{Cinlar2013}]
If $\mathcal{S} \subset \mathbb{R}^n$ is closed and convex, then
\begin{equation}
\langle \mathbf{x}-\mathbf{P}_\mathcal{S}[\mathbf{x}], \mathbf{P}_\mathcal{S}[\mathbf{x}]-\mathbf{y} \rangle \geq 0 \qquad \forall \mathbf{x} \in \mathbb{R}^n, \forall\mathbf{y} \in \mathcal{S}.
\label{eq:Prop3_3}
\end{equation}
\end{proposition}

\begin{remark}
$\mathcal{S}$ belongs to a half-space $\mathcal{H}_\mathbf{x}= \{\mathbf{z} \in \mathbb{R}^n: \langle \mathbf{x}-\mathbf{P}_\mathcal{S}[\mathbf{x}], \mathbf{P}_\mathcal{S}[\mathbf{x}]-\mathbf{z} \rangle \geq 0 \}$, which is defined for every $\mathbf{x} \in (\mathbb{R}^n \setminus \mathcal{S})$. $\mathbf{P}_\mathcal{S}[\mathbf{x}]$ is a boundary point of the $\mathcal{H}_\mathbf{x}$.
\end{remark}

\section{\textbf{Properties of the Constraint Sets and Associated Metric Projections} \label{sec:SMTheory1}}

In this section, we establish the convexity, closedness, and other relevant properties of the constraint sets $\Sgeqs$ and $\Sw$ that lay the basis for the formulation of the AP demodulation algorithms and determine their convergence properties. We then define the concrete metric projection operators of points in $\mathbb{R}^n$ onto $\Sgeqs$ and $\Sw$, which are the main building blocks of the AP demodulation algorithms defined in Section~III of the main text.

\subsection*{\underline{Properties of the constraint sets}}


The constraint sets $\Sgeqs$ and $\Sw$ that define the AP approach to demodulation introduced in Section~II of the main text have the following properties.

\begin{proposition}
The set $\Sgeqs$ is convex and closed. Its interior is $\Sgeqs^\circ = \{\mathbf{x} \in \mathbb{R}^n: x_i > |s_i|, ~ i \in \In \}$.
\label{prop:SetCl}
\end{proposition}

\begin{proof}
The range of values of each component $x_i$ of $\mathbf{x} \in \Sgeqs$ and $y_i$ of $\mathbf{y} \in \Sgeqs$ is $[|s_i|,+\infty) \subset \mathbb{R}$. Obviously,
\begin{equation*}
\theta \cdot x_i + (1-\theta) \cdot y_i \geq \min[x_i,y_i] \geq |s_i|
\end{equation*}
and
\begin{equation*}
\theta \cdot x_i + (1-\theta) \cdot y_i \leq \max[x_i,y_i] < +\infty,
\end{equation*}
where $\theta \in [0,1]$. Therefore,
\begin{equation*}
\theta \cdot x_i + (1-\theta) \cdot y_i \in [|s_i|,+\infty) \qquad \forall \theta \in [0,1], \forall i \in \In.
\end{equation*}
This, in turn, implies $\theta \cdot \mathbf{x} + (1-\theta) \cdot \mathbf{y} \in \Sgeqs$ because the range of values of each component of $\mathbf{x}$ and $\mathbf{y}$ is independent of the actual values of the remaining components. Hence, by \DefRef{def:ConvexSet}, $\Sgeqs$ is convex.

To show that $\Sgeqs$ is closed, we first note that its complement $\Sgeqsc = \mathbb{R}^n \setminus \Sgeqs$ is equal to $\{\mathbf{x} \in \mathbb{R}^n: x_i < |s_i|,~ i \in \In\}$. For any $\mathbf{x} \in \Sgeqsc$ and $\epsilon \leq \min[|\mathbf{s}|-\mathbf{x}]$,\footnote{Note that $\min[|\mathbf{s}|-\mathbf{x}]>0$.} there exists no $\mathbf{y} \in \Sgeqs$ such that $\| \mathbf{x} - \mathbf{y}\|_2 < \epsilon$. Thus, none of the elements of $\Sgeqsc$ are contact points of $\Sgeqs$. On the other hand, trivially, all points of $\Sgeqs$ are its contact points. Hence, $\Sgeqs$ coincides with the set of its contact points, i.e., it is a closed set.

To determine the interior of $\Sgeqs$, let us denote
\begin{align*}
\mathcal{A}_0 &= \{\mathbf{x} \in \mathbb{R}^n: x_i > |s_i|, ~ i \in \In \}, \\
\mathcal{A}_i &= \{\mathbf{x} \in \mathbb{R}^n: x_i=|s_i|, x_j \geq |s_j|, ~ j \in (\In \setminus \{i\}) \} \qquad \forall i \in \In.
\end{align*}
By construction, $\Sgeqs = \mathcal{A}_0 \cup (\cup_{i=1}^n \mathcal{A}_i)$. Next, we note that
\begin{equation*}
\{\mathbf{y} \in \mathbb{R}^n: \|\mathbf{x}-\mathbf{y} \|_2 < \epsilon \}\subset \Sgeqs \qquad \forall\,\epsilon \leq \min[\mathbf{x}-|\mathbf{s}|], ~\forall \mathbf{x} \in \mathcal{A}_0.
\end{equation*}
Thus, by \DefRef{def:InteriorPoint}, all points of $\mathcal{A}_0$ are interior points of $\Sgeqs$. On the other hand, for all $\mathbf{x} \in \mathcal{A}_i$ and $\epsilon > 0$, there exists a $\mathbf{y} \in \mathbb{R}^n$ such that $\|\mathbf{x}-\mathbf{y}\|_2 < \epsilon$ and $\mathbf{y} \notin \Sgeqs$. In particular, this is satisfied by $\mathbf{y}$ such that $y_i = x_i - \bar{\epsilon}$ with $0 < \bar{\epsilon} < \epsilon$ and $y_j = x_j$ for all $j \in (\In \setminus \{i\})$. Therefore, none of $\mathbf{x} \in \mathcal{A}_i$ with $i \in \In$ are interior points of $\Sgeqs$. Altogether, this allows us to conclude that the interior of $\Sgeqs$ is equal to $\mathcal{A}_0$. That $\mathcal{A}_0$ is nonempty follows straightforwardly from the fact that there exists an $x_i > s_i$ for any $s_i \in \mathbb{R}$.
\end{proof}


\begin{proposition}
The set $\Sw$ is convex, closed, and void of interior points. In particular, $\Sw$ is a linear subspace of $\mathbb{R}^n$.
\label{prop:SetCw}
\end{proposition}

\begin{proof}
It follows directly from the definition of $\mathbb{R}^n$ that
\begin{equation}
\theta \cdot \mathbf{x} + (1-\theta) \cdot \mathbf{y} \in \mathbb{R}^n \qquad \forall \mathbf{x} \in \mathbb{R}^n, \forall \mathbf{y} \in \mathbb{R}^n, \forall \theta \in [0,1],
\label{eq:Prop2_1}
\end{equation}
and thus,\footnote{Note that $\Sw \subset \mathbb{R}^n$.}
\begin{equation}
\theta \cdot \mathbf{x}_\varpi + (1-\theta) \cdot \mathbf{y}_\varpi \in \mathbb{R}^n \qquad \forall \mathbf{x}_\omega \in \Sw, \forall \mathbf{y}_\varpi \in \Sw, \forall \theta \in [0,1].
\label{eq:Prop2_2}
\end{equation}
The definition of $\Sw$ implies that $(\mathbf{F}\mathbf{x}_\varpi)_i=0$ and $(\mathbf{F}\mathbf{y}_\varpi)_i=0$ for all $\mathbf{x}_\varpi \in \Sw$, $\mathbf{y}_\varpi \in \Sw$, and $i \in (\In \setminus \Inw)$, so that
\begin{equation}
\theta \cdot (\mathbf{F}\mathbf{x}_\varpi)_i + (1-\theta) \cdot (\mathbf{F}\mathbf{y}_\varpi)_i = (\mathbf{F} (\theta \cdot \mathbf{x}_\varpi+(1-\theta) \cdot \mathbf{y}_\varpi))_i = 0 \qquad  \forall i \in (\In \setminus \Inw)
\label{eq:Prop2_3}
\end{equation}
because of the linearity of the Fourier transform. Combining \eqref{eq:Prop2_2} and \eqref{eq:Prop2_3} with the definition of $\Sw$, we conclude that $\theta \cdot \mathbf{x}_\varpi+(1-\theta) \cdot \mathbf{y}_\varpi \in \Sw$ for all $\mathbf{x}_\varpi \in \Sw$, $\mathbf{y}_\varpi \in \Sw$, i.e., $\Sw$ is convex.

To prove that $\Sw$ is closed, we show that no point of $\Swc$ is a contact point of $\Sw$. Indeed, the complement of $\Sw$ in $\mathbb{R}^n$ is given by
\begin{equation}
\textstyle \Swc = \big\{\mathbf{y} \in \mathbb{R}^n: \sum_{i \in (\In \setminus \Inw)} (\mathbf{F}\mathbf{y})_i^2 > 0 \big\}.
\label{eq:Prop2_4}
\end{equation}
Let us consider some $\mathbf{y} \in \Swc$. Taking into account the definitions of $\Sw$ and $\Swc$ and the fact that $\mathbf{F}$ is unitary, 
we have that, for any $\mathbf{x}\in\Sw$,
\begin{equation}
\textstyle
\begin{aligned}
\textstyle \|\mathbf{x} - \mathbf{y}\|_2^2 &= \|\mathbf{F}(\mathbf{x} - \mathbf{y})\|_2^2 = \|\mathbf{F}\mathbf{x} - \mathbf{F}\mathbf{y}\|_2^2 \\
&= \textstyle \sum_{i \in \Inw} ((\mathbf{F}\mathbf{x})_i-(\mathbf{F}\mathbf{y})_i)^2 + \sum_{i \in (\In \setminus \Inw)} (\mathbf{F}\mathbf{y})_i^2 \\
&\geq \textstyle \sum_{i \in (\In \setminus \Inw)} (\mathbf{F}\mathbf{y})_i^2 > 0.
\end{aligned}
\label{eq:Prop2_5}
\end{equation}
Thus, for every $\mathbf{y} \in \Swc$, there exist no $\mathbf{x} \in \Sw$ such that $\| \mathbf{x} - \mathbf{y} \|_2 < \epsilon$ with $\epsilon = \sqrt{\sum_{i \in (\In \setminus \Inw)} (\mathbf{F}\mathbf{y})_i^2}$, which means that none of the elements of $\Swc$ are contact points of $\Sw$. Therefore, $\Sw$ coincides with the set of its contact points, i.e., it is closed.

It follows from the definition of $\Sw$ that $\mathbf{y}=(\mathbf{x} + \epsilon \cdot \mathbf{e}^{(1)}) \notin \Sw$ for all $ \mathbf{x} \in \Sw$ and $\epsilon > 0$, where $\mathbf{e}^{(1)}$ is the unit vector with all but its first components equal to zero. Indeed, $(\mathbf{F}\mathbf{e}^{(1)})_i = 1/\sqrt{n} \neq 0$ for all $i \in \In$. Moreover, in that case, $\| \mathbf{x} - \mathbf{y} \|_2 = \epsilon$. Thus, for all $\mathbf{x} \in \Sw$ and $\epsilon > 0$, there exists a $\mathbf{y} \in \mathbb{R}^n$ such that $\|\mathbf{x}-\mathbf{y}\|_2 < \epsilon$ and $\mathbf{y} \notin \Sw$, which means that $\Sgeqs$ has no interior points.

A necessary and sufficient condition for a subset $\mathcal{S}$ of a linear space $\mathbb{R}^n$ to be a subspace is that $(\alpha \cdot \mathbf{x} + \beta \cdot \mathbf{y}) \in \mathcal{S}$ for all $\alpha \in \mathbb{R}$, $\beta \in \mathbb{R}$, $\mathbf{x} \in \mathcal{S}$, and $\mathbf{y} \in \mathcal{S}$ (see, e.g., \cite[p.\,121]{Kolmogorov1975}). That this applies to $\Sw$ follows from the proof of its convexity above if we replace $\theta$ and $(1-\theta)$ by, respectively, $\alpha$ and $\beta$.
\end{proof}


\begin{proposition}
The intersection of sets $\Sgeqs^\circ$ and $\Sw$ is nonempty, i.e., there exists an $\mathbf{x} \in \Sgeqs^\circ \cap \Sw$. Consequently, $\Sgeqs \cap \Sw$ is also nonempty.
\label{prop:SetClCw}
\end{proposition}

\begin{proof}
Let us consider $\mathbf{x} = \lambda \cdot \mathbf{1}$, where $\mathbf{1}$ is an element of $\mathbb{R}^n$ with all its components equal to 1, and $\lambda > \max[\mathbf{s}]$. It follows directly from the definition of $\Sgeqs$ that  $\mathbf{x} \in \Sgeqs^\circ \subset \Sgeqs$. It also follows from the definitions of $\Sw$ and unitary discrete Fourier transform that $(\mathbf{F}\mathbf{x})_i= \delta_{i,0} \cdot \sqrt{n} \cdot \lambda$, i.e., $\mathbf{x} \in \Sw$. Therefore, $\mathbf{x} \in (\Sgeqs^\circ \cap \Sw) \subset (\Sgeqs \cap \Sw)$.
\end{proof}

\begin{remark}
It is a direct consequence of \textit{Propositions~\ref*{prop:SetC1C2}}, \textit{\ref*{prop:SetCl}}, and \textit{\ref*{prop:SetCw}} that $\Sgeqs \cap \Sw$ is also closed and convex.
\end{remark}


\subsection*{\underline{Metric projections onto $\Sgeqs$ and $\Sw$}}


The metric projections of any point in $\mathbb{R}^n$ onto its convex subsets relevant to us, i.e., $\Sgeqs$ and $\Sw$, are achieved by the following operators.

\begin{proposition}
The metric projection operator $\mathbf{P}_{\Sgeqss}$ is defined by the elementwise maximum of the target signal $\mathbf{s}$ and the input argument $\mathbf{z}$\rm{:}
\begin{equation}
\mathbf{P}_{\Sgeqss}[\mathbf{z}] = |\mathbf{s}| + (\mathbf{z}-|\mathbf{s}|) \circ \theta (\mathbf{z}-|\mathbf{s}|).
\label{eq:MathPrel0}
\end{equation}
\end{proposition}

\begin{proof}
Note that
\begin{align}
\big(\mathbf{P}_{\Sgeqss}[\mathbf{z}] \big)_i =\begin{cases} z_i, & \mbox{if $z_i \geq |s_i|$}\\ |s_i|, & \mbox{if $z_i < |s_i|$} \end{cases} \qquad \forall i\in \In.
\label{eq:MathPrel1}
\end{align}
Hence, a necessary and sufficient condition for transforming any $\mathbf{z} \in \mathbb{R}^n$ to $\mathbf{x} \in \Sgeqs$ is to increase every component $z_i$ of $\mathbf{z}$ that does not satisfy $z_i \geq |s_i|$ by at least $|s_i|-z_i$, independent of values of the remaining components. Now, we have from the definition of the Euclidean norm that
\begin{equation}
\textstyle \|\mathbf{z} - \mathbf{x} \|_2 = \sqrt{\sum_{i=1}^n(z_i-x_i)^2} \qquad \forall \mathbf{z} \in \mathbb{R}^n, \forall\mathbf{x} \in \Sgeqs.
\end{equation}
Thus, for every $\mathbf{z} \in \mathbb{R}^n$, $\|\mathbf{z} - \mathbf{x} \|_2$ is minimized by an $\mathbf{x}\in\Sgeqs$ that is obtained by incrementing all components of $\mathbf{z}$ that satisfy $z_i < |s_i|$ by no more than necessary, i.e., by $|s_i|-z_i$, and leaving the remaining components intact. However, this is precisely how the operator $\mathbf{P}_{\Sgeqss}$ is defined via \eqref{eq:MathPrel1}. Therefore, by using the \DefRef{def:ProjOp} of the projection, we conclude that $\mathbf{P}_{\Sgeqss}$ projects $\mathbf{z} \in \mathbb{R}^n$ onto $\Sgeqs$.
\end{proof}

\begin{proposition}
The metric projection operator $\mathbf{P}_{\Sw}$ is defined by a rectangular low-pass-filter transformation
\begin{equation}
\mathbf{P}_{\Sw}[\mathbf{z}] = (\mathbf{F^{-1}} \, \mathbf{W}_\varpi \, \mathbf{F}) \, \mathbf{z} = \sum_{i\in\Inw} \langle \mathbf{f}^{(i)}, \mathbf{z}\rangle \cdot \mathbf{f}^{(i)}. \label{eq:MathPrel3}
\end{equation}
Here, $\mathbf{f}^{(i)}$ is the i-th column of the \rm{DFT matrix} $\mathbf{F}$. $\mathbf{W}_\varpi$ is a diagonal matrix such that
\begin{align}
(W_{\varpi})_{ii} =\begin{cases} 1, & \mbox{if $i \in \Inw$}\\ 0, & \mbox{otherwise} \end{cases}.
\label{eq:MathPrel4}
\end{align}
\end{proposition}

\begin{proof}
Let us consider an element of the set $\Sw$ expressed by $\mathbf{x}=\sum_{i\in\Inw}a_i \cdot \mathbf{f}^{(i)}$. It follows from the definition of the Euclidean norm (see Section~II in the main text) that
\begin{equation}
\begin{aligned}
\| \mathbf{x} - \mathbf{z}\|_2^2 &= \bigg\langle \bigg(\mathbf{z} - \sum_{i\in\Inw}a_i \cdot \mathbf{f}^{(i)}\bigg), \bigg(\mathbf{z} - \sum_{i\in\Inw}a_i \cdot \mathbf{f}^{(i)}\bigg) \bigg\rangle \\
&= \langle\mathbf{z}, \mathbf{z} \rangle -2 \cdot \sum_{i\in\Inw}a_i \cdot \langle \mathbf{f}^{(i)}, \mathbf{z}\rangle + \sum_{i\in\Inw}a_i^2 \\
&= \langle\mathbf{z}, \mathbf{z} \rangle - \sum_{i\in\Inw}\langle \mathbf{f}^{(i)},\mathbf{z}\rangle^2 + 	\sum_{i\in\Inw}(a_i-\langle \mathbf{f}^{(i)}, \mathbf{z}\rangle)^2\,.
\end{aligned}\label{eq:MathPrel5}
\end{equation}
When writing the second equality above, we used the fact that $\{ \mathbf{f}^{(1)}, \mathbf{f}^{(2)}, \ldots, \mathbf{f}^{(n)} \}$ are orthonormal. It follows from the last equality of \eqref{eq:MathPrel5} that $\| \mathbf{x} - \mathbf{z}\|_2$, as a function of $a_i$, is minimized by $a_i = \langle \mathbf{f}^{(i)}, \mathbf{z}\rangle$ for every $i \in \Inw$. Thus, by \DefRef{def:ProjOp}, $\sum_{i\in\Inw} \langle \mathbf{f}^{(i)}, \mathbf{z}\rangle \cdot \mathbf{f}^{(i)} = \mathbf{P}_{\Sw}[\mathbf{z}]$ is the projection of $\mathbf{z}$ onto $\Sw$.
\end{proof}

\begin{remark}
$\mathbf{P}_{\Sw}$ is a linear operator, whereas $\mathbf{P}_{\Sgeqss}$ 
is not.
\end{remark}

\section{\textbf{Convergence Proofs} \label{sec:SMTheory2}}

Here, we provide proofs of the propositions concerning the convergence of the AP algorithms that are formulated in Section~III of the main text. The proofs are adapted for finite-dimensional Euclidean spaces and exploit the particular structure of the modulator constraint sets.\footnote{For the foundations of AP algorithms in a more general context of arbitrary closed convex subsets of Hilbert spaces, we refer an interested reader to the seminal works by Bregman \cite{Bregman1967}, Gurin et al. \cite{Gubin1967}, and Boyle \& Dykstra \cite{Boyle1986}.} For the sake of convenience, we repeat the original assertions as well.

\subsection*{\textbf{AP-B algorithm}}


\begin{proposition_III_1}
A sequence $\mathbf{m}^{(0)},\mathbf{m}^{(1)}, \ldots, \mathbf{m}^{(i)},\ldots$ formed by the AP-B algorithm for $\epsilon_{tol} = 0$ and $N_{iter} \to +\infty$ converges to some $\mathbf{m}^\dagger \in \Sgeqs \cap \Sw$. The convergence is geometric and monotonic, i.e., there exist $\gamma>0$ and $0<r<1$ such that $\| \mathbf{m}^{(i)} - \mathbf{m}^\dagger \|_2 \leq \gamma \cdot r^i$ and $\| \mathbf{m}^{(i+1)} - \mathbf{m}^\dagger \|_2 \leq \| \mathbf{m}^{(i)} - \mathbf{m}^\dagger \|_2$ for $i \geq 0$.
\label{prop:APBSolConv_}
\end{proposition_III_1}

\begin{proof}
If the sequence $\mathbf{m}^{(0)},\mathbf{m}^{(1)}, \ldots, \mathbf{m}^{(i)}, \ldots $ terminates with some $\mathbf{m}^{(N)}$, i.e., $\mathbf{m}^{(N+j)}=\mathbf{m}^{(N)}$ for every $j>0$, it follows from the formulation of the AP-B algorithm that $\mathbf{m}^{(N)}=\mathbf{a}^{(N)}$, i.e., $\mathbf{m}^{(N)} \in \Sgeqs \cap \Sw$.\footnote{We remind the reader that $\mathbf{a}^{(i)} = \mathbf{P}_{\Sw}[\mathbf{m}^{(i-1)}]$ for any $i>0$ in the case of the AP-B algorithm.} Thus, $\mathbf{m}^{(N)} = \mathbf{m}^\dagger$, which means that the solution is achieved in a finite number of iterations. We next consider the case when the sequence $\mathbf{m}^{(0)},\mathbf{m}^{(1)}, \ldots, \mathbf{m}^{(i)}, \ldots$ is infinite. The rest of the proof is divided into three parts for clarity.


\underline{Convergence.} The outline of the convergence proof is as follows. We first demonstrate that the distance between any $\mathbf{x} \in \Sgeqs \cap \Sw$ and $\mathbf{m}^{(i)}$ or $\mathbf{a}^{(i)}$ decreases with every iteration. Using this, we then show that the sequence $\mathbf{m}^{(0)},\mathbf{m}^{(1)}, \ldots, \mathbf{m}^{(i)}, \ldots$ is bounded, and thus, by the Bolzano-Weierstrass theorem and closedness of $\Sgeqs$ and $\Sw$, has a subsequence that converges to some $\mathbf{m}^\dagger \in \Sgeqs \cap \Sw$. Referring to the first result again (that the distance between any $\mathbf{x} \in \Sgeqs \cap \Sw$ and $\mathbf{m}^{(i)}$ decreases with every iteration), we finally deduce that the original sequence $\mathbf{m}^{(0)},\mathbf{m}^{(1)}, \ldots, \mathbf{m}^{(i)}, \ldots$ converges to the same $\mathbf{m}^\dagger \in \Sgeqs \cap \Sw$ as any of its infinite subsequences.

$\Sgeqs \cap \Sw$ is nonempty by \PropRef{prop:SetClCw}. Let us consider some $\mathbf{x} \in \Sgeqs \cap \Sw$ together with $\mathbf{m}^{(i)}$ and $\mathbf{a}^{(i)}$ taken from the sequences $\mathbf{m}^{(0)},\mathbf{m}^{(1)}, \ldots, \mathbf{m}^{(i)}, \ldots$ and $\mathbf{a}^{(0)},\mathbf{a}^{(1)}, \ldots, \mathbf{a}^{(i)}, \ldots$ for some $i \geq 0$. Then, we have
\begin{equation}
\begin{aligned}
\| \mathbf{x} - \mathbf{m}^{(i)} \|_2^2 &= \| \mathbf{x} - \mathbf{a}^{(i+1)} + \mathbf{a}^{(i+1)} - \mathbf{m}^{(i)} \|_2^2 \\
&= \underbrace{\| \mathbf{x} - \mathbf{a}^{(i+1)} \|_2^2}_{\geq 0} + \underbrace{\| \mathbf{a}^{(i+1)} - \mathbf{m}^{(i)} \|_2^2}_{\geq 0} + 2 \cdot \underbrace{\langle \mathbf{m}^{(i)} - \mathbf{a}^{(i+1)} , \mathbf{a}^{(i+1)} - \mathbf{x} \rangle}_{\geq 0}.
\label{eq:PropIII.1_1}
\end{aligned}
\end{equation}
The nonnegativity of the last term in the second line of \eqref{eq:PropIII.1_1} follows from the containing-half-space inequality \eqref{eq:Prop3_3} and $\mathbf{a}^{(i+1)} = \mathbf{P}_{\Sw}[\mathbf{m}^{(i)}]$. Therefore, \eqref{eq:PropIII.1_1} implies that
\begin{equation}
\| \mathbf{x} - \mathbf{m}^{(i)} \|_2^2 \geq \| \mathbf{x} - \mathbf{a}^{(i+1)} \|_2^2 + \| \mathbf{a}^{(i+1)} - \mathbf{m}^{(i)} \|_2^2 \qquad \forall i \geq 0
\label{eq:PropIII.1_2}
\end{equation}
and
\begin{equation}
\| \mathbf{x} - \mathbf{m}^{(i)} \|_2 \geq \| \mathbf{x} - \mathbf{a}^{(i+1)} \|_2 \qquad \forall i \geq 0.
\label{eq:PropIII.1_3}
\end{equation}
Replacing $\mathbf{m}^{(i)}$ by $\mathbf{a}^{(i+1)}$ and $\mathbf{a}^{(i+1)}$ by $\mathbf{m}^{(i+1)}$ in \eqref{eq:PropIII.1_1}, and using the same argumentation as above, including $\mathbf{m}^{(i+1)} = \mathbf{P}_{\Sgeqss}[\mathbf{a}^{(i+1)}]$, we deduce that
\begin{equation}
\| \mathbf{x} - \mathbf{a}^{(i+1)} \|_2^2 \geq \| \mathbf{x} - \mathbf{m}^{(i+1)} \|_2^2 + \| \mathbf{m}^{(i+1)} - \mathbf{a}^{(i+1)} \|_2^2 \qquad \forall i \geq -1
\label{eq:PropIII.1_4}
\end{equation}
and
\begin{equation}
\| \mathbf{x} - \mathbf{a}^{(i+1)} \|_2 \geq \| \mathbf{x} - \mathbf{m}^{(i+1)} \|_2 \qquad \forall i \geq -1.
\label{eq:PropIII.1_5}
\end{equation}
The validity of \eqref{eq:PropIII.1_4} and \eqref{eq:PropIII.1_5} for not only $i \geq 0$ but also $i = -1$ follows from the particular initial conditions of the AP-B algorithm. Combining \eqref{eq:PropIII.1_3} and \eqref{eq:PropIII.1_5} yields

\begin{equation}
\| \mathbf{x} - \mathbf{m}^{(0)} \|_2 \geq \| \mathbf{x} - \mathbf{m}^{(1)} \|_2 \geq \ldots \geq \| \mathbf{x} - \mathbf{m}^{(i)} \|_2 \geq \ldots
\label{eq:PropIII.1_6}
\end{equation}
and, equivalently,
\begin{equation}
\| \mathbf{x} - \mathbf{m}^{(0)} \|_2^2 \geq \| \mathbf{x} - \mathbf{m}^{(1)} \|_2^2 \geq \ldots \geq \| \mathbf{x} - \mathbf{m}^{(i)} \|_2^2 \geq \ldots.
\label{eq:PropIII.1_7}
\end{equation}

\eqref{eq:PropIII.1_7} states that the sequence $\| \mathbf{x} - \mathbf{m}^{(0)} \|_2^2, \| \mathbf{x} - \mathbf{m}^{(1)} \|_2^2, \ldots \| \mathbf{x} - \mathbf{m}^{(i)} \|_2^2, \ldots$ is monotonically decreasing. This sequence is bounded from below by 0 because of the nonnegativity of the norm, and therefore, it converges to its greatest lower bound $L \geq 0$ by the monotone convergence theorem (see \PropRef{prop:MonConvTh}). Thus, for every $\epsilon>0$, there exists an $N(\epsilon)$ such that $0 \leq \| \mathbf{x} - \mathbf{m}^{(i)} \|_2^2 - L \leq \epsilon$ whenever $i > N(\epsilon)$. It follows then from \eqref{eq:PropIII.1_3} and \eqref{eq:PropIII.1_5} that $L \leq \|\mathbf{x} -\mathbf{a}^{(i+1)} \|_2^2 \leq \|\mathbf{x} -\mathbf{m}^{(i)} \|_2^2 \leq L+\epsilon$, so that $0 \leq \|\mathbf{x} - \mathbf{m}^{(i)} \|_2^2 - \|\mathbf{x} - \mathbf{a}^{(i+1)} \|_2^2 < \epsilon$, and because of \eqref{eq:PropIII.1_2}, also $0 \leq \|\mathbf{a}^{(i+1)} - \mathbf{m}^{(i)} \|_2 < \sqrt{\epsilon} $ whenever $i > N(\epsilon)$, i.e., the sequence $\|\mathbf{a}^{(1)} - \mathbf{m}^{(0)} \|_2, \|\mathbf{a}^{(2)} - \mathbf{m}^{(1)} \|_2, \ldots, \|\mathbf{a}^{(i)} - \mathbf{m}^{(i-1)} \|_2, \ldots$ converges to 0. If the sequence converges, then any of its infinite subsequences (see \DefRef{def:SubSeq}) converges as well, because the removal of elements from the sequence does not change the validity of the convergence condition:
\begin{equation}
\forall \epsilon >0 ~~\exists N'(\epsilon): ~ i > N'(\epsilon) \implies \|\mathbf{a}^{(k_i+1)} - \mathbf{m}^{(k_i)} \|_2 < \epsilon,
\label{eq:PropIII.1_8}
\end{equation}
where $k_i > k_j$ for $i > j \geq 0$ and $k_i \geq i$ for $i \geq 0$.

Next, we have from the triangle inequality \eqref{eq:TriIneq} that
\begin{equation}
\| \mathbf{m}^{(i)} - \mathbf{m}^{(j)} \|_2 \leq \| \mathbf{x} - \mathbf{m}^{(i)} \|_2 + \| \mathbf{x} - \mathbf{m}^{(j)} \|_2 \qquad \forall i, j \geq 0.
\label{eq:PropIII.1_9}
\end{equation}
Moreover, according to \eqref{eq:PropIII.1_6}, $\| \mathbf{x} - \mathbf{m}^{(i)} \|_2 \leq \| \mathbf{x} - \mathbf{m}^{(0)} \|_2$ for $ i \geq 0$. Thus, for all $i,j \geq 0$,
\begin{equation}
\| \mathbf{m}^{(i)} - \mathbf{m}^{(j)} \|_2 \leq 2 \cdot \| \mathbf{x} - \mathbf{m}^{(0)} \|_2,
\label{eq:PropIII.1_10}
\end{equation}
i.e., the sequence $\mathbf{m}^{(0)}, \mathbf{m}^{(1)}, \ldots, \mathbf{m}^{(i)}, \ldots$ is bounded (see \DefRef{def:BoundedSet}). Consequently, according to the Bolzano-Weierstrass theorem (see \PropRef{prop:BzWeTh}), this sequence has a subsequence $\mathbf{m}^{(k_0)}, \mathbf{m}^{(k_1)}, \ldots, \mathbf{m}^{(k_i)}, \ldots$ that converges to some $\mathbf{m}^\dagger \in \mathbb{R}^n$:
\begin{equation}
\forall \epsilon >0 ~~\exists N''(\epsilon): ~ i > N''(\epsilon) \implies \|\mathbf{m}^\dagger - \mathbf{m}^{(k_i)} \|_2 < \epsilon.
\label{eq:PropIII.1_11}
\end{equation}

We show now that $\mathbf{m}^\dagger \in \Sgeqs \cap \Sw$. By construction, $\mathbf{m}^{(i)} \in \Sgeqs $ for every $i \geq 0$. According to \eqref{eq:PropIII.1_11}, there exists an $\mathbf{m}^{(i)}$ for any $\epsilon > 0$ such that $\|\mathbf{m}^{(i)} - \mathbf{m}^\dagger \|_2 < \epsilon$. Hence, $\mathbf{m}^\dagger$ is a contact point of $\Sgeqs$ (see \DefRef{def:ContactPoint}). Moreover, because the latter set is closed (see \DefRef{def:ClosedSet} and \PropRef{prop:SetCl}), $\mathbf{m}^\dagger \in \Sgeqs$. Next, by exploiting the triangle inequality \eqref{eq:TriIneq}, we can write
\begin{equation}
\| \mathbf{m}^\dagger - \mathbf{a}^{(k_i+1)} \|_2 \leq \| \mathbf{a}^{(k_i+1)} - \mathbf{m}^{(k_i)} \|_2 + \| \mathbf{m}^\dagger - \mathbf{m}^{(k_i)} \|_2 \qquad \forall i \geq 0.
\label{eq:PropIII.1_12}
\end{equation}
Combining \eqref{eq:PropIII.1_12} with \eqref{eq:PropIII.1_8} and \eqref{eq:PropIII.1_11} and introducing $N'''(\epsilon) = \max[N'(\epsilon/2), N''(\epsilon/2)]$, we get
\begin{equation}
\forall \epsilon >0 ~~\exists N'''(\epsilon): ~ i > N'''(\epsilon) \implies \|\mathbf{m}^\dagger - \mathbf{a}^{(k_i+1)} \|_2 < \epsilon,
\label{eq:PropIII.1_13}
\end{equation}
i.e., the subsequence $\mathbf{a}^{(k_0+1)}, \mathbf{a}^{(k_1+1)}, \ldots, \mathbf{a}^{(k_i+1)}, \ldots$ converges to $\mathbf{m}^\dagger$. The set $\Sw$ is closed (see \PropRef{prop:SetCw}), and $\mathbf{a}^{(i)} \in \Sw$ for every $i \geq 0$ by construction. Therefore, using the same argumentation as for the subsequence $\mathbf{m}^{(k_0)},$ $\mathbf{m}^{(k_1)}, \ldots, \mathbf{m}^{(k_i)}, \ldots$, we conclude that $\mathbf{m}^\dagger \in \Sw$.

Finally, because $\mathbf{m}^\dagger \in \Sgeqs \cap \Sw$, \eqref{eq:PropIII.1_6} gives
\begin{equation}
\| \mathbf{m}^\dagger - \mathbf{m}^{(0)} \|_2 \geq \| \mathbf{m}^\dagger - \mathbf{m}^{(1)} \|_2 \geq \ldots \geq \| \mathbf{m}^\dagger - \mathbf{m}^{(i)} \|_2 \geq \ldots
\label{eq:PropIII.1_14}.
\end{equation}
In the light of \eqref{eq:PropIII.1_14}, the statement of \eqref{eq:PropIII.1_11} generalizes to
\begin{equation}
\forall \epsilon >0 ~~\exists N''''(\epsilon): ~ i > N''''(\epsilon) \implies \|\mathbf{m}^\dagger - \mathbf{m}^{(i)} \|_2 < \epsilon,
\label{eq:PropIII.1_15}
\end{equation}
where $N''''(\epsilon)=k_{N''(\epsilon)+1}$. Thus, the sequence $\mathbf{m}^{(0)},$ $\mathbf{m}^{(1)}, \ldots, \mathbf{m}^{(i)}, \ldots$ converges to $\mathbf{m}^\dagger \in \Sgeqs \cap \Sw$.


\underline{Monotonicity.} The monotonicity of the convergence of the sequence $\mathbf{m}^{(0)},$ $\mathbf{m}^{(1)}, \ldots, \mathbf{m}^{(i)}, \ldots$ to $\mathbf{m}^\dagger$ is declared by \eqref{eq:PropIII.1_14}.


\underline{Rate.}
The key point in establishing the geometric convergence of the AP-B algorithm is the fact that the intersection of $\Sw$ and the interior of $\Sgeqs$ is nonempty. Using this fact, we first show that the distances between $\mathbf{m}^{(i)}$ and $\Sgeqs^\circ \cap \Sw$ or $\mathbf{a}^{(i+1)}$ and $\Sgeqs^\circ \cap \Sw$ can be bounded by, respectively, $\|\mathbf{m}^{(i)} - \mathbf{a}^{(i+1)}\|_2$ or $\|\mathbf{m}^{(i+1)} - \mathbf{a}^{(i+1)}\|_2$ multiplied by a universal factor that is greater than one and independent of the iteration number $i$. In the next step, we exploit properties of metric projections to obtain two additional inequalities, which, in combination with the first result, allow deriving a decreasing geometric sequence that bounds $\|\mathbf{m}^{(0)}-\mathbf{m}^\dagger\|_2, \|\mathbf{m}^{(1)}-\mathbf{m}^\dagger\|_2, \ldots, \|\mathbf{m}^{(i)}-\mathbf{m}^\dagger\|_2, \ldots$ from above.

To start, let us consider some $\mathbf{x} \in \Sgeqs^\circ \cap \Sw$. Such an element exists according to \PropRef{prop:SetClCw}. Also, there exists an $\epsilon > 0$ such that $\mathbf{y} \in \Sgeqs$ if $\|\mathbf{x}-\mathbf{y}\|_2<\epsilon$ because $\mathbf{x}$ belongs to the interior of $\Sgeqs$ (see \DefRef{def:InteriorPoint}). If so, then it is also possible to choose a positive $\beta < \epsilon$ such that $\mathbf{y} \in \Sgeqs$ if $\|\mathbf{x}-\mathbf{y}\|_2\leq\beta$. We now introduce
\begin{equation}
\mathbf{z}^{(i)} = \frac{\alpha_i}{\alpha_i+\beta} \cdot \mathbf{x} + \frac{\beta}{\alpha_i+\beta} \cdot \mathbf{a}^{(i)}, \quad i \geq 0,
\label{eq:PropIII.1_16}
\end{equation}
where $\alpha_i = \|\mathbf{a}^{(i)}-\mathbf{P}_{\Sgeqss}[\mathbf{a}^{(i)}]\|_2 = \|\mathbf{a}^{(i)}-\mathbf{m}^{(i)}\|_2$. Note that $\alpha_i/(\alpha_i+\beta) \in (0,1)$, and $\beta/(\alpha_i+\beta) = 1 - \alpha_i/(\alpha_i+\beta)$. Moreover, $\mathbf{x} \in \Sw$, and $\mathbf{a}^{(i)} \in \Sw$ by construction. Therefore, by the \DefRef{def:ConvexSet} of a convex set and the fact that $\Sw$ is convex (see \PropRef{prop:SetCw}), we have $\mathbf{z}^{(i)} \in \Sw$. On the other hand, \eqref{eq:PropIII.1_16} can be rewritten as
\begin{equation}
\mathbf{z}^{(i)} = \frac{\alpha_i}{\alpha_i+\beta} \cdot \underbrace{\Big(\mathbf{x}+\frac{\beta}{\alpha_i} \cdot (\mathbf{a}^{(i)}-\mathbf{m}^{(i)})\Big)}_{\mathbf{y}'} + \frac{\beta}{\alpha_i+\beta} \cdot \mathbf{m}^{(i)}.
\label{eq:PropIII.1_17}
\end{equation}
In the above expression, $\|\mathbf{x}-\mathbf{y}'\|_2=\beta$. Therefore, $\mathbf{y}' \in \Sgeqs$ by the definition of $\beta$. Moreover, $\mathbf{m}^{(i)} \in \Sgeqs$ by construction, which implies $\mathbf{z}^{(i)} \in \Sgeqs$ because $\Sgeqs$ is convex (see \PropRef{prop:SetCl}). Hence, altogether, we conclude that $\mathbf{z}^{(i)} \in \Sgeqs \cap \Sw$ for $i \geq 0$.

Based on the above consideration, we show now that
\begin{equation}
\|\mathbf{m}^{(i)} - \mathbf{P}_{\Sgeqss\cap\Sw}[\mathbf{m}^{(i)}]\|_2 \leq \|\mathbf{m}^{(i)} - \mathbf{a}^{(i+1)}\|_2 \cdot (1+\|\mathbf{x}-\mathbf{m}^{(0)}\|_2 / \beta)
\label{eq:PropIII.1_18}
\end{equation}
and
\begin{equation}
\|\mathbf{a}^{(i+1)} - \mathbf{P}_{\Sgeqss\cap\Sw}[\mathbf{a}^{(i+1)}]\|_2 \leq \|\mathbf{m}^{(i+1)} - \mathbf{a}^{(i+1)}\|_2 \cdot (1+\|\mathbf{x}-\mathbf{m}^{(0)}\|_2 / \beta)
\label{eq:PropIII.1_19}
\end{equation}
for $i \geq 0$. To demonstrate \eqref{eq:PropIII.1_18}, note that
\begin{equation}
\begin{aligned}
\|\mathbf{m}^{(i)} - \mathbf{P}_{\Sgeqss\cap\Sw}[\mathbf{m}^{(i)}]\|_2 &\leq \|\mathbf{m}^{(i)} - \mathbf{z}^{(i)}\|_2 \\
&\leq \|\mathbf{m}^{(i)} - \mathbf{a}^{(i+1)}\|_2 + \|\mathbf{a}^{(i+1)} - \mathbf{z}^{(i+1)}\|_2 \\
&= \|\mathbf{m}^{(i)} - \mathbf{a}^{(i+1)}\|_2 + \frac{\alpha_{i+1}}{\alpha_{i+1}+\beta} \cdot \|\mathbf{x} - \mathbf{a}^{(i+1)}\|_2 \\
&\leq \|\mathbf{m}^{(i)} - \mathbf{a}^{(i+1)}\|_2 + \frac{\alpha_{i+1}}{\beta} \cdot \|\mathbf{x} - \mathbf{m}^{(0)}\|_2\\
&= \|\mathbf{m}^{(i)} - \mathbf{a}^{(i+1)}\|_2 + \frac{\|\mathbf{m}^{(i+1)} - \mathbf{a}^{(i+1)}\|_2}{\beta} \cdot \|\mathbf{x} - \mathbf{m}^{(0)}\|_2\\
&\leq \|\mathbf{m}^{(i)} - \mathbf{a}^{(i+1)}\|_2 + \frac{\|\mathbf{m}^{(i)} - \mathbf{a}^{(i+1)}\|_2}{\beta} \cdot \|\mathbf{x} - \mathbf{m}^{(0)}\|_2.
\end{aligned}
\label{eq:PropIII.1_20}
\end{equation}
In \eqref{eq:PropIII.1_20}, we used the fact that $\mathbf{z} \in \Sgeqs \cap \Sw$ and the \DefRef{def:ProjOp} of the projection operator (the first line), the triangle inequality \eqref{eq:TriIneq} (the second line), \eqref{eq:PropIII.1_16} (the third line), combined inequalities \eqref{eq:PropIII.1_3} and \eqref{eq:PropIII.1_5} (the fourth line), and the \DefRef{def:ProjOp} of the projection operator again (the last line). Similarly to \eqref{eq:PropIII.1_20}, we can write
\begin{equation}
\begin{aligned}
\|\mathbf{a}^{(i+1)} - \mathbf{P}_{\Sgeqss\cap\Sw}[\mathbf{a}^{(i+1)}]\|_2 &\leq \|\mathbf{a}^{(i+1)} - \mathbf{z}^{(i+1)}\|_2 \\
&= \frac{\alpha_{i+1}}{\alpha_{i+1}+\beta} \cdot \|\mathbf{x}-\mathbf{a}^{(i+1)}\|_2 \\
&\leq \frac{\alpha_{i+1}}{\beta} \cdot \|\mathbf{x} - \mathbf{m}^{(0)}\|_2 \\
&\leq \alpha_{i+1} + \frac{\alpha_{i+1}}{\beta} \cdot \|\mathbf{x} - \mathbf{m}^{(0)}\|_2\\
&=\|\mathbf{m}^{(i+1)} - \mathbf{a}^{(i+1)}\|_2 + \frac{\|\mathbf{m}^{(i+1)} - \mathbf{a}^{(i+1)}\|_2}{\beta} \cdot \|\mathbf{x} - \mathbf{m}^{(0)}\|_2,
\end{aligned}
\label{eq:PropIII.1_21}
\end{equation}
which proves \eqref{eq:PropIII.1_19}.

Next, we derive two additional inequalities. In particular, we have
\begin{equation}
\begin{aligned}
\|\mathbf{m}^{(i)} - \mathbf{P}_{\Sgeqss\cap\Sw}[&\mathbf{m}^{(i)}]\|_2^2 - \|\mathbf{a}^{(i+1)} - \mathbf{P}_{\Sgeqss\cap\Sw}[\mathbf{a}^{(i+1)}]\|_2^2 \\[6pt]
&\geq \|\mathbf{m}^{(i)} - \mathbf{P}_{\Sgeqss\cap\Sw}[\mathbf{m}^{(i)}]\|_2^2 - \|\mathbf{a}^{(i+1)} - \mathbf{P}_{\Sgeqss\cap\Sw}[\mathbf{m}^{(i)}]\|_2^2 \\
&= \|\mathbf{m}^{(i)} - \mathbf{P}_{\Sgeqss\cap\Sw}[\mathbf{m}^{(i)}]\|_2^2 - \|(\mathbf{a}^{(i+1)} - \mathbf{m}^{(i)}) + (\mathbf{m}^{(i)} - \mathbf{P}_{\Sgeqss\cap\Sw}[\mathbf{m}^{(i)}])\|_2^2 \\
&= -\|\mathbf{a}^{(i+1)}-\mathbf{m}^{(i)}\|_2^2 + 2 \cdot \langle \mathbf{a}^{(i+1)}-\mathbf{m}^{(i)}, \mathbf{P}_{\Sgeqss\cap\Sw}[\mathbf{m}^{(i)}] -\mathbf{m}^{(i)}\rangle\\
&= -\|\mathbf{a}^{(i+1)}-\mathbf{m}^{(i)}\|_2^2 + 2 \cdot \langle \mathbf{a}^{(i+1)}-\mathbf{m}^{(i)}, (\mathbf{P}_{\Sgeqss\cap\Sw}[\mathbf{m}^{(i)}] -\mathbf{a}^{(i+1)})+(\mathbf{a}^{(i+1)}-\mathbf{m}^{(i)}) \rangle\\
&= \|\mathbf{a}^{(i+1)}-\mathbf{m}^{(i)}\|_2^2 + 2 \cdot \underbrace{\langle \mathbf{m}^{(i)}-\mathbf{a}^{(i+1)}, \mathbf{a}^{(i+1)} - \mathbf{P}_{\Sgeqss\cap\Sw}[\mathbf{m}^{(i)}] \rangle}_{\geq 0}\\
&\geq \|\mathbf{a}^{(i+1)}-\mathbf{m}^{(i)}\|_2^2
\end{aligned}
\label{eq:PropIII.1_22}
\end{equation}
for $i \geq 0$. Thus,
\begin{equation}
\begin{aligned}
\|\mathbf{m}^{(i)} - \mathbf{P}_{\Sgeqss\cap\Sw}[\mathbf{m}^{(i)}]\|_2^2 - \|\mathbf{a}^{(i+1)} - \mathbf{P}_{\Sgeqss\cap\Sw}[\mathbf{a}^{(i+1)}]\|_2^2 \geq \|\mathbf{m}^{(i)}-\mathbf{a}^{(i+1)}\|_2^2, \quad i \geq 0.
\end{aligned}
\label{eq:PropIII.1_23}
\end{equation}
In \eqref{eq:PropIII.1_22}, we used the \DefRef{def:ProjOp} of the projection operator (the second line) and the containing-half-space inequality \eqref{eq:Prop3_3} along with $\mathbf{m}^{(i)}=\mathbf{P}_{\Sgeqss}[\mathbf{a}^{(i)}]$ (the sixth line). Replacing $\mathbf{a}^{(i+1)}$ by $\mathbf{m}^{(i+1)}$ and $\mathbf{m}^{(i)}$ by $\mathbf{a}^{(i+1)}$ and repeating the same steps as in \eqref{eq:PropIII.1_22}, we obtain
\begin{equation}
\begin{aligned}
\|\mathbf{a}^{(i+1)} - \mathbf{P}_{\Sgeqss\cap\Sw}[\mathbf{a}^{(i+1)}]\|_2^2 - \|\mathbf{m}^{(i+1)} - \mathbf{P}_{\Sgeqss\cap\Sw}[\mathbf{m}^{(i+1)}]\|_2^2 \geq \|\mathbf{m}^{(i+1)}-\mathbf{a}^{(i+1)}\|_2^2, \quad i \geq 0.
\end{aligned}
\label{eq:PropIII.1_24}
\end{equation}

Combining \eqref{eq:PropIII.1_18} with \eqref{eq:PropIII.1_23} and \eqref{eq:PropIII.1_19} with \eqref{eq:PropIII.1_24}, we obtain, respectively,
\begin{equation}
\begin{aligned}
\bigg(1-\frac{1}{(1+\|\mathbf{x}-\mathbf{m}^{(0)}\|_2 / \beta)^2} \bigg) \cdot \|\mathbf{m}^{(i)} - \mathbf{P}_{\Sgeqss\cap\Sw}[\mathbf{m}^{(i)}]\|_2^2 \geq \|\mathbf{a}^{(i+1)} - \mathbf{P}_{\Sgeqss\cap\Sw}[\mathbf{a}^{(i+1)}]\|_2^2
\end{aligned}
\label{eq:PropIII.1_25}
\end{equation}
and
\begin{equation}
\begin{aligned}
\bigg(1-\frac{1}{(1+\|\mathbf{x}-\mathbf{m}^{(0)}\|_2 / \beta)^2} \bigg) \cdot \|\mathbf{a}^{(i+1)} - \mathbf{P}_{\Sgeqss\cap\Sw}[\mathbf{a}^{(i+1)}]\|_2^2 \geq \|\mathbf{m}^{(i+1)} - \mathbf{P}_{\Sgeqss\cap\Sw}[\mathbf{m}^{(i+1)}]\|_2^2,
\end{aligned}
\label{eq:PropIII.1_26}
\end{equation}
which then lead to
\begin{equation}
\begin{aligned}
\underbrace{\bigg(1-\frac{1}{(1+\|\mathbf{x}-\mathbf{m}^{(0)}\|_2 / \beta)^2} \bigg)}_{r < 1} \cdot \|\mathbf{m}^{(i)} - \mathbf{P}_{\Sgeqss\cap\Sw}[\mathbf{m}^{(i)}]\|_2 \geq \|\mathbf{m}^{(i+1)} - \mathbf{P}_{\Sgeqss\cap\Sw}[\mathbf{m}^{(i+1)}]\|_2
\end{aligned}
\label{eq:PropIII.1_27}
\end{equation}
for $i \geq 0$. Starting with $i=0$ and applying \eqref{eq:PropIII.1_27} iteratively, we get
\begin{equation}
\begin{aligned}
r^i \cdot \|\mathbf{m}^{(0)} - \mathbf{P}_{\Sgeqss\cap\Sw}[\mathbf{m}^{(0)}]\|_2 \geq \|\mathbf{m}^{(i)} - \mathbf{P}_{\Sgeqss\cap\Sw}[\mathbf{m}^{(i)}]\|_2, \quad i \geq 0.
\end{aligned}
\label{eq:PropIII.1_28}
\end{equation}
According to the triangle inequality \eqref{eq:TriIneq},
\begin{equation}
\begin{aligned}
\|\mathbf{m}^{(i)} - \mathbf{m}^\dagger\|_2 \leq \|\mathbf{m}^{(i)} - \mathbf{P}_{\Sgeqss\cap\Sw}[\mathbf{m}^{(i)}]\|_2 + \|\mathbf{m}^\dagger - \mathbf{P}_{\Sgeqss\cap\Sw}[\mathbf{m}^{(i)}]\|_2
\end{aligned}
\label{eq:PropIII.1_29}
\end{equation}
and [see \eqref{eq:TriIneq2}]
\begin{equation}
\begin{aligned}
\|\mathbf{m}^{(j)} - \mathbf{m}^\dagger\|_2 \geq \big|\|\mathbf{m}^{(j)} - \mathbf{P}_{\Sgeqss\cap\Sw}[\mathbf{m}^{(i)}]\|_2 - \|\mathbf{m}^\dagger - \mathbf{P}_{\Sgeqss\cap\Sw}[\mathbf{m}^{(i	)}]\|_2\big|.
\end{aligned}
\label{eq:PropIII.1_30}
\end{equation}
\eqref{eq:PropIII.1_30} and \eqref{eq:PropIII.1_15} together imply that a sequence
\begin{equation}
\|\mathbf{m}^{(0)} - \mathbf{P}_{\Sgeqss\cap\Sw}[\mathbf{m}^{(i)}]\|_2, \|\mathbf{m}^{(1)} - \mathbf{P}_{\Sgeqss\cap\Sw}[\mathbf{m}^{(i)}]\|_2, \ldots, \|\mathbf{m}^{(j)} - \mathbf{P}_{\Sgeqss\cap\Sw}[\mathbf{m}^{(i)}]\|_2, \ldots
\label{eq:PropIII.1_31}
\end{equation}
converges to $\|\mathbf{m}^\dagger - \mathbf{P}_{\Sgeqss\cap\Sw}[\mathbf{m}^{(i)}]\|_2$ for every $i \geq 0$. On the other hand, this sequence is monotonically decreasing [see \eqref{eq:PropIII.1_6}], and thus, by the monotone convergence theorem (\PropRef{prop:MonConvTh}), it converges to its greatest lower bound, i.e., $\|\mathbf{m}^{(j)} - \mathbf{P}_{\Sgeqss\cap\Sw}[\mathbf{m}^{(i)}]\|_2 \geq \|\mathbf{m}^\dagger - \mathbf{P}_{\Sgeqss\cap\Sw}[\mathbf{m}^{(i)}]\|_2$ for all $i,j \geq 0$. Consequently, \eqref{eq:PropIII.1_29} reduces to
\begin{equation}
\begin{aligned}
\|\mathbf{m}^{(i)} - \mathbf{m}^\dagger\|_2 \leq 2 \cdot \|\mathbf{m}^{(i)} - \mathbf{P}_{\Sgeqss\cap\Sw}[\mathbf{m}^{(i)}]\|_2.
\end{aligned}
\label{eq:PropIII.1_32}
\end{equation}
Applying \eqref{eq:PropIII.1_32} to \eqref{eq:PropIII.1_28}, we finally obtain
\begin{equation}
\begin{aligned}
\gamma \cdot r^i \geq \|\mathbf{m}^{(i)} - \mathbf{m}^\dagger\|_2, \quad i \geq 0,
\end{aligned}
\label{eq:PropIII.1_33}
\end{equation}
where $\gamma = 2 \cdot \|\mathbf{m}^{(0)} - \mathbf{P}_{\Sgeqss\cap\Sw}[\mathbf{m}^{(0)}]\|_2 > 0$.
\end{proof}

\begin{remark}
1) The convergence proof of the iterative scheme AP-B relies entirely on the convexity and closedness of the constraint sets. Therefore, this algorithm extends to more general sets than $\Sgeqs$ and $\Sw$. 2) The geometric nature of the convergence requires additionally that the interior of at least one of the constraint sets is nonempty and shares elements with the other set.
\end{remark}


\begin{proposition_III_2}
Consider $\mathbf{m} \in \Mwm$ and $\mathbf{c} \in \Cd$ with $|c_j| = \sum_{k=1}^{n/\nu}(\tilde{c}_{\nu \cdot k} \cdot e^{\imath 2 \pi \nu (k-1) (j-1) / n})$, where $\tilde{c}_{\nu \cdot k} \in \mathbb{C}$ and $n / \nu \in \mathbb{N}$. If $\varpi \geq \omega$ and $\nu \geq \varpi + \omega - 1$, then  a sequence $\mathbf{m}^{(0)}, \mathbf{m}^{(1)}, \ldots, \mathbf{m}^{(i)},\ldots$ formed by the AP-B algorithm for $\epsilon_{tol}=0$ and $N_{iter} \to +\infty$ converges to $\mathbf{m}$.
\label{prop:APBSolConv_b}
\end{proposition_III_2}

\begin{proof}
The proof relies on two auxiliary results that apply to the $\mathbf{m}$ and $\mathbf{c}$ specified in the proposition:
\begin{itemize}
\item For every $q \in \mathbb{R}$,
\begin{equation}
\mathbf{z} = q + (|\mathbf{c}|-q) \circ \theta (|\mathbf{c}|-q) \implies z_j = \sum_{k=1}^{n/\nu}\big(\tilde{z}_{\nu \cdot k} \cdot e^{\imath 2 \pi \nu (k-1)(j-1) / n}\big), \quad j \in \mathcal{I}_n.
\label{eq:PropIII.1b_1}
\end{equation}
\item For $\mathbf{z}$ defined by \eqref{eq:PropIII.1b_1},
\begin{equation}
\mathbf{P}_{\Sw}[\mathbf{m} \circ \mathbf{z}] = \langle \mathbf{z} \rangle \cdot \mathbf{m},
\label{eq:PropIII.1b_2}
\end{equation}
where, $\langle |\mathbf{z}| \rangle = \frac{1}{n}\sum_{i=1}^n|z_i|$.
\end{itemize}

To show \eqref{eq:PropIII.1b_1}, consider a $\mathbf{g} \in \mathbb{R}^n$ whose elements form a periodic sequence with the fundamental frequency $\nu \in \mathbb{N}$ such that $n / \nu \in \mathbb{N}$. Like any element of $\mathbb{R}^n$, $\mathbf{g}$ can be expressed through its DFT:
\begin{equation}
g_j = \frac{1}{\sqrt{n}}\sum_{k=1}^{n}\big(\tilde{\tilde{g}}_k \cdot e^{\imath 2 \pi (k-1)(j-1) / n}\big), \quad j \in \mathcal{I}_n.
\label{eq:PropIII.1b_3}
\end{equation}
On the other hand, the periodicity of $g_1, g_2, \ldots, g_n$ implies that, for every $k \in \mathcal{I}_n$,
\begin{equation}
\begin{aligned}
\tilde{\tilde{g}}_k = (\mathbf{F} \mathbf{g})_k &= \frac{1}{\sqrt{n}}\sum_{j=1}^{n}\big(g_j \cdot e^{-\imath 2 \pi (k-1)(j-1) / n}\big)\\
&= \frac{1}{\sqrt{n}}\sum_{j=1}^{n/\nu}g_j \sum_{l=1}^{\nu} e^{-\imath 2 \pi ((l-1) \cdot (n/\nu)+ (j-1))(k-1) / n}\\
&= \frac{1}{\sqrt{n}}\sum_{j=1}^{n/\nu} \bigg( g_j \cdot e^{-\imath 2 \pi (k-1)(j-1) / n} \sum_{l=1}^{\nu} \big(e^{-\imath 2 \pi (k-1) / \nu}\big)^{l-1} \bigg)\\
&= \underbrace{\bigg(\frac{1-e^{-\imath 2 \pi (k-1)}}{1-e^{-\imath 2 \pi (k-1)/\nu}}\bigg)}_{=0,~\mathrm{if}~((k-1)/\nu) \notin \mathbb{N}} \cdot \frac{1}{\sqrt{n}}\sum_{j=1}^{n/\nu} \big( g_j \cdot e^{-\imath 2 \pi (k-1)(j-1) / n} \big).
\end{aligned}
\label{eq:PropIII.1b_4}
\end{equation}
Combining \eqref{eq:PropIII.1b_3} and \eqref{eq:PropIII.1b_4} gives
\begin{equation}
g_j = \frac{1}{\sqrt{n}}\sum_{k=1}^{n/\nu}\big(\tilde{g}_{\nu \cdot k} \cdot e^{\imath 2 \pi \nu (k-1)(j-1) / n}\big), \quad j \in \mathcal{I}_n,
\label{eq:PropIII.1b_5}
\end{equation}
where $\tilde{g}_{\nu \cdot k} = \tilde{\tilde{g}}_{\nu \cdot (k-1)+1}$. Now, note that $|\mathbf{c}|$ defined in the proposition is also periodic with the fundamental frequency $\nu$ such that $n/\nu \in \mathbb{N}$. If so, then the same holds for $\mathbf{z}$ defined by \eqref{eq:PropIII.1b_1} because adding a constant or rectifying a function does not change its periodicity properties. Combining this result with \eqref{eq:PropIII.1b_5} validates the claim of \eqref{eq:PropIII.1b_1}.

To show \eqref{eq:PropIII.1b_2}, consider $\mathbf{W}_\varpi \mathbf{F} (\mathbf{m} \circ \mathbf{z})$. For every $r \in \mathcal{I}_n$, we have 
\begin{equation}
\begin{aligned}
(\mathbf{W}_\varpi \mathbf{F} (\mathbf{m} \circ \mathbf{z}))_r &= \frac{(\mathbf{W}_\varpi)_{rr}}{\sqrt{n}} \sum_{j=1}^{n} \big(m_j \cdot z_j \cdot e^{-\imath 2 \pi (r-1)(j-1)/n} \big) \\
&= \frac{(\mathbf{W}_\varpi)_{rr}}{\sqrt{n}} \sum_{j=1}^{n} \bigg(m_j \cdot \bigg[\sum_{l=1}^n \bigg( z_l \cdot \underbrace{\frac{1}{n} \sum_{k=1}^n e^{\imath 2 \pi (k-1)(j-l)/n}}_{\delta_{j,l}} \bigg)\bigg] \cdot e^{-\imath 2 \pi (r-1)(j-1)/n} \bigg)\\
&=  \frac{(\mathbf{W}_\varpi)_{rr}}{\sqrt{n}} \sum_{k=1}^{n} \bigg(\frac{1}{\sqrt{n}} \sum_{l=1}^n \big( z_l \cdot e^{-\imath 2 \pi (k-1)(l-1)/n} \big) \cdot \frac{1}{\sqrt{n}} \sum_{j=1}^n \big(m_j \cdot e^{-\imath 2 \pi (j-1)(r-k)/n}\big) \bigg)\\
&= \frac{(\mathbf{W}_\varpi)_{rr}}{\sqrt{n}} \sum_{k=1}^{n} \big(\tilde{\tilde{z}}_k \cdot \tilde{\tilde{m}}_{r-k} \big)\\
&= \frac{\tilde{\tilde{z}}_1 \cdot (\mathbf{W}_\varpi)_{rr}}{\sqrt{n}} \cdot \tilde{\tilde{m}}_r + \underbrace{\frac{(\mathbf{W}_\varpi)_{rr}}{\sqrt{n}} \sum_{k=\nu+1}^{n-\nu+1} \big(\tilde{\tilde{z}}_k \cdot \tilde{\tilde{m}}_{r-k}}_{=0 \impliedby \omega \leq \varpi \leq \nu-\omega+1} \big) \\
&= \langle \mathbf{z} \rangle \cdot \tilde{\tilde{m}}_r.
\end{aligned}
\label{eq:PropIII.1b_6}
\end{equation}
When writing the second equality above, we used the orthonormality of $\mathbf{F}^{-1}$. Combining \eqref{eq:PropIII.1b_6} with the definition of $\mathbf{P}_{\Sw}[\ldots]$ (see (10) in the main text), we obtain \eqref{eq:PropIII.1b_2}.

After establishing \eqref{eq:PropIII.1b_1} and \eqref{eq:PropIII.1b_2}, consider the sequences $\mathbf{m}^{(0)}, \mathbf{m}^{(1)}, \ldots, \mathbf{m}^{(i)}, \ldots$ and $\mathbf{a}^{(0)}, \mathbf{a}^{(1)}, \ldots, \mathbf{a}^{(i)}, \ldots$. Note that
\begin{equation}
\begin{aligned}
\mathbf{m}^{(0)} &= \mathbf{m} \circ |\mathbf{c}| \\
&= \mathbf{m} \circ \big( q^{(0)} + (|\mathbf{c}|-q^{(0)}) \circ \theta (|\mathbf{c}|-q^{(0)}) \big),
\end{aligned}
\label{eq:PropIII.1b_7}
\end{equation}
where $q^{(0)}=0$. Hence, $\mathbf{m}^{(0)}$ can be expressed as an elementwise product of the true modulator $\mathbf{m}$ and a vector that satisfies \eqref{eq:PropIII.1b_1}. Let us now assume that, for some $i \geq 1$, $\mathbf{m}^{(i)}$ can be expressed as
\begin{equation}
\mathbf{m}^{(i)} = \mathbf{m} \circ \big( q^{(i)} + (|\mathbf{c}|-q^{(i)}) \circ \theta (|\mathbf{c}|-q^{(i)}) \big).
\label{eq:PropIII.1b_8}
\end{equation}
Then, by \eqref{eq:PropIII.1b_2},
\begin{equation}
\mathbf{a}^{(i+1)} = \mathbf{P}_{\Sw}[\mathbf{m^{(i)}}] = q^{(i+1)} \cdot \mathbf{m},
\label{eq:PropIII.1b_9}
\end{equation}
where,
\begin{equation}
q^{(i+1)} = \langle q^{(i)} + (|\mathbf{c}|-q^{(i)}) \circ \theta (|\mathbf{c}|-q^{(i)}) \rangle.
\label{eq:PropIII.1b_10}
\end{equation}
Further, by the definition of $\mathbf{P}_{\Sgeqs}[\ldots]$ (see (9) in the main text),
\begin{equation}
\begin{aligned}
\mathbf{m}^{(i+1)} &= \mathbf{P}_{\Sgeqs}[\mathbf{a}^{(i+1)}] \\
&= \mathbf{m} \circ \big( q^{(i+1)} + (|\mathbf{c}|-q^{(i+1)}) \circ \theta (|\mathbf{c}|-q^{(i+1)}) \big),
\end{aligned}
\label{eq:PropIII.1b_11}
\end{equation}
i.e., $\mathbf{m}^{(i+1)}$ can also be expressed as the product of $\mathbf{m}$ and a vector that satisfies \eqref{eq:PropIII.1b_1}. Hence, we conclude by using mathematical induction that, for every $i \geq 1$,
\begin{equation}
\mathbf{a}^{(i)} = q^{(i)} \cdot \mathbf{m},
\label{eq:PropIII.1b_12}
\end{equation}
and
\begin{equation}
q^{(i)} = \langle q^{(i-1)} + (|\mathbf{c}|-q^{(i-1)}) \circ \theta (|\mathbf{c}|-q^{(i-1)}) \rangle,
\label{eq:PropIII.1b_13}
\end{equation}
with $q_0 = 0$.

Next, observe that, by \eqref{eq:PropIII.1b_13},
\begin{equation}
\begin{aligned}
q^{(i)} - q^{(i-1)} &= \langle (|\mathbf{c}|-q^{(i-1)}) \circ \theta (|\mathbf{c}|-q^{(i-1)}) \rangle \\
&\geq 0,
\end{aligned}
\label{eq:PropIII.1b_14}
\end{equation}
i.e., the sequence $q^{(0)}, q^{(1)}, \ldots, q^{(i)}, \ldots$ is monotonically increasing. Moreover, it follows from \eqref{eq:PropIII.1b_13} that, for every $i \geq 1$,
\begin{equation}
\begin{aligned}
q^{(i-1)} \leq 1 \implies q^{(i)} &= q^{(i-1)} + \langle (|\mathbf{c}|-q^{(i-1)}) \circ \theta (|\mathbf{c}|-q^{(i-1)}) \rangle \\
&\leq q^{(i-1)} + (1-q^{(i-1)}) \cdot \underbrace{\theta(1-q^{(i-1)})}_{=1} = 1.
\end{aligned}
\label{eq:PropIII.1b_15}
\end{equation}
Taken together with $q^{(0)}=0$, \eqref{eq:PropIII.1b_15} implies that the sequence $q^{(0)}, q^{(1)}, \ldots, q^{(i)}, \ldots$ is bounded from above by 1. Hence, by the monotone convergence theorem (see \PropRef{prop:MonConvTh}), $q^{(0)}, q^{(1)}, \ldots, q^{(i)}, \ldots$ converges to its least upper bound $\bar{q} \leq 1$. If we assume that $\bar{q} < 1$, then the convergence of $q^{(0)}, q^{(1)}, \ldots, q^{(i)}, \ldots$ to $\bar{q}$ implies that there exists an $N$ such that
\begin{equation}
\bar{q} - q^{(N)} \leq \langle (|\mathbf{c}|-\bar{q}) \circ \theta (|\mathbf{c}|-\bar{q}) \rangle / 2.
\label{eq:PropIII.1b_16}
\end{equation}
By \eqref{eq:PropIII.1b_13},
\begin{equation}
\begin{aligned}
q^{(N+1)} - q^{(N)} &= \langle (|\mathbf{c}|-q^{(N)}) \circ \theta (|\mathbf{c}|-q^{(N)}) \rangle \\
&\geq \langle (|\mathbf{c}|-\bar{q}) \circ \theta (|\mathbf{c}|-\bar{q}) \rangle \\
&> \langle (|\mathbf{c}|-\bar{q}) \circ \theta (|\mathbf{c}|-\bar{q}) \rangle / 2,
\label{eq:PropIII.1b_17}
\end{aligned}
\end{equation}
which means that $q^{(N+1)} > \bar{q}$, i.e., the initial assumption that $\bar{q}<1$ is incorrect. Therefore, $\bar{q}=1$, and $q^{(0)}, q^{(1)}, \ldots, q^{(i)}, \ldots$ converges to 1, i.e., 
\begin{equation}
\forall \epsilon > 0~\exists N(\epsilon): i>N(\epsilon) \implies |1-q^{(N)}| < \epsilon.
\label{eq:PropIII.1b_18}
\end{equation}

Finally, note that
\begin{equation}
\begin{aligned}
|1-q^{(N)}| < \epsilon &\implies \| \mathbf{m} \|_2 \cdot |1-q^{(N)}| < \underbrace{\| \mathbf{m} \|_2 \cdot \epsilon}_{\epsilon'} \\
&\implies \| \mathbf{m} - \mathbf{m} \cdot q^{(N)}\|_2 < \epsilon'\\
&\implies \underbrace{\| \mathbf{m} - \mathbf{a}^{(N)}\|_2 < \epsilon'}_{\textrm{by}~\eqref{eq:PropIII.1b_12}},
\end{aligned}
\label{eq:PropIII.1b_19}
\end{equation}
i.e., \eqref{eq:PropIII.1b_18} implies that the sequence $\mathbf{a}^{(0)}, \mathbf{a}^{(1)}, \ldots, \mathbf{a}^{(i)}, \ldots$ converges to $\mathbf{m}$. In the light of \eqref{eq:PropIII.1_5} and \eqref{eq:PropIII.1_6} in the proof of \hyperref[prop:APBSolConv_]{\textit{Proposition~III.1}}, this result allows concluding that $\mathbf{m}^{(0)}, \mathbf{m}^{(1)}, \ldots, \mathbf{m}^{(i)}, \ldots$ also converges to $\mathbf{m}$.
\end{proof}

\subsection*{\textbf{AP-A algorithm}}


\begin{proposition_III_3}
A sequence $\mathbf{m}^{(0)},\mathbf{m}^{(1)}, \ldots, \mathbf{m}^{(i)},\ldots$ formed by the AP-A algorithm for $\epsilon_{tol} = 0$ and $N_{iter} \to +\infty$ converges to some $\mathbf{m}^\dagger \in \Sgeqs \cap \Sw$. The convergence is monotonic, i.e., $\| \mathbf{m}^{(i+1)} - \mathbf{m}^\dagger \|_2 \leq \| \mathbf{m}^{(i)} - \mathbf{m}^\dagger \|_2$ for $i \geq 0$.
\end{proposition_III_3}

\begin{proof}
Analogously to the AP-B algorithm, if the sequence $\mathbf{m}^{(0)},\mathbf{m}^{(1)}, \ldots, \mathbf{m}^{(i)}, \ldots$ terminates with some finite $i=N$, then we have $\mathbf{m}^{(N)} = \mathbf{m}^\dagger \in \Sgeqs \cap \Sw$. Therefore, we next consider the case when the sequence is infinite. The main idea behind the proof is to show that the inequalities \eqref{eq:PropIII.1_2} and \eqref{eq:PropIII.1_4} apply not only to the AP-B but also to the AP-A algorithm. When that is established, we can proceed along the path of the convergence proof of the AP-B scheme.

To this end, we first derive some auxiliary (in)equalities. Specifically, it follows from the definition of the operator $\mathbf{P}_{\Sw}$ [see \eqref{eq:MathPrel3}] that, for all $\mathbf{z}, \mathbf{y} \in \mathbb{R}^n$,
\begin{equation}
\begin{aligned}
\langle \mathbf{z}, \mathbf{P}_{\Sw}[\mathbf{y}] \rangle &= \mathbf{z}^\text{T} \mathbf{F}^{-1} \mathbf{W}_\varpi \mathbf{F} \mathbf{y} = \mathbf{z}^\text{T} \mathbf{F}^{-1} \mathbf{W}_\varpi \mathbf{W}_\varpi \mathbf{F} \mathbf{y} = \mathbf{z}^\text{T} \mathbf{F}^{-1} \mathbf{W}_\varpi \mathbf{F} \mathbf{F}^{-1} \mathbf{W}_\varpi \mathbf{F} \mathbf{y}\\
&= (\mathbf{z}^\text{T} \mathbf{F}^{-1} \mathbf{W}_\varpi \mathbf{F})^* (\mathbf{F}^{-1} \mathbf{W}_\varpi \mathbf{F} \mathbf{y}) = (\mathbf{z}^\text{T} \mathbf{F} \mathbf{W}_\varpi \mathbf{F}^{-1}) (\mathbf{F}^{-1} \mathbf{W}_\varpi \mathbf{F} \mathbf{y}) \\
&= ((\mathbf{F} \mathbf{W}_\varpi \mathbf{F}^{-1})^\text{T}\mathbf{z})^\text{T} (\mathbf{F}^{-1} \mathbf{W}_\varpi \mathbf{F} \mathbf{y}) = (\mathbf{F}^{-1} \mathbf{W}_\varpi \mathbf{F} \mathbf{z})^\text{T} (\mathbf{F}^{-1} \mathbf{W}_\varpi \mathbf{F} \mathbf{y})\\
&= \langle \mathbf{P}_{\Sw}[\mathbf{z}], \mathbf{P}_{\Sw}[\mathbf{y}] \rangle.
\label{eq:PropIII.2_1}
\end{aligned}
\end{equation}
Here, T and * mark, respectively, the transposition and complex conjugation. We used the following properties of matrices $\mathbf{W}_\varpi$ and $\mathbf{F}$ in \eqref{eq:PropIII.2_1}: 1) $\mathbf{W}_\varpi \mathbf{W}_\varpi = \mathbf{W}_\varpi$; 2) $\mathbf{W}_\varpi^* = \mathbf{W}_\varpi^\text{T} = \mathbf{W}_\varpi$; 3) $\mathbf{F}^* = \mathbf{F}^{-1}$; and 4) $ \mathbf{F}^\text{T} = \mathbf{F}$. The result of \eqref{eq:PropIII.2_1} can be rewritten as
\begin{equation}
\langle \mathbf{z}-\mathbf{P}_{\Sw}[\mathbf{z}], \mathbf{P}_{\Sw}[\mathbf{y}] \rangle = 0.
\label{eq:PropIII.2_2}
\end{equation}
\eqref{eq:PropIII.2_2} is a particular instance of a more general result that the difference between any $\mathbf{z} \in \mathbb{R}^n$ and its projection onto a linear subspace of $\mathbb{R}^n$ is perpendicular to any element of that subspace. Using \eqref{eq:PropIII.2_2}, we obtain the following:
\begin{equation}
\begin{aligned}
\|\mathbf{z}\|_2^2 &= \langle \mathbf{z}, \mathbf{z} \rangle = \langle \mathbf{P}_{\Sw}[\mathbf{z}]+(\mathbf{z}-\mathbf{P}_{\Sw}[\mathbf{z}]), \mathbf{P}_{\Sw}[\mathbf{z}]+(\mathbf{z}-\mathbf{P}_{\Sw}[\mathbf{z}]) \rangle \\
&= \|\mathbf{P}_{\Sw}[\mathbf{z}]\|_2^2 + \|\mathbf{z}- \mathbf{P}_{\Sw}[\mathbf{z}]\|_2^2 \\
&\geq \|\mathbf{P}_{\Sw}[\mathbf{z}]\|_2^2.
\label{eq:PropIII.2_3}
\end{aligned}
\end{equation}
Applying \eqref{eq:PropIII.2_3} to the \textit{line~6} of the AP-A algorithm $\big(\lambda = \| \mathbf{m}^{(i-1)} - \mathbf{a}^{(i-1)} \|_2^2 / \|\mathbf{b}^{(i)} \|_2^2\big)$, we get
\begin{equation}
\lambda \geq 1,
\label{eq:PropIII.2_4}
\end{equation}
with the equality holding if and only if $(\mathbf{m}^{(i-1)}-\mathbf{a}^{(i-1)})$ is mapped by $\mathbf{P}_{\Sw}$ to itself, i.e., $\mathbf{m}^{(i-1)}\in\Sw$. However, this would mean that the convergence was reached at iteration $i-1$. Finally, we note that \textit{line~7} of the AP-A algorithm $\big(\mathbf{a}^{(i)} = \mathbf{a}^{(i-1)} + \lambda \cdot \mathbf{b}^{(i)}\big)$ implies
\begin{equation}
(\mathbf{P}_{\Sw}[\mathbf{m}^{(i)}]-\mathbf{a}^{(i+1)}) = (\lambda-1) \cdot (\mathbf{a}^{(i)}-\mathbf{P}_{\Sw}[\mathbf{m}^{(i)}])
\label{eq:PropIII.2_5}
\end{equation}
and
\begin{equation}
\begin{aligned}
\langle \mathbf{m}^{(i)}-\mathbf{a}^{(i+1)}, \mathbf{m}^{(i)}-\mathbf{a}^{(i)} \rangle &= \langle \mathbf{m}^{(i)}-\mathbf{a}^{(i)}, \mathbf{m}^{(i)}-\mathbf{a}^{(i)} \rangle - \lambda \cdot \langle \mathbf{P}_{\Sw}[\mathbf{m}^{(i)}]-\mathbf{a}^{(i+1)}, \mathbf{m}^{(i)}-\mathbf{a}^{(i)}) \rangle \\
&= \|\mathbf{m}^{(i)}-\mathbf{a}^{(i)}\|_2^2 - \frac{\|\mathbf{m}^{(i)}-\mathbf{a}^{(i)}\|_2^2}{\|\mathbf{P}_{\Sw}[\mathbf{m}^{(i)}]-\mathbf{a}^{(i)}\|_2^2} \cdot \|\mathbf{P}_{\Sw}[\mathbf{m}^{(i)}]-\mathbf{a}^{(i)}\|_2^2 \\
&= 0.
\end{aligned} 
\label{eq:PropIII.2_6}
\end{equation}
In \eqref{eq:PropIII.2_6} we applied \eqref{eq:PropIII.2_1} to the term $\langle \mathbf{P}_{\Sw}[\mathbf{m}^{(i)}]-\mathbf{a}^{(i+1)}, \mathbf{m}^{(i)}-\mathbf{a}^{(i)}) \rangle$.

We are now ready to prove that \eqref{eq:PropIII.1_2} and \eqref{eq:PropIII.1_4} hold for the AP-A algorithm. In particular, we have
\begin{equation}
\begin{aligned}
\| \mathbf{x} - \mathbf{m}^{(i)} \|_2^2 &= \| \mathbf{x} - \mathbf{a}^{(i+1)} + \mathbf{a}^{(i+1)} - \mathbf{m}^{(i)} \|_2^2 \\
&= \| \mathbf{x} - \mathbf{a}^{(i+1)} \|_2^2 + \| \mathbf{a}^{(i+1)} - \mathbf{m}^{(i)} \|_2^2 + 2 \cdot \langle \mathbf{m}^{(i)} - \mathbf{a}^{(i+1)} , \mathbf{a}^{(i+1)} - \mathbf{x} \rangle.
\label{eq:PropIII.2_7}
\end{aligned}
\end{equation}
The first two terms in the second line of \eqref{eq:PropIII.2_7} are nonnegative by the definition of the norm. The last term is nonnegative too (note that $\mathbf{x} \in \Sgeqs \cap \Sw$):
\begin{equation}
\begin{aligned}
\langle \mathbf{m}^{(i)} - \mathbf{a}^{(i+1)} , \mathbf{a}^{(i+1)} &- \mathbf{x} \rangle = \langle \mathbf{m}^{(i)} -\mathbf{P}_{\Sw}[\mathbf{m}^{(i)}] + \mathbf{P}_{\Sw}[\mathbf{m}^{(i)}] - \mathbf{a}^{(i+1)} , \mathbf{a}^{(i+1)} - \mathbf{x} \rangle\\[6pt]
&= \underbrace{\langle \mathbf{m}^{(i)} -\mathbf{P}_{\Sw}[\mathbf{m}^{(i)}] , \mathbf{a}^{(i+1)} - \mathbf{x} \rangle}_{=0 \text { by } \eqref{eq:PropIII.2_2}} + \langle \mathbf{P}_{\Sw}[\mathbf{m}^{(i)}] - \mathbf{a}^{(i+1)} , \mathbf{a}^{(i+1)} - \mathbf{x} \rangle\\
&= \underbrace{(\lambda-1) \cdot \langle \mathbf{a}^{(i)} - \mathbf{P}_{\Sw}[\mathbf{m}^{(i)}], \mathbf{a}^{(i+1)} - \mathbf{x} \rangle}_{\text{by } \eqref{eq:PropIII.2_5}}\\
&= (\lambda-1) \cdot \langle \mathbf{a}^{(i)} - \mathbf{m}^{(i)} + \mathbf{m}^{(i)} - \mathbf{P}_{\Sw}[\mathbf{m}^{(i)}], \mathbf{a}^{(i+1)} - \mathbf{x} \rangle\\
&= (\lambda-1) \cdot \underbrace{\langle \mathbf{m}^{(i)} -\mathbf{P}_{\Sw}[\mathbf{m}^{(i)}] , \mathbf{a}^{(i+1)} - \mathbf{x} \rangle}_{=0 \text { by } \eqref{eq:PropIII.2_2}} + (\lambda-1) \cdot \langle \mathbf{a}^{(i)} - \mathbf{m}^{(i)} , \mathbf{a}^{(i+1)} - \mathbf{x} \rangle\\
&= (\lambda-1) \cdot \langle \mathbf{a}^{(i)} - \mathbf{m}^{(i)} , \mathbf{a}^{(i+1)} - \mathbf{m}^{(i)} + \mathbf{m}^{(i)} - \mathbf{x} \rangle\\
&= (\lambda-1) \cdot \underbrace{\langle \mathbf{a}^{(i)} - \mathbf{m}^{(i)} , \mathbf{a}^{(i+1)} - \mathbf{m}^{(i)} \rangle}_{=0 \text{ by }\eqref{eq:PropIII.2_6}} + \underbrace{(\lambda-1)}_{\geq0 \text{ by } \eqref{eq:PropIII.2_4}} \cdot \underbrace{\langle \mathbf{a}^{(i)} - \mathbf{m}^{(i)} , \mathbf{m}^{(i)} - \mathbf{x} \rangle}_{\geq0 \text{ by } \eqref{eq:Prop3_3}}\\
&\geq 0.
\label{eq:PropIII.2_8}
\end{aligned}
\end{equation}
Therefore,
\begin{equation}
\begin{aligned}
\| \mathbf{x} - \mathbf{m}^{(i)} \|_2^2 \geq \| \mathbf{x} - \mathbf{a}^{(i+1)} \|_2^2 + \| \mathbf{a}^{(i+1)} - \mathbf{m}^{(i)} \|_2^2.
\label{eq:PropIII.2_9}
\end{aligned}
\end{equation}
The derivation of the inequality \eqref{eq:PropIII.1_4} for the AP-A algorithm is equivalent to that for the AP-B:
\begin{equation}
\begin{aligned}
\| \mathbf{x} - \mathbf{a}^{(i+1)} \|_2^2 &= \| \mathbf{x} - \mathbf{m}^{(i+1)} + \mathbf{m}^{(i+1)} - \mathbf{a}^{(i+1)} \|_2^2 \\
&= \| \mathbf{x} - \mathbf{m}^{(i+1)} \|_2^2 + \| \mathbf{m}^{(i+1)} - \mathbf{a}^{(i+1)} \|_2^2 + 2 \cdot \underbrace{\langle \mathbf{a}^{(i+1)} - \mathbf{m}^{(i+1)} , \mathbf{m}^{(i+1)} - \mathbf{x} \rangle}_{\geq0 \text{ by } \eqref{eq:Prop3_3}}.
\label{eq:PropIII.2_10}
\end{aligned}
\end{equation}
Thus,
\begin{equation}
\begin{aligned}
\| \mathbf{x} - \mathbf{a}^{(i+1)} \|_2^2 \geq \| \mathbf{x} - \mathbf{m}^{(i+1)} \|_2^2 + \| \mathbf{m}^{(i+1)} - \mathbf{a}^{(i+1)} \|_2^2.
\label{eq:PropIII.2_11}
\end{aligned}
\end{equation}

The rest of the proof follows step by step the proof of \emph{\hyperref[prop:APBSolConv_]{Proposition~III.1}}. Indeed, starting with the inequalities \eqref{eq:PropIII.1_2} and \eqref{eq:PropIII.1_4}, to which \eqref{eq:PropIII.2_9} and \eqref{eq:PropIII.2_11} are equivalent, the proof of \emph{\hyperref[prop:APBSolConv_]{Proposition~III.1}} proceeds based on them and the closedness and convexity of the sets $\Sgeqs$ and $\Sw$ entirely. The monotonicity of the convergence follows from an equivalent of \eqref{eq:PropIII.1_14}, which is derived as a part of the proof of the convergence itself.

\end{proof}

\begin{remark}
1) Concerning the set $\Sgeqs$, the proof above relies entirely on its closedness and convexity. Thus, the AP-A algorithm is still valid under these more general assumptions. 2) The expressions \eqref{eq:PropIII.2_1}\,--\,\eqref{eq:PropIII.2_6} apply to projection operators onto any linear space. Hence, the AP-A algorithm still works if $\Sw$ is replaced by an arbitrary linear space.

\end{remark}


\begin{proposition_III_4}
Consider $\mathbf{m} \in \Mwm$ and $\mathbf{c} \in \Cd$ with $|c_j| = \sum_{k=1}^{n/\nu}(\tilde{c}_{\nu \cdot k} \cdot e^{\imath 2 \pi \nu (k-1) (j-1) / n})$, where $\tilde{c}_{\nu \cdot k} \in \mathbb{C}$ and $n / \nu \in \mathbb{N}$. If $\varpi \geq \omega$ and $\nu \geq \varpi + \omega - 1$, then  a sequence $\mathbf{m}^{(0)}, \mathbf{m}^{(1)}, \ldots, \mathbf{m}^{(i)},\ldots$ formed by the AP-A algorithm for $\epsilon_{tol}=0$ and $N_{iter} \to +\infty$ converges to $\mathbf{m}$.
\label{prop:APASolConv_b}
\end{proposition_III_4}

\begin{proof}
The proof of this proposition goes along the lines of that of \emph{\hyperref[prop:APBSolConv_b]{Proposition~III.2}}. First, by using \eqref{eq:PropIII.1b_1} and \eqref{eq:PropIII.1b_2}, we derive the result analogous to \eqref{eq:PropIII.1b_12} and \eqref{eq:PropIII.1b_13}. Then, we show the monotonic convergence of $q^{(0)}, q^{(1)}, \ldots, q^{(i)}, \ldots$ to 1, which assures the convergence of $\mathbf{a}^{(0)}, \mathbf{a}^{(1)}, \ldots, \mathbf{a}^{(i)}, \ldots$ and $\mathbf{m}^{(0)}, \mathbf{m}^{(1)}, \ldots, \mathbf{m}^{(i)}, \ldots$ to $\mathbf{m}$.

To derive the result analogous to \eqref{eq:PropIII.1b_12} and \eqref{eq:PropIII.1b_13}, assume that, for some $i \geq 0$ and $q^{(i)} \in \mathbb{R}$,
\begin{equation}
\begin{aligned}
\mathbf{a}^{(i)} &= q^{(i)} \cdot \mathbf{m},\\
\mathbf{m}^{(i)} &= \mathbf{m} \circ \big( q^{(i)} + (|\mathbf{c}|-q^{(i)}) \circ \theta (|\mathbf{c}|-q^{(i)}) \big).
\end{aligned}
\label{eq:PropIII.2b_1}
\end{equation}
First, note that \eqref{eq:PropIII.2b_1} applies when $i=0$ with $q^{(0)}=0$. Next, by the definition of the AP-A algorithm, we have
\begin{equation}
\begin{aligned}
\mathbf{b}^{(i+1)} &= \mathbf{P}_{\Sw}[\mathbf{m}^{(i)}-\mathbf{a}^{(i)}]\\
&= \mathbf{P}_{\Sw}[\mathbf{m}^{(i)}]-\mathbf{a}^{(i)}\\
&= \mathbf{P}_{\Sw}[\mathbf{m} \circ \big( q^{(i)} + (|\mathbf{c}|-q^{(i)}) \circ \theta (|\mathbf{c}|-q^{(i)}) \big)] - q^{(i)} \cdot \mathbf{m}\\
&= \underbrace{\langle (|\mathbf{c}|-q^{(i)}) \circ \theta (|\mathbf{c}|-q^{(i)}) \rangle \cdot \mathbf{m}}_{\textrm{by \eqref{eq:PropIII.1b_1} and \eqref{eq:PropIII.1b_2}}},
\end{aligned}
\label{eq:PropIII.2b_2}
\end{equation}
\begin{equation}
\hspace*{-69pt}\begin{aligned}
\lambda &= \frac{\|\mathbf{m}^{(i)}-\mathbf{a}^{(i)}\|_2}{\|\mathbf{b}^{(i+1)}\|_2} \\
&= \frac{\|(|\mathbf{c}|-q^{(i)}) \circ \theta (|\mathbf{c}|-q^{(i)}) \circ \mathbf{m}\|_2}{\langle (|\mathbf{c}|-q^{(i)}) \circ \theta (|\mathbf{c}|-q^{(i)}) \rangle \cdot \|\mathbf{m}\|_2},
\end{aligned}
\label{eq:PropIII.2b_3}
\end{equation}
\begin{equation}
\hspace*{-33pt}\begin{aligned}
\mathbf{a}^{(i+1)} &= \mathbf{a}^{(i)} + \lambda \cdot \mathbf{b}^{(i+1)} \\
&= \underbrace{\bigg(q^{(i)} + \frac{\| (|\mathbf{c}|-q^{(i)}) \circ \theta (|\mathbf{c}|-q^{(i)}) \circ \mathbf{m}\|_2}{\| \mathbf{m} \|_2} \bigg)}_{q^{(i+1)}} \cdot \mathbf{m},
\end{aligned}
\label{eq:PropIII.2b_4}
\end{equation}
and
\begin{equation}
\hspace*{-42pt}\begin{aligned}
\mathbf{m}^{(i+1)} &= \mathbf{P}_{\Sgeqs}[\mathbf{a}^{(i+1)}],\\
&= \mathbf{m} \circ \big( q^{(i+1)} + (|\mathbf{c}|-q^{(i+1)}) \circ \theta (|\mathbf{c}|-q^{(i+1)}) \big).
\end{aligned}
\label{eq:PropIII.2b_5}
\end{equation}
Applying the principle of mathematical induction to the above results, we conclude that, for any $i \geq 0$,
\begin{align}
\mathbf{a}^{(i)} &= q^{(i)} \cdot \mathbf{m},\label{eq:PropIII.2b_6}\\
q^{(i)} &= \bigg(q^{(i-1)} + \frac{\| (|\mathbf{c}|-q^{(i-1)}) \circ \theta (|\mathbf{c}|-q^{(i-1)}) \circ \mathbf{m}\|_2}{\| \mathbf{m} \|_2} \bigg),
\label{eq:PropIII.2b_7}
\end{align}
with $q^{(0)}=0$.

By \eqref{eq:PropIII.2b_7},
\begin{equation}
\begin{aligned}
q^{(i)} - q^{(i-1)} &= \frac{\| (|\mathbf{c}|-q^{(i-1)}) \circ \theta (|\mathbf{c}|-q^{(i-1)}) \circ \mathbf{m}\|_2}{\| \mathbf{m} \|_2} \\
&\geq 0,
\end{aligned}
\label{eq:PropIII.2b_8}
\end{equation}
i.e., the sequence $q^{(0)}, q^{(1)}, \ldots$ is monotonically increasing. Moreover, it follows from \eqref{eq:PropIII.2b_7} that, for every $i\geq 1$,
\begin{equation}
\begin{aligned}
q^{(i-1)} \leq 1 \implies q^{(i)} &= \bigg(q^{(i-1)} + \frac{\| (|\mathbf{c}|-q^{(i-1)}) \circ \theta (|\mathbf{c}|-q^{(i-1)}) \circ \mathbf{m}\|_2}{\| \mathbf{m} \|_2} \bigg) \\
&\leq \bigg(q^{(i-1)} + \frac{\| (1-q^{(i-1)}) \cdot \mathbf{m}\|_2}{\| \mathbf{m} \|_2} \bigg) = 1,
\end{aligned}
\label{eq:PropIII.2b_9}
\end{equation}
Taken together with $q^{(0)}=0$, \eqref{eq:PropIII.2b_9} implies that the sequence $q^{(0)}, q^{(1)}, \ldots, q^{(i)}, \ldots$ is bounded from above by 1. Hence, by the monotone convergence theorem (see \PropRef{prop:MonConvTh}), $q^{(0)}, q^{(1)}, \ldots, q^{(i)}, \ldots$ converges to its least upper bound $\bar{q} \leq 1$. If we assume that $\bar{q} < 1$, then the convergence of $q^{(0)}, q^{(1)}, \ldots, q^{(i)}, \ldots$ to $\bar{q}$ implies that there exists an $N$ such that
\begin{equation}
\bar{q} - q^{(N)} \leq  \frac{1}{2} \cdot \frac{\|(|\mathbf{c}|-\bar{q}) \circ \theta (|\mathbf{c}|-\bar{q}) \circ \mathbf{m}\|_2}{\|\mathbf{m}\|_2}.
\label{eq:PropIII.2b_10}
\end{equation}
By \eqref{eq:PropIII.2b_7},
\begin{equation}
\begin{aligned}
q^{(N+1)} - q^{(N)} &= \frac{\|(|\mathbf{c}|-q^{(N)}) \circ \theta (|\mathbf{c}|-q^{(N)}) \circ \mathbf{m}\|_2}{\|\mathbf{m}\|_2} \\
&\geq \frac{\|(|\mathbf{c}|-\bar{q}) \circ \theta (|\mathbf{c}|-\bar{q}) \circ \mathbf{m}\|_2}{\|\mathbf{m}\|_2} \\
&> \frac{1}{2} \cdot \frac{\|(|\mathbf{c}|-\bar{q}) \circ \theta (|\mathbf{c}|-\bar{q}) \circ \mathbf{m}\|_2}{\|\mathbf{m}\|_2},
\label{eq:PropIII.2b_11}
\end{aligned}
\end{equation}
which means that $q^{(N+1)} > \bar{q}$, i.e., the initial assumption that $\bar{q}<1$ is incorrect. Therefore, $\bar{q}=1$, and $q^{(0)}, q^{(1)}, \ldots, q^{(i)}, \ldots$ converges to 1. This, as shown in the last paragraph of the proof of \emph{\hyperref[prop:APBSolConv_b]{Proposition~III.2}}, implies that $\mathbf{a}^{(0)}, \mathbf{a}^{(1)}, \ldots, \mathbf{a}^{(i)}, \ldots$ and $\mathbf{m}^{(0)}, \mathbf{m}^{(1)}, \ldots, \mathbf{m}^{(i)}, \ldots$ converge to $\mathbf{m}$.
\end{proof}

\subsection*{\textbf{AP-P algorithm}}

\vspace{0.15cm}
\begin{proposition_III_5}
A sequence $\mathbf{m}^{(0)},\mathbf{m}^{(1)}, \ldots, \mathbf{m}^{(i)},\ldots$ formed by the AP-P algorithm for $\epsilon_{tol}=0$ and $N_{iter} \to +\infty$ converges to a unique $\mathbf{m}^\dagger \in \Sgeqs \cap \Sw$ such that $\|\mathbf{m}^\dagger\|_2 \leq \|\mathbf{x}\|_2$ for every $\mathbf{x} \in \Sgeqs \cap \Sw$. The convergence is monotonic, i.e., $\| \mathbf{m}^{(i+1)} - \mathbf{m}^\dagger \|_2 \leq \| \mathbf{m}^{(i)} - \mathbf{m}^\dagger \|_2$ for $i \geq 0$.
\end{proposition_III_5}

\begin{proof}
The proof follows as a corollary of a more general theorem proved for a finite number of closed convex sets in a Hilbert space by Boyle and Dykstra \cite[\textit{Theorem 2}]{Boyle1986}. Specifically, for $\mathbf{m}^{(0)} = \mathbf{0}$ and particular constraint sets $\Sgeqs$ and $\Sw$, the sequence $\mathbf{m}^{(0)},\mathbf{m}^{(1)}, \ldots, \mathbf{m}^{(i)}, \ldots$ formed by the algorithm formulated there (the so-called Dykstra's algorithm) converges to a unique $\mathbf{m}^\dagger \in \Sgeqs \cap \Sw$ such that $\|\mathbf{m}^\dagger\|_2 \leq \|\mathbf{x}\|_2$ for any $\mathbf{x} \in \Sgeqs \cap \Sw$. For our purposes, it is thus enough to show that the sequence $\mathbf{m}^{(0)},\mathbf{m}^{(1)}, \ldots, \mathbf{m}^{(i)}, \ldots$ generated by the AP-P algorithm converges to the same $\mathbf{m}^\dagger$.

First, we observe that Dykstra's algorithm formally turns into the AP-P (except the difference in the initial conditions) if we set $\mathrm{r}=2$, denote $\mathrm{g_{i,1}} \equiv \mathbf{a}^{(i-1)}$, $\mathrm{g_{i,2}} \equiv \mathbf{m}^{(i-1)}$, $\mathrm{I_{i,2}} \equiv \mathbf{c}^{(i-1)}$, and assign $\mathrm{I_{i,1}} \equiv 0$ there (see \textit{Theorem 2} in \cite{Boyle1986}). Note that, originally, $\mathrm{I_{i,1}} = \mathrm{g_{i,1}}-(\mathrm{g_{i,2}}-\mathrm{I_{i-1,1}})$ and $\mathrm{g_{i,1}} = \mathbf{P}_{\Sw}[\mathrm{g_{i-1,2}}-\mathrm{I_{i-1,1}}]$. However, due to the linearity of $\mathbf{P}_{\Sw}$ [see \eqref{eq:MathPrel3}], we have
\begin{equation}
\begin{aligned}
\mathrm{g_{i,1}} &= \mathbf{P}_{\Sw}[\mathrm{g_{i-1,2}}] + \mathbf{P}_{\Sw}[\mathrm{I_{i-1,1}}] \\
&= \mathbf{P}_{\Sw}[\mathrm{g_{i-1,2}}] + \underbrace{\mathbf{P}_{\Sw}[\mathrm{g_{i-1,1}}]}_{=\mathrm{g_{i-1,1}}}-\underbrace{\mathbf{P}_{\Sw}[\mathrm{g_{i-1,2}}]}_{=\mathrm{g_{i-1,1}}}+\mathbf{P}_{\Sw}[\mathrm{I_{i-2,1}}] \\
&= \mathbf{P}_{\Sw}[\mathrm{g_{i-1,2}}] +\mathbf{P}_{\Sw}[\mathrm{I_{i-2,1}}].
\end{aligned}
\label{eq:PropIII.3_1}
\end{equation}
Applying \eqref{eq:PropIII.3_1} iteratively, we obtain $\mathrm{g_{i,1}}=\mathbf{P}_{\Sw}[\mathrm{g_{i-1,2}}] +\mathbf{P}_{\Sw}[\mathrm{I_{0,1}}] = \mathbf{P}_{\Sw}[\mathrm{g_{i-1,2}}]$ because $\mathrm{I_{0,1}}=0$ by construction. Therefore, given the constraint sets $\Sgeqs$ and $\Sw$, $\mathrm{I_{i,1}}$ can indeed be canceled from the Dykstra's algorithm by setting it to 0 without any consequences to its convergence properties.

Further, assuming $\mathbf{m}^{(0)}=\mathbf{c}^{(0)}=\mathbf{0}$, we can verify that Dykstra's algorithm after the first iteration coincides with the initialized AP-P algorithm, i.e., the latter is shifted forward by one iteration with respect to the former. Hence, the sequence $\mathbf{m}^{(0)},\mathbf{m}^{(1)}, \ldots, \mathbf{m}^{(i)}, \ldots$ generated by the AP-P converges to the same $\mathbf{m}^\dagger \in \Sgeqs \cap \Sw$ as the analogous sequence formed by the Dykstra's algorithm initialized with $\mathbf{m}^{(0)} = \mathbf{0}$.

The monotonicity of the convergence of Dykstra's, and thus the AP-P, algorithms follows from the fact that all terms on the right-hand side of \cite[(3.2)]{Boyle1986} are nonnegative and that each subsequent iteration only adds additional terms without discarding the old ones.
\end{proof}

\section{\textbf{Lower Bound on the Convergence Error} \label{sec:SMTheory3}}

The proof of the statement about the lower bound on the convergence error made in Sections~III-A and III-B of the main text are presented here.

\begin{proposition}
In the case of the AP-B and AP-A algorithms, the convergence error $\|\mathbf{m}^{(i-1)}-\mathbf{m}^\dagger\|_2 / \sqrt{n}$ is bounded from below by the infeasibility error $\epsilon^{(i)}$ for any $i \geq 1$.
\end{proposition}

\begin{proof}
According to \eqref{eq:PropIII.1_3}, which applies to both the AP-B and AP-A algorithms,
\begin{equation}
\|\mathbf{m}^{(i-1)}-\mathbf{m}^\dagger\|_2 / \sqrt{n} \geq \|\mathbf{a}^{(i)}-\mathbf{m}^\dagger\|_2 / \sqrt{n} \qquad \forall i \geq 1.
\end{equation}
Moreover, by the \DefRef{def:ProjOp} of the projection operator and the fact that $\mathbf{m}^\dagger \in \Sgeqs$, we have
\begin{equation}
\|\mathbf{a}^{(i)}-\mathbf{m}^\dagger\|_2 / \sqrt{n} \geq \|\mathbf{a}^{(i)}-\mathbf{P}_{\Sgeqss}[\mathbf{a}^{(i)}]\|_2 / \sqrt{n} = \|\mathbf{a}^{(i)}-\mathbf{m}^{(i)}\|_2 / \sqrt{n} = \epsilon^{(i)} \qquad \forall i \geq 1.
\end{equation}
Thus, $\|\mathbf{m}^{(i-1)}-\mathbf{m}^\dagger\|_2 / \sqrt{n} \geq \epsilon^{(i)}$ for all $i \geq 1$.
\end{proof}

\begin{remark}
Using the \DefRef{def:ProjOp} of the projection operator and the fact that $\mathbf{m}^\dagger \in \Sw$, we similarly conclude that
\begin{equation}
\|\mathbf{m}^{(i)}-\mathbf{m}^\dagger\|_2 / \sqrt{n} \geq \|\mathbf{m}^{(i)}-\mathbf{P}_{\Sw}[\mathbf{m}^{(i)}]\|_2 / \sqrt{n} \qquad \forall i \geq 0.
\label{eq:PropLB.1}
\end{equation}
\eqref{eq:PropLB.1} also applies in the context of the AP-P algorithm.
\end{remark}

\newpage
\phantomsection

\markboth{SIGNAL CLASSES}%
{SIGNAL CLASSES}

\hspace{-11pt}\textbf{SIGNAL CLASSES}
\addcontentsline{toc}{section}{SIGNAL CLASSES}

\section{\textbf{Types of Wideband Carriers Found in Practice} \label{sec:SMCarriers}}

Wideband carriers found in practice fall into three main classes: \textit{1)} \textbf{(quasi-)harmonic}; \textit{2)} \textbf{\mbox{(quasi-)}random}; and \textit{3)} \textbf{spike-train}. Below, we provide examples of real-world signals featuring these carrier types and applications of their demodulation.\vspace{0.0em}

The need for demodulating signals formed of harmonic carriers is well recognized in acoustic imaging \cite{Humphrey2000, Duck2002}. There, sinusoidal wavepackets are used as probing signals. However, in many situations, the interaction between sound and matter is nonlinear. This makes the returning waves, whose time-dependent amplitude carries information about the imaged object, harmonic. The possibility of efficient and accurate demodulation of signals of this type would also allow generalizing the probing wave packets themselves from sinusoidal to harmonic. 

A representative example of quasi-random carriers manifests in surface electromyography. In particular, an electrical signal detected at the skin surface during the skeletal muscle activity is an amplitude-modulated colored noise resulting from low-pass filtering of electric pulse trains generated by a large set of conditionally independent muscle fibers \cite[Ch.\,5]{Sornmo2005}. The amplitude component of the recorded electrical signal carries information about the force pattern generated by the muscle being studied \cite{Gottlieb1970, Roberts2008} (see Fig.\,\ref{fig:X1} for an example).

\begin{figure*}[h]
\centering
\includegraphics[width=1\textwidth]{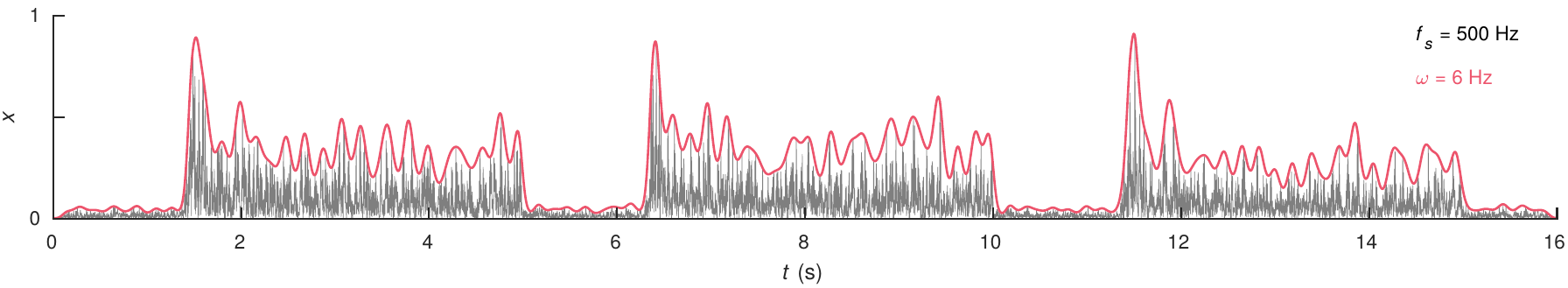}
\caption{\footnotesize Demodulation of an electromyogram. The figure shows a 16\,s long absolute-value electromyogram signal corresponding to the response of a flexor carpi ulnaris muscle (forearm) to three consecutive grasps of a spherical object (grey) and its modulator inferred by using the AP-A algorithm (red). The original recording was taken from \cite{UCIRepo,Sapsanis2013}.}
\label{fig:X1}
\end{figure*}

Probably the most elaborate applications of amplitude demodulation in the context of wideband signals are found in human speech processing. Speech signals are built of temporarily structured segments of quasi-harmonic and quasi-random carriers \cite{Shoup1976} that are amplitude-modulated at different timescales \cite{Turner2010, Keitel2018}. Amplitude demodulation of broad-band, as well as sub-band speech demodulation, is widely exploited in automatic speech recognition \cite{Kingsbury1998, Wu2011, Lee2016}, hearing restoration \cite{Wilson1991, Zeng2005} tasks, and fundamental studies of the neural mechanisms of auditory information processing in the brain \cite{Shannon1995, Smith2002, Joris2004, Goswami2019}. In all these cases, the modulators and carriers convey the information about specific aspects of speech, e.g., semantic meaning, associated emotion, or speaker identity, that need to be extracted. Demodulation of signals with mixed quasi-harmonic and quasi-random carrier types is also known in diagnostic phonocardiography \cite{Sarkady1976, Gill2005}. There, the extracted modulators of heart sounds are used for the detection of events of normal or abnormal functioning of this organ.

One of the most popular applications using demodulation of signals built of the spike-train type carriers is the Pulse-Code Modulation (PCM) protocol \cite{Oliver1948}. There, pulses are regularly spaced with specified locations and have a known constant amplitude and vanishing width.\footnote{These properties enable the pulse-coded signals to be demodulated by a simple low-pass filtering procedure \cite{Oliver1948}.} More complex quasi-regular or stochastic pulse sequences of finite width manifest in electric activities of the cardiac muscles and neurons \cite[Ch.\,6]{Sornmo2005}, \cite[Ch.\,1]{Rieke1999}. Physiological and diagnostic information contained in these responses is typically associated with the microstructure and timing of the pulses. However, these recordings often come contaminated by artifacts or other physiological signals that slowly modulate the baseline and amplitude of the pulses \cite{Felblinger1997}, \cite[Ch.\,7.1]{Sornmo2005}. Hence, to separate the fast cardiac or neural activity (the carrier) from other slowly changing physiological processes or artifacts (the modulator), the raw signal must be demodulated \cite{Felblinger1997} (see Fig.\,\ref{fig:X2} for an example). Finally, as we discuss in Sections~IX.A and IX.B of the main text, carriers of the spike-train type also effectively manifest in applications when modulator recovery from nonuniformly or sparsely sampled signals of any origin is needed.

\begin{figure*}[h]
\centering
\includegraphics[width=1\textwidth]{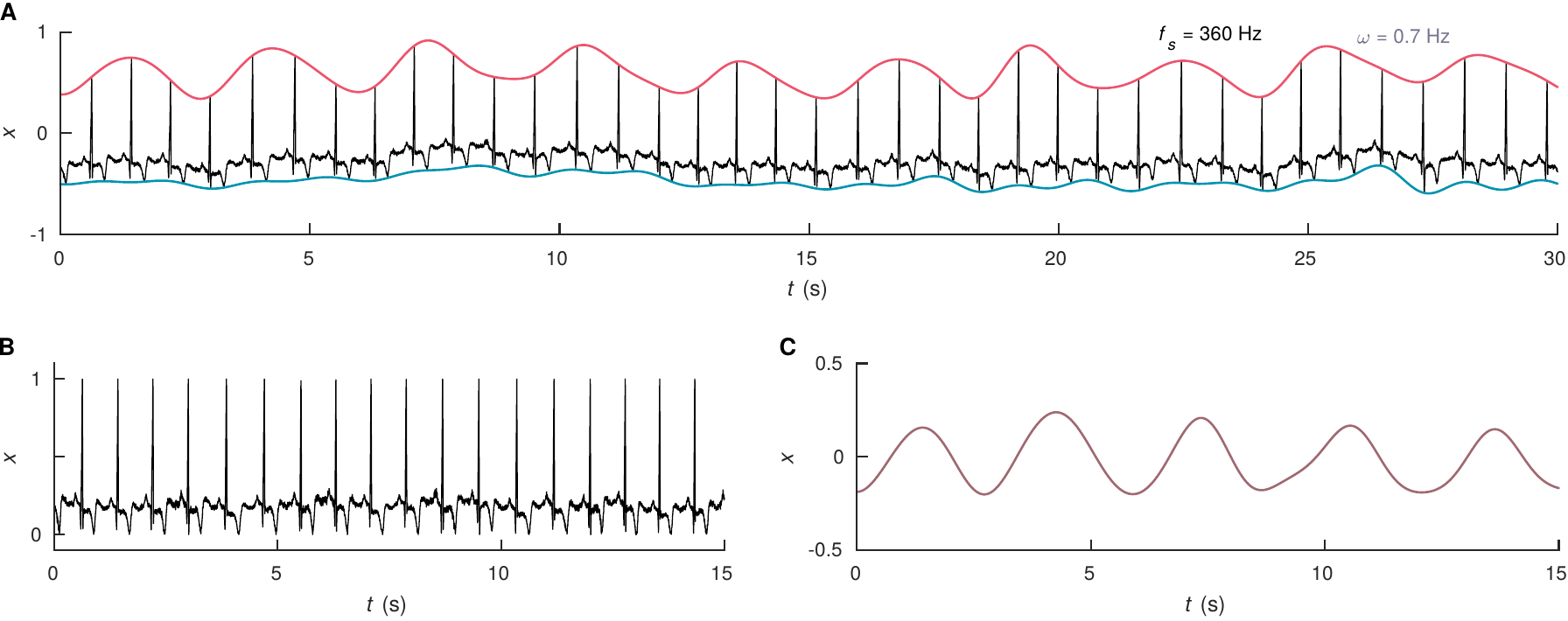}
\caption{\footnotesize Demodulation of an electrocardiogram. \textbf{A}: A 30\,s long fragment of an electrocardiogram recording (black) and its upper (red) and lower (blue) envelopes inferred by using the AP-A algorithm.\protect\footnotemark~\textbf{B}: The first 15\,s of the recovered carrier, i.e., the desired electrocardiogram with canceled artifacts. \textbf{C}: The first 15\,s of the respiratory curve estimated from the upper and lower envelopes of the original signal in panel A. The original recording was taken from \cite{Gargiulo2016,Goldberger2000}.}
\label{fig:X2}
\end{figure*}

\footnotetext{Here, we define the upper envelope of a signal $\mathbf{s}$ as $\mathbf{m}+\min[\mathbf{s}] \cdot \mathbf{1}$, where $\mathbf{m}$ is the modulator of $\mathbf{s}-\min[\mathbf{s}] \cdot \mathbf{1}$. Accordingly, the lower envelope of $\mathbf{s}$ is defined as the negative upper envelope of $-\mathbf{s}$.}

\section{\textbf{Synthetic Test Signals} \label{sec:SMSynthetic}}

The mathematical models and numerical sampling algorithms of synthetic modulators and carriers examined in this work are described next.

\subsection*{Modulators (recovery tests)}

For signal recovery tests discussed in Section~\ref{sec:SMRecovery3}, modulators were generated by uniformly sampling from $\Mwm$. Without loss of generality, we assumed the subset of $\Mwm$ whose elements' norms are fixed to $\sqrt{n}$. The sampling was achieved by using a specially adapted version of the Metropolis-Hastings algorithm (see \cite[p.\,411]{Wasserman2004} for an introduction to this method). In the concrete, starting with some $\mathbf{m}^{(0)} \in \Mwm$, a sequence $\mathbf{m}^{(0)}, \mathbf{m}^{(1)}, \ldots$ defined by
\begin{align}
\mathbf{m}^{(i)} =\begin{cases} \sqrt{n} \cdot \mathbf{r}^{(i)} / \|\mathbf{r}^{(i)}\|_2, & \mbox{if $r_j^{(i)} \geq 0~~\forall j \in \mathcal{I}_n$}\\ \mathbf{m}^{(i-1)}, & \mbox{otherwise} \end{cases}
\label{eq:TestSign1}
\end{align}
for $i \geq 1$ was generated. Here,
\begin{align}
\mathbf{r}^{(i)} = \mathbf{m}^{(i-1)} + \mathbf{F}^{-1} \mathbf{g}^{(i)},
\label{eq:TestSign2}
\end{align}
and $\mathbf{g}^{(i)} \in \mathbb{C}^n$ such that 
\begin{equation}
\begin{aligned}
&\myre\big(g_1^{(i)}\big) \sim \mathcal{N}(0,\sigma), \quad \myim\big(g_1^{(i)}\big)=0,\\
&\myre\big(g_j^{(i)}\big) \sim \mathcal{N}(0,\sigma), \quad \myim\big(g_j^{(i)}\big) \sim \mathcal{N}(0,\sigma), \quad 2 \leq j \leq \omega,\\
&g_j^{(i)} = \big(g_{n+2-j}^{(i)}\big)^*, \quad n+2-\omega \leq j \leq n,\\
&g_j^{(i)}=0, \quad \omega+1 \leq j \leq n+1-\omega,
\end{aligned}
\label{eq:TestSign3}
\end{equation}
and $g_j \indep g_{k}$ for $j \neq k$. We adjusted the parameter $\sigma$ in \eqref{eq:TestSign3} to achieve the acceptance rate of the sampling scheme \eqref{eq:TestSign1} between 0.4 and 0.6. The initial $\mathbf{m}^{(0)}$ was taken as
\begin{align}
\mathbf{m}^{(0)} &= \sqrt{n} \cdot (\mathbf{g}^{(0)}-\min[\mathbf{g}^{(0)}] \cdot \mathbf{1}) / \|\mathbf{g}^{(0)}-\min[\mathbf{g}^{(0)}] \cdot \mathbf{1}\|_2. \label{eq:TestSign4}
\end{align}
The iterative process generating the sequence $\mathbf{m}^{(0)}, \mathbf{m}^{(1)}, \ldots$ was terminated upon the first instance of adherence to the following equilibration criterion of the underlying Markov chain:
\begin{align}
\frac{1}{2\omega-2} \cdot \sum_{j=2}^{\omega} \Bigg|\frac{1}{i} \cdot \sum_{l=1}^i (\mathbf{F}\mathbf{m}^{(l)})_j \Bigg|^2 < 0.01 \cdot \frac{n}{2\omega-1}. \label{eq:TestSign5}
\end{align}
The corresponding $\mathbf{m}^{(i)}$ was then chosen as a sample point from $\Mwm$.

\subsection*{Modulators (performance tests)}

For the performance, convergence, and robustness tests of the demodulation algorithm discussed in Sections~IV\,--\,VI of the main text, two types of modulators were considered:

\begin{itemize}

\item \textit{Nonstationary Gaussian} modulators were produced by transforming a delta-correlated Gaussian process with a time-dependent low-pass filter. Specifically, $\mathbf{m}$ was calculated as
\begin{align}
\mathbf{m} &= \boldsymbol{w} \circ (\mathbf{m}'-\min[\mathbf{m}'] \cdot \mathbf{1}) / \max[\mathbf{m}'-\min[\mathbf{m}'] \cdot \mathbf{1}],
\label{eq:TestSign60}
\end{align}
with $\boldsymbol{w}$ being a window ``function'' (see \eqref{eq:TestSign13}), and $\mathbf{m}'$ defined by
\begin{align}
m_i' = \sum_{j=1}^n \Big(g_j^{(0)} \cdot \big(\mathbf{F}^{-1} \mathbf{h}^{(i)}\big)_{i-j}\Big), \quad i \in \mathcal{I}_n.  \label{eq:TestSign6}
\end{align}
In \eqref{eq:TestSign6}, for every $i \in \mathcal{I}_n$,
\begin{equation}
\begin{aligned}
&\myre\big(h_1^{(i)}\big) \sim \big(\mathbf{P}_{\Swm}[\mathbf{g}^{(1)}]\big)_i, \quad \myim\big(h_1^{(i)}\big)=0,\\
&\myre\big(h_j^{(i)}\big) = \big(\mathbf{P}_{\Swm}[\mathbf{g}^{(2j-2)}]\big)_i, \quad \myim\big(h_j^{(i)}\big)=\big(\mathbf{P}_{\Swm}[\mathbf{g}^{(2j-1)}]\big)_i, \quad 2 \leq j \leq \omega,\\
&h_j^{(i)} = \big(h_{n+2-j}^{(i)}\big)^*, \quad n+2-\omega \leq j \leq n,\\
&h_j^{(i)}=0, \quad \omega+1 \leq j \leq n+1-\omega,
\end{aligned}
\label{eq:TestSign7}
\end{equation}
with $\mathbf{g}^{(j)}, ~j \in \mathcal{I}_{2\omega-1}$, being independent samples from the standard $n$-dimensional Gaussian distribution, i.e., $g^{(j)}_l \sim \mathcal{N}(0,1)$ for $l \in \mathcal{I}_n$, and $g^{(j)}_l \indep g^{(j)}_k$ for $l \neq k$.

\item \textit{Maximally-uniformly distributed} modulators were created by using the NORTA algorithm \cite{Cario1997}. Specifically, the modulators were calculated as
\begin{align}
\mathbf{m} &= \boldsymbol{w} \circ (\mathbf{m}'-\min[\mathbf{m}'] \cdot \mathbf{1}) / \max[\mathbf{m}'-\min[\mathbf{m}'] \cdot \mathbf{1}],
\label{eq:TestSign80}
\end{align}
with $\boldsymbol{w}$ being the window ``function'' (see \eqref{eq:TestSign13}), and $\mathbf{m}'$ given by
\begin{align}
\mathbf{m}'&=\mathbf{P}_{\Swm}[\Phi(\mathbf{F}^{-1}\mathbf{r})],
\label{eq:TestSign8}
\end{align}
where $\Phi(\ldots)$ is the cumulative distribution function of the standard Gaussian random variable and $\mathbf{r}$  is given by
\begin{align}
r_j = \begin{cases} \sqrt{p^{\star}_1} \cdot g_1, & \mbox{$j=1$}\\ \sqrt{p^{\star}_j / 2} \cdot (g_{2j-2}+\imath \cdot g_{2j-1}), & \mbox{$2 \leq j \leq \lfloor(n+1)/2\rfloor$}\\ \sqrt{p^{\star}_{(n+2)/2} } \cdot g_n, & \mbox{$j=(n+2)/2$}\\ (r_{n+2-j})^*, & \mbox{$(n+2)/2 < j \leq n$} \end{cases}.
\label{eq:TestSign9}
\end{align}
In \eqref{eq:TestSign9}, $\mathbf{g}$ is the standard $n$-dimensional Gaussian random vector, and
\begin{align}
\mathbf{p}^\star = \sqrt{n} \cdot (\mathbf{F}\mathbf{c}^\star) \circ \theta(\mathbf{F}\mathbf{c}^\star),
\label{eq:TestSign10}
\end{align}
with $c^\star_1=1$, and the remaining components of $\mathbf{c}^\star$ implicitly defined by
\begin{align}
(\mathbf{F}^{-1} \mathbf{p} / \sqrt{n})_j = \int_{-\infty}^{+\infty}\int_{-\infty}^{+\infty} \big( \big(\Phi(x)-0.5\big) \cdot \big(\Phi(y)-0.5\big) \cdot \phi_2(x,y|c^\star_j)\big)dxdy, \quad 2 \leq j \leq n.
\label{eq:TestSign11}
\end{align}
In \eqref{eq:TestSign11}, $\phi_2(x,y|c^\star_j)$ is the probability density function of a 2-dimensional Gaussian random variable with zero mean, unit variance, and the correlation between its two elements equal to $c^\star_j$. Elements of $\mathbf{p}$ are given by 
\begin{align}
p_j = \begin{cases} n/(24\cdot(\omega-1)), & \mbox{$2 \leq j \leq \omega$}\\ n/(24\cdot(\omega-1)), & \mbox{$n+2-\omega \leq j \leq n$}\\ 0, & \mbox{otherwise} \end{cases},
\label{eq:TestSign12}
\end{align}
If all elements of $\mathbf{F}\mathbf{c}^\star$ were nonnegative, the $\mathbf{m}'$ defined by \eqref{eq:TestSign8} would have a rectangular power spectrum with cutoff frequency $\omega$, and $m_i' \sim \mathcal{U}(0,1)$ for every $i \in \mathcal{I}_n$. However, in reality, such random vector does not exist. Therefore, \eqref{eq:TestSign10} is used to replace $\mathbf{F}\mathbf{c}^\star$ by the closest point in $\mathbb{C}^n$ that is elementwise nonnegative. This modification expands the power spectrum of the resulting $\Phi(\mathbf{F}^{-1}\mathbf{r})$ in \eqref{eq:TestSign8} beyond the intended one, which is corrected by mapping $\Phi(\mathbf{F}^{-1}\mathbf{r})$ onto $\Mwm$ in \eqref{eq:TestSign8}.

\end{itemize}

The window ``function'' $\boldsymbol{w}$ in \eqref{eq:TestSign60} and \eqref{eq:TestSign80} is used to scale the modulator to zero smoothly at the boundaries with no effect on the remaining points:
\begin{align}
w_i =\begin{cases} \sin^2\Big(\frac{\pi \cdot (i-1)}{2 \cdot n_{trn}}\Big), & \mbox{$\quad 1 \leq i \leq n_{trn}$}\\ 1, & \mbox{$\quad n_{trn}<i\leq n-n_{trn}$}\\
\cos^2\Big(\frac{\pi \cdot (i-n+n_{trn})}{2 \cdot n_{trn}}\Big), & \mbox{$\quad n-n_{trn} < i \leq n$} \end{cases} \qquad \forall i \in \mathcal{I}.
\label{eq:TestSign13}
\end{align}
In \eqref{eq:TestSign13}, $n_{trn}$ is the length of the transition window.  We assumed $n_{trn} = 3 \cdot \lceil f_s / \omega \rceil$ in the present work.

For simulations discussed in Sections~IV\,--\,VI of the main text, we used $f_s = 4\,\mathrm{kHz}$. $\omega$ was set to $15\,\mathrm{Hz}$ for \textit{nonstationary Gaussian} modulators and $20\,\mathrm{Hz}$ for \textit{maximally-uniformly distributed} modulators. The cutoff frequency control parameter $\varpi$ of the AP algorithms was fixed to $30\,\mathrm{Hz}$.

\subsection*{Carriers (recovery tests)}

For signal recovery tests discussed in Section~\ref{sec:SMRecovery3}, carriers were generated by uniformly sampling from a subset of $\Cd$ formed by all spike-train carriers with a fixed number of spikes $n_s$:
\begin{align}
\textstyle
\Cd^{n_s} = \big\{\mathbf{x} \in \Cd: \big( \sum_{i=1}^n I_{\{1\}}(x_i)=n_s \big) \, \land \, \big( \sum_{i=1}^n I_{\{0\}}(x_i)=n-n_s \big) \big\}
\label{eq:TestCar1}
\end{align}
The sampling was achieved by exploiting a customized version of the Metropolis-Hastings algorithm. In particular, starting with some $\mathbf{q}^{(0)} \in \mathcal{Q}$ such that
\begin{align}
\textstyle
\mathcal{Q} = \big\{\mathbf{x} \in \mathbb{N}^{n_*}_0: \big( \sum_{i=1}^{n_*}x_i = n_- \big) \land \big( x_i < d ~\forall i \in \mathcal{I}_{n_*} \big) \big\},
\label{eq:TestCar2}
\end{align}
where $n_-= n_s \cdot d - n$, and $n_*=\min\{n_-,n_s-1\}$, a sequence $\mathbf{q}^{(0)}, \mathbf{q}^{(1)}, \ldots$ defined by
\begin{align}
\mathbf{q}^{(i)} =\begin{cases} \mathbf{q}^{(i-1)}+z_i\cdot(\mathbf{e}^{(l_i)}-\mathbf{e}^{(k_i)}), & \mbox{if $\big(\mathbf{q}^{(i-1)}+z_i\cdot(\mathbf{e}^{(l_i)}-\mathbf{e}^{(k_i)})\big) \in \mathcal{Q}$}\\ \mathbf{q}^{(i-1)}, & \mbox{otherwise} \end{cases}
\label{eq:TestCar3}
\end{align}
for $i \geq 1$ was generated. In \eqref{eq:TestCar3}, $\mathbf{e}^{(l)}$ is an element of $\mathbb{N}_0^{n_*}$ with its $l$-th component equal to one and the remaining components equal to zero; $l_i$ and $k_i$ are independent random numbers from a uniform distribution on $\mathcal{I}_{n_*}$; $z_i$ is a random number from a uniform distribution on the set of all integer numbers between $-\sigma$ and $+\sigma$, where $\sigma$ is a positive integer chosen to achieve the acceptance rate of the sampling scheme \eqref{eq:TestCar3} between 0.4 and 0.6.

The iterative process generating the sequence $\mathbf{q}^{(0)}, \mathbf{q}^{(1)}, \ldots$ was terminated upon the first instance of adherence to the following equilibration criterion of the underlying Markov chain:
\begin{align}
\frac{1}{n_*} \cdot \sum_{j=1}^{n_*} \Bigg|\frac{1}{i} \cdot \sum_{l=1}^i q_j^{(i)} -\frac{n_-}{n_*} \Bigg| < 0.01 \cdot \frac{n_-}{n_*}. \label{eq:TestCar4}
\end{align}
The corresponding $\mathbf{q}^{(i)}$ was then taken as a sample point from $\mathcal{Q}$. Next, an extended vector $\mathbf{\bar{q}} \in \mathbb{N}_0^{n_s}$ with
\begin{align}
\bar{q}_j =\begin{cases} d-q_j, & \mbox{if $1 \leq j \leq n_*$}\\ d, & \mbox{$n_*+1 \leq j \leq n_s$} \end{cases}
\label{eq:TestCar5}
\end{align}
was defined. Its components were then randomly permuted to produce another vector $\mathbf{\tilde{q}} \in \mathbb{N}_0^{n_s}$. The latter was used to create an $\mathbf{r} \in \mathbb{N}_+^{n_s}$ with
\begin{align}
r_j = \eta + \sum_{i=1}^j\tilde{q}_i, \quad j \in \mathcal{I}_{n_s},
\label{eq:TestCar6}
\end{align}
where $\eta$ is a random integer number (the same for all $j \in \mathcal{I}_{n_s}$) from a uniform distribution on $\mathcal{I}_n$, and the summation is assumed to be modulo $\mathcal{I}_n$. Finally, the carrier $\mathbf{c}$ was generated by taking the zero element of $\mathbb{R}^n$ and setting its components whose indexes are defined by the components of $\mathbf{r}$ to one.

\subsection*{Carriers (performance tests)}

For the performance, convergence, and robustness tests of the demodulation algorithm discussed in Sections~IV\,--\,VI of the main text, four types of carriers were considered:

\begin{itemize}

\item \textit{Nonstationary sinusoid},
\begin{equation}
\mathbf{c} = \cos\big(2\pi (f \mathbf{t} + \bm{\psi})\big),\label{eq:PerfAnl3}
\end{equation}
with $f=200\:\mathrm{Hz}$,
\begin{equation}
\mathbf{t} = f_s^{-1} \cdot (0, 1, \ldots, n-1)^{\mathrm{T}},\label{eq:PerfAnl3a}
\end{equation}

and
\begin{equation}
\bm{\psi} = (\mathbf{g}-\min[\mathbf{g}]) / \max[\mathbf{g}-\min[\mathbf{g}]],\label{eq:PerfAnl3b}
\end{equation}
where $\mathbf{g}=\mathbf{P}_{\Sw}[\mathbf{u}]$
\begin{align}
u_i \sim \mathcal{U}(0,1), \quad i \in \mathcal{I}_n,\label{eq:PerfAnl3c}
\end{align}
such that $u_i \indep u_j$ whenever $i \neq j$.

\item \textit{Nonstationary harmonic},
\begin{equation}
\mathbf{c} = \sum_{l=1}^{n_f} \Big( (2/3)^{l-1} \cdot \cos\big(2\pi l f (\mathbf{t} + \bm{\psi}_l) + \eta_l\big) \Big),\label{eq:PerfAnl4}
\end{equation}
with $f=85\:\mathrm{Hz}$, $n_f= \lfloor f_s/(2 \cdot f) \rfloor$,
\begin{align}
\eta_l &\sim \mathcal{U}(0,1), \quad l \in \mathcal{I}_{n_f},\label{eq:PerfAnl4a}
\end{align}

and
\begin{equation}
\bm{\psi}_l = \eta_l + 0.2 \cdot (\mathbf{g}-\min[\mathbf{g}]) / \max[\mathbf{g}-\min[\mathbf{g}]],\label{eq:PerfAnl4b}
\end{equation}
where $\mathbf{g}=\mathbf{P}_{\omega'}[\mathbf{u}]$ with $\omega' = 1\,\mathrm{Hz}$ and
\begin{align}
u_i &\sim \mathcal{U}(0,1), \quad i \in \mathcal{I}_n,\label{eq:PerfAnl4c}
\end{align}
such that $\eta_i \indep \eta_j$ and $u_i \indep u_j$ whenever $i \neq j$.

\item \textit{Nonstationary spikes},
\begin{align}
\textstyle
\mathbf{c} = \mathbf{1} - \theta \big( \mathbf{1}-\sum_{i=1}^{n_s} \mathbf{h} * \mathbf{e}^{(r_i)}\big) \circ \big( \mathbf{1}-\sum_{i=1}^{n_s} \mathbf{h} * \mathbf{e}^{(r_i)}\big) ,\label{eq:PerfAnl5}
\end{align}
where $\mathbf{e}^{(r_i)}$ is the unit vector with all but the $r_i$-th of its components equal to zero, $\mathbf{h}$ is defined by 
\begin{align}
h_i =\begin{cases} e^{-(i-1)^2/4}, & \mbox{if $1 \leq i \leq 11$}\\ e^{-(n-i+1)^2/4}, & \mbox{if $n-9 \leq i \leq n$}\\ 0, & \mbox{otherwise} \end{cases},\label{eq:PerfAnl5a}
\end{align}
and $r_1, r_2, \ldots, r_{n_s}$ is a sequence of different elements of $\mathcal{I}_n$ generated following
\begin{align}
r_i = r_{i-1} + d_i + \eta,\label{eq:PerfAnl5b}
\end{align}
until $r_{i-1} \leq n-d_i$ with $r_1=1$. In \eqref{eq:PerfAnl5b}, $\eta$ is a random number from a uniform distribution on $\mathcal{I}_n$ (the same for all $i$); $d_i$ is a random number taken independently from a uniform distribution on $\mathcal{I}_{z_i}$ at each iteration, where
\begin{align}
z_i = \lceil f_s / (2\varpi) \cdot (0.8+0.2 \cdot \sin(2\pi 5 t_i)) \rceil, \quad i \in \In. \label{eq:PerfAnl5c}
\end{align}
Note that, differently from the spike-train carriers generated for the purpose of recovery tests, \eqref{eq:PerfAnl5} defines sequences of finite-width spikes.

\item \textit{White-noise},
\begin{align}
c_i \sim \mathcal{U}(-1,1), \quad i \in \mathcal{I}_n,\label{eq:PerfAnl6}
\end{align}
with $c_i \indep c_j$ whenever $i \neq j$.

\end{itemize}

\newpage
\phantomsection

\markboth{PERFORMANCE TESTS}%
{PERFORMANCE TESTS}

\hspace{-11pt}\textbf{PERFORMANCE TESTS}
\addcontentsline{toc}{section}{PERFORMANCE TESTS}

\section{\textbf{Configurations of Demodulation Algorithms for Performance Analysis} \label{sec:SMPerformance}}

The following configurations of the AP and LDC demodulation algorithms were used for the performance analysis in Section~IV of the main text.

\subsection*{Window splitting}

Segment lengths $n_{seg} = \{2^6, 2^7, \ldots, 2^{15} \}$ were examined in the case of the AP demodulation approach. In the case of the LDC demodulation method, the range of segment lengths was more limited, $n_{seg} = \{2^6, 2^7, \ldots, 2^{11} \}$, to make the simulation times feasible. In order to reduce demodulation errors stemming from the window decomposition, signals were split into segments with a particular overlap. On top of that, each segment was windowed with the Hann function \cite{Sell2010}. For each segment length $n_{seg}$, different overlap spans were assumed: $n_{olp} = \{2^5, 2^6, \ldots, n_{seg} / 2 \}$. 
The AS-based demodulation was performed only with the original signal windows using no decomposition.  

\subsection*{Control parameters of the AP algorithms}

Overall, only two parameters are associated with the AP demodulation algorithms: 1) the cutoff frequency $\varpi$; and 2) the number of iterations $N_{iter}$. In all cases, we fixed $\varpi$ to $40~\mathrm{Hz}$. A set of different values of $N_{iter}$ was considered, dependent on the particular algorithm: the range from 1~to~600 for AP-B, 1~to~40 for AP-A, and 1~to~6000 for AP-P. It was made sure that the maximum iteration numbers in all these sets allow for reaching demodulation precision that is high enough not to affect the conclusions of the performance analysis.

\subsection*{Control parameters of the quadratic programming solvers for the LDC approach}

In the case of the LDC approach, all elements of the weighting vector $\mathbf{w}$ [see (11)] corresponding to frequencies below the threshold $\varpi$ were assumed to be zero. The remaining elements of $\mathbf{w}$ were set to either of $\{10^2, 10^3\}$. In the case of the OSQP solver, the following control parameters were tuned:
\begin{itemize}
\item Linear System Solver: \{\textit{Suite-Sparse-LDL}, \textit{MKL-Para\-di\-so}\};
\item Solution Polishing: $\{\mathit{false}, \mathit{true}\}$;
\item Warm Starting: $\{\mathit{false}, \mathit{true}\}$;
\item Absolute Tolerance: $\{10^{-2}, 10^{-3}, \ldots, 10^{-6} \}$;
\item Relative Tolerance~=~Absolute Tolerance;
\item Primal Infeasibility Tolerance~=~Absolute Tolerance;
\item Dual Infeasibility Tolerance~=~Absolute Tolerance;
\item Maximum Iteration Number: $2^{15}-1$;
\item Run Time Limit: 0.
\end{itemize}
The Gurobi solver was tried with these settings:
\begin{itemize}
\item Method: $\{ 0, 1, 2 \} \equiv$ \{\textit{primal-simplex}, \textit{dual-simplex}, \textit{barrier}\};
\item Optimality Tolerance: $\{10^{-2}, 10^{-3}, \ldots, 10^{-6} \}$;
\item Feasibility Tolerance~=~Optimality Tolerance;
\item Run Time Limit: 0.
\end{itemize}

\section{\textbf{Implementation on a Computer} \label{sec:SMProgrammatic}}

All demodulation algorithms considered  in the present work were implemented in C and then interfaced with MATLAB (R2018a) for a large-scale management of different instances and data analysis. The C code was compiled with GCC (v8.3) using no optimization. Evaluation of the discrete Fourier transform, used in the AS and AP approaches, relied on the FFT implementation of the Intel Math Kernel Library 2019 (update~5). The calculations were performed on a Lenovo ``ThinkCentre M910t Tower'' desktop computer with an Intel Core i7-7700 processor and Linux Ubuntu 16.04 operating system. In order to minimize the influence of other operating system processes on the benchmarking results, one of the four CPU cores was dedicated exclusively to the execution of the demodulation program. This was achieved by using the ``isolplus'' option of the kernel scheduler. All computations were done in single-thread mode.

\newpage
\phantomsection

\markboth{ADDITIONAL RESULTS}%
{ADDITIONAL RESULTS}

\hspace{-11pt}\textbf{ADDITIONAL RESULTS}
\addcontentsline{toc}{section}{ADDITIONAL RESULTS}

\section{\textbf{Errors of Carrier Estimates} \label{sec:SMCarrierEstimates}}

Here, we discuss findings from the performance and robustness analyses of demodulation algorithms (introduced in Sections~IV and VI of the main text) in terms of carrier estimates. Analogous to the modulator recovery error $E_m$, we use $E_c = \|\mathbf{c}-\mathbf{\hat{c}} \|_2/\|\mathbf{c}\|_2$ next.

\subsection*{Performance tests}

\begin{figure*}[h]
\centering
\includegraphics[width=1\textwidth]{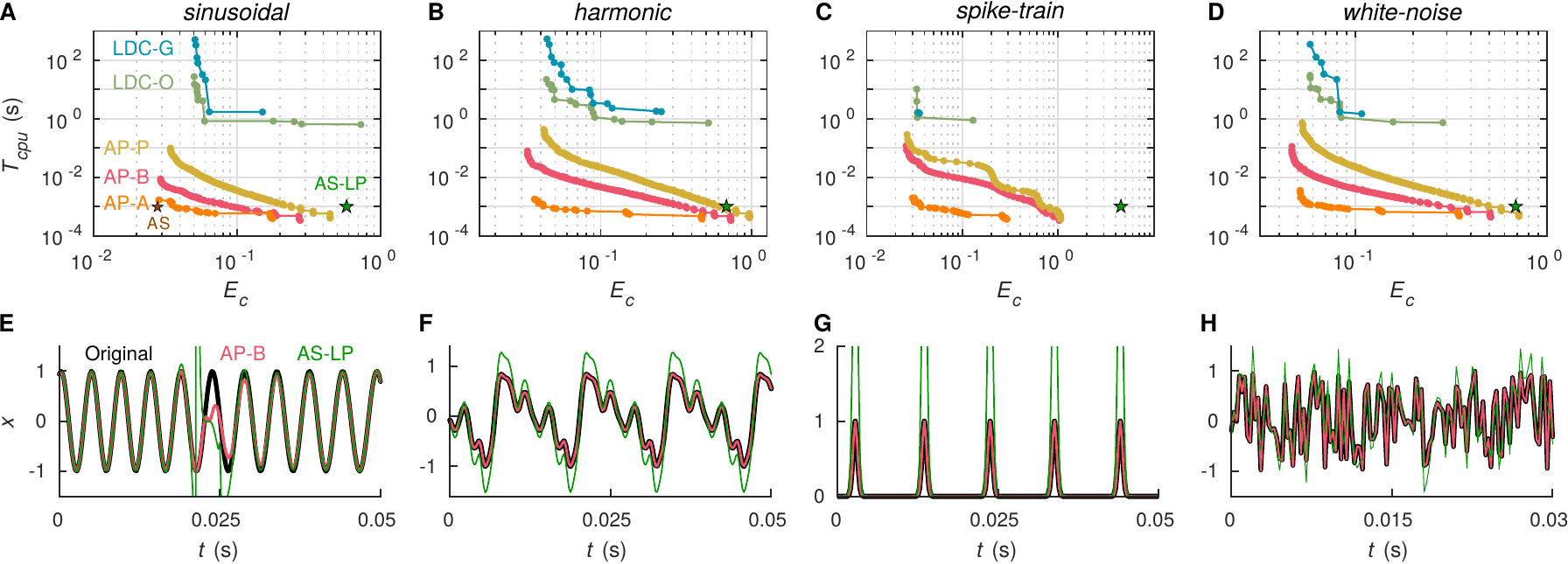}
\caption{\footnotesize Performance evaluation. \textbf{A}--\textbf{D}: Pareto fronts in the $(E_c,T_\mathrm{cpu})$ plane for different demodulation algorithms applied to the four different types of test signals when window splitting is used. The green and brown stars mark the results of, respectively, the AS-LP and AS methods. \textbf{E}--\textbf{H}: Examples of the original carriers considered in the present work (black) and their estimates obtained by using the AP-B (red) and AS-LP (green) algorithms.}
\label{fig:16}
\end{figure*}

The high precision of modulator estimates (see Section~IV-C in the main text) and the boundedness of $\hat{c}_i$ between $-1$ and $1$ (see Section~VII in the main text) predetermine good quality carrier estimates provided by the AP approach.\footnote{Remember that $c_i/\hat{c}_i = m_i/\hat{m}_i$ in our case.} This view is evidenced by the Pareto fronts in the $(E_c,T_\mathrm{cpu})$ plane shown in Fig.\,\ref{fig:16}\,A--D and comparisons of exemplary $\mathbf{c}$ and $\mathbf{\hat{c}}$ in Fig.\,\ref{fig:16}\,E--H. We observed noticeable discrepancies between $\mathbf{c}$ and $\mathbf{\hat{c}}$ only locally, around points with modulator levels very close to zero (see the signal segment at $t=0.025$ in Fig.\,\ref{fig:16}\,E). That finding is explained by the fact that, in our case, $\hat{c}_i - c_i = s_i \cdot (\hat{m}_i^{-1} - m_i^{-1})$. Hence, even small differences between $m_i$ and $\hat{m}_i$ on the absolute scale can give notable differences between $c_i$ and $\hat{c}_i$ if $m_i/\hat{m}_i \gg 1$. However, the condition $\hat{m}_i \geq s_i$ ($\implies |c_i| \leq 1$) inherent to all $\mathbf{\hat{c}}$ obtained by the AP algorithms assures tolerable distortions of the carrier estimates even around the points with vanishingly small $m_i$.

The situation is different in the case of the AS-LP approach. Then, $|s_i / \hat{m}_i| \gg 1$, i.e., $|\hat{c}_i| \gg 1$, are possible (see the signal segment at $t=0.025$ in Fig.\,\ref{fig:16}\,E). These divergences noticeably increase the recovery errors $E_c$ even for sinusoidal signals that AS-LP recovers sufficiently well outside the segments of vanishing $m_i$ (see Fig.\,\ref{fig:16}\,A,\,E). Globally-inaccurate recovery of other carrier types demonstrated by the AS-LP (see Fig.\,\ref{fig:16}\,B--\,D, \,F--\,H) is predetermined by inaccurate  estimates of the modulators (see Section~IV-C in the main text). We note that, for sinusoidal carriers, appropriate estimates $\mathbf{\hat{c}}$ can be obtained by using the original AS method instead of the AS-LP (Fig.\,\ref{fig:16}\,A). However, the AS gives very erroneous modulator estimates $\mathbf{\hat{m}}$, and hence, does not improve carrier predictions $\mathbf{\hat{c}}$ considerably, for other types of signals (data not shown).

\subsection*{Robustness tests}

\begin{figure*}[h]
\centering
\includegraphics[width=1\textwidth]{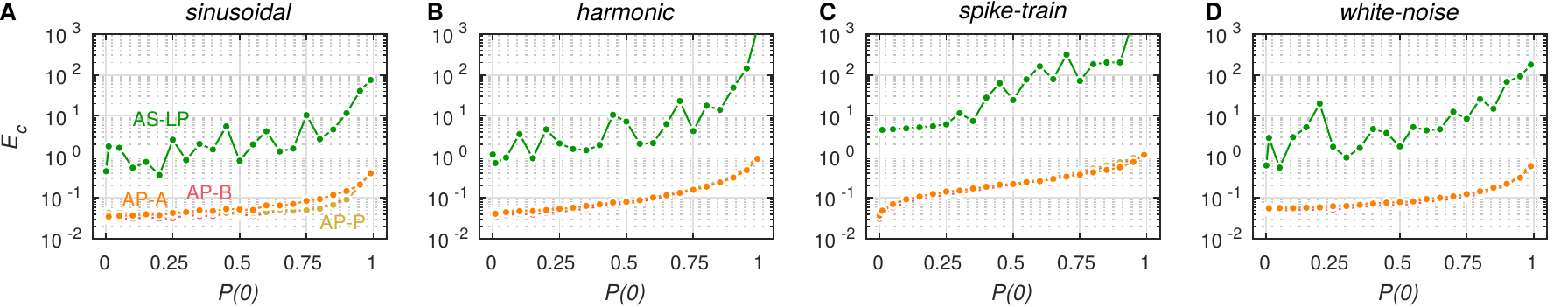}
\caption{\footnotesize Robustness evaluation. \textbf{A}--\textbf{D}: Dependence of the carrier recovery error $E_c$ on $P(0)$ (the probability of missing points) for the four types of test signals and different AP algorithms at $\epsilon_{tol}=10^{-4}$ (color coding).}
\label{fig:16b}
\end{figure*}

Fig.\,\ref{fig:16b} shows demodulation results of the AS-LP, AP-B, AP-A, and AP-P algorithms in terms of carrier recovery error for test signals corrupted by the multiplicative Bernoulli-$\{0,1\}$ noise (see Section~VI). The obtained $E_c$~vs.~$P(0)$ relations for the AP algorithms are analogous to their counterparts $E_m$~vs.~$P(0)$ shown in Fig.\,5. In contrast, the AS-LP is even more inferior to the AP approach in terms of carrier reconstruction (Fig.\,\ref{fig:16b}) than it is in terms of modulator recovery (Fig.\,5\,A--D). This fact is explained by the presence of the $|\hat{c}_i| \gg 1$ divergences discussed above and illustrated by Fig.\,\ref{fig:16}\,E.

\section{\textbf{Convergence Rates at Different Signal Lengths} \label{sec:SMConvergence}}

\begin{figure*}[ht]
\centering
\includegraphics[width=1\textwidth]{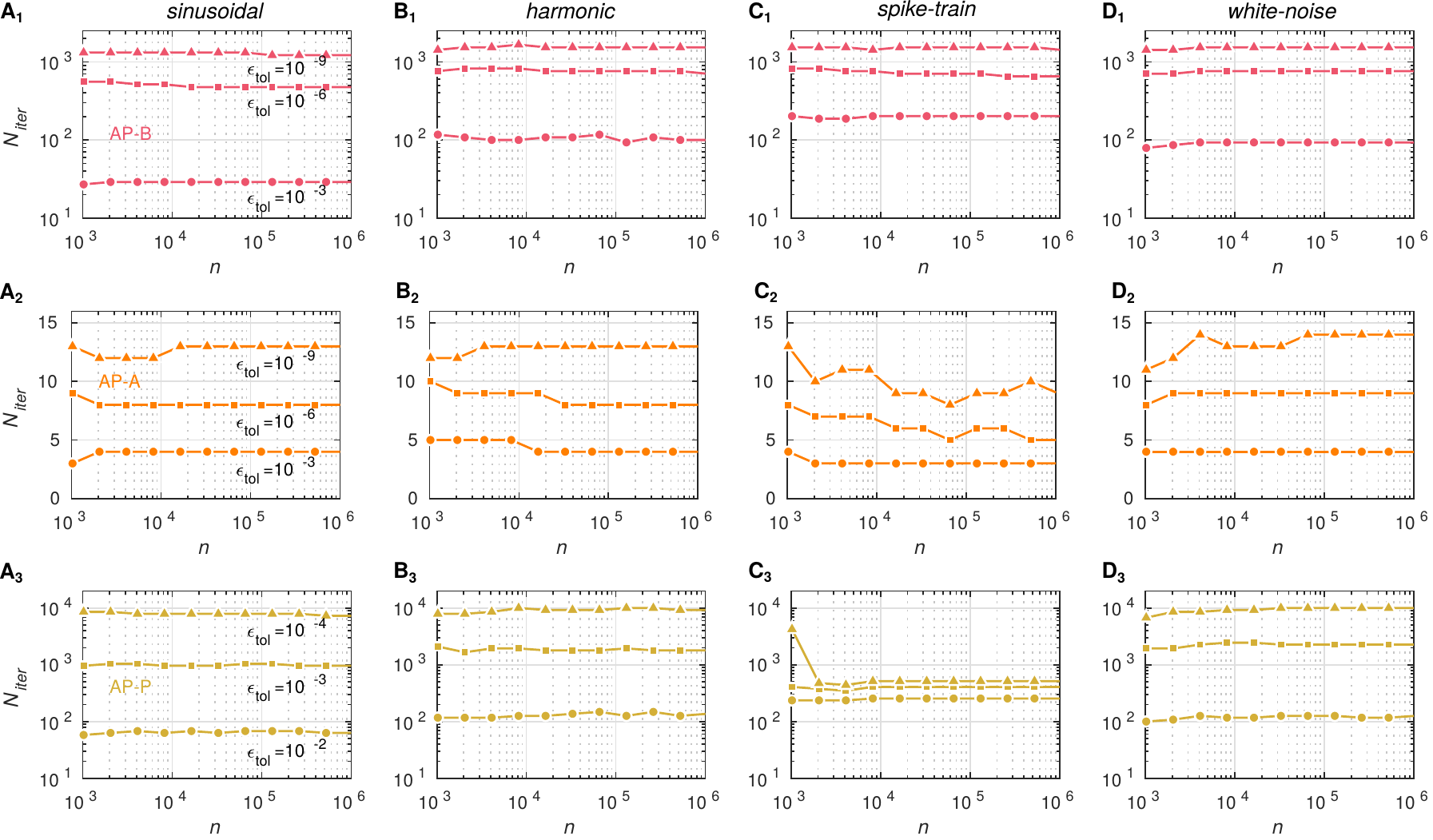}
\caption{\footnotesize Convergence analysis of the AP algorithms at different $n$ and fixed $f_s=4\,\mathrm{kHz}$. $\mathbf{A_1}$--$\mathbf{D_1}$: Dependence of the iteration number $N_{iter}$ necessary to reach a specific infeasibility error $\epsilon$ on the length of the signal for the AP-B algorithm applied to four different classes of test signals. Filled circles, rectangles, and triangles label curves for $\epsilon$ levels of, respectively, $10^{-3}$, $10^{-6}$, and $10^{-9}$. $\mathbf{A_2}$--$\mathbf{D_2}$: The same as $\mathrm{A_1}$--$\mathrm{D_1}$ but shown for the AP-A algorithm. $\mathbf{A_3}$--$\mathbf{D_3}$: The same as $\mathrm{A_1}$--$\mathrm{D_1}$ but shown for the AP-P algorithm and different levels of the infeasibility error $\epsilon$.}
\label{fig:10}
\end{figure*}

The convergence results shown in Fig.\,4 of the main text represent only signals with the length $n$ fixed to $2^{15}$ sample points. Therefore, we performed additional simulations with different $n$ values to test the impact of this parameter on the rate of the iterative process. We found no clear dependence of $N_{iter}$ required to achieve a particular infeasibility error value $\epsilon$ on $n$, except transitional changes due to diminishing contributions of the boundary effects in some cases (see Fig.\,\ref{fig:10}\,$\mathrm{C_2}$\,--\,$\mathrm{C_2}$). These findings suggest that the increased $T_{\mathrm{cpu}}$ of demodulation for longer $n$ is primarily determined by the increased computational demands of single projections (see Section~III-D in the main text).

\section{\textbf{Repeated Demodulation} \label{sec:SMRedemodulation}}

\begin{figure*}[ht]
\centering
\includegraphics[width=1\textwidth]{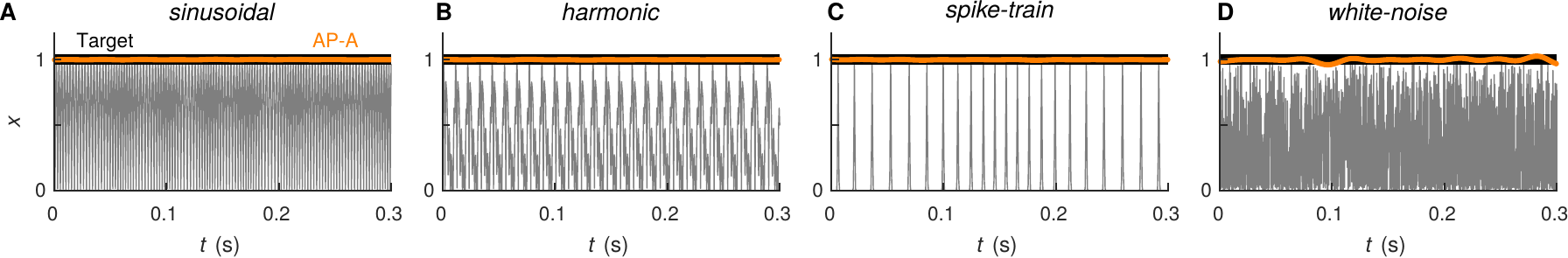}
\caption{\footnotesize Repeated demodulation of inferred carriers. Gray -- sinusoidal (\textbf{A}), harmonic (\textbf{B}), spike-train (\textbf{C}), and white-noise (\textbf{D}) carriers inferred by demodulating test signals shown in Fig.\,2\,A--D with the AP-A algorithm. Black -- target repeated modulators following from the assumption that the inferred carriers are fully demodulated. Orange -- the real repeated modulators obtained with AP-A.}
\label{fig:12}
\end{figure*}

Fig.\,\ref{fig:12} shows typical results of redemodulation of carriers inferred from the four types of synthetic test signals considered in the present work by using the AP-A algorithm. Redemodulation of the carriers returns the identity modulator to a very good approximation ($E_m \leq 10^{-2}$), implying a nearly complete separation of the modulator and carrier information in the first step, as discussed in Section~VII of the main text.

\section{\textbf{Demodulation of Speech Signals} \label{sec:SMSpeech}}

In this section, we present results of additional simulations used to support the statements about the suitability of the AP approach to demodulate wideband speech signals in Section~VIII of the main text.

\subsection*{Synthetic modulators \label{subsec:SMSpeech1}}

The first question that we considered was up to which values of the cutoff frequency $\omega$ the natural speech carriers meet the recovery condition $\lceil n / d \rceil \geq 2\omega-1$. As mentioned in the main text, these carriers are of quasi-random and quasi-harmonic origins, possibly featuring frequency glides. It is well known that typical fundamental frequencies ($f_0$) of male and female speaker voice are, respectively, 120 and 210\,Hz (see \cite{Traunmuller1994} and Table\,1 therein). Hence, at least for the harmonic components, the condition $\lceil n / d \rceil \geq 2\omega-1$ is expected to be satisfied with $\omega \leq 60\,\mathrm{Hz}$. That is more than sufficient for appropriate demodulation, assuming that strictly all spectral energy of the speech amplitude modulator is located below 20\,Hz (see Section~VIII-A). The situation with the quasi-random and mixed components of speech signals is less certain and must be tested numerically.

\begin{figure*}[ht]
\centering
\includegraphics[width=1\textwidth]{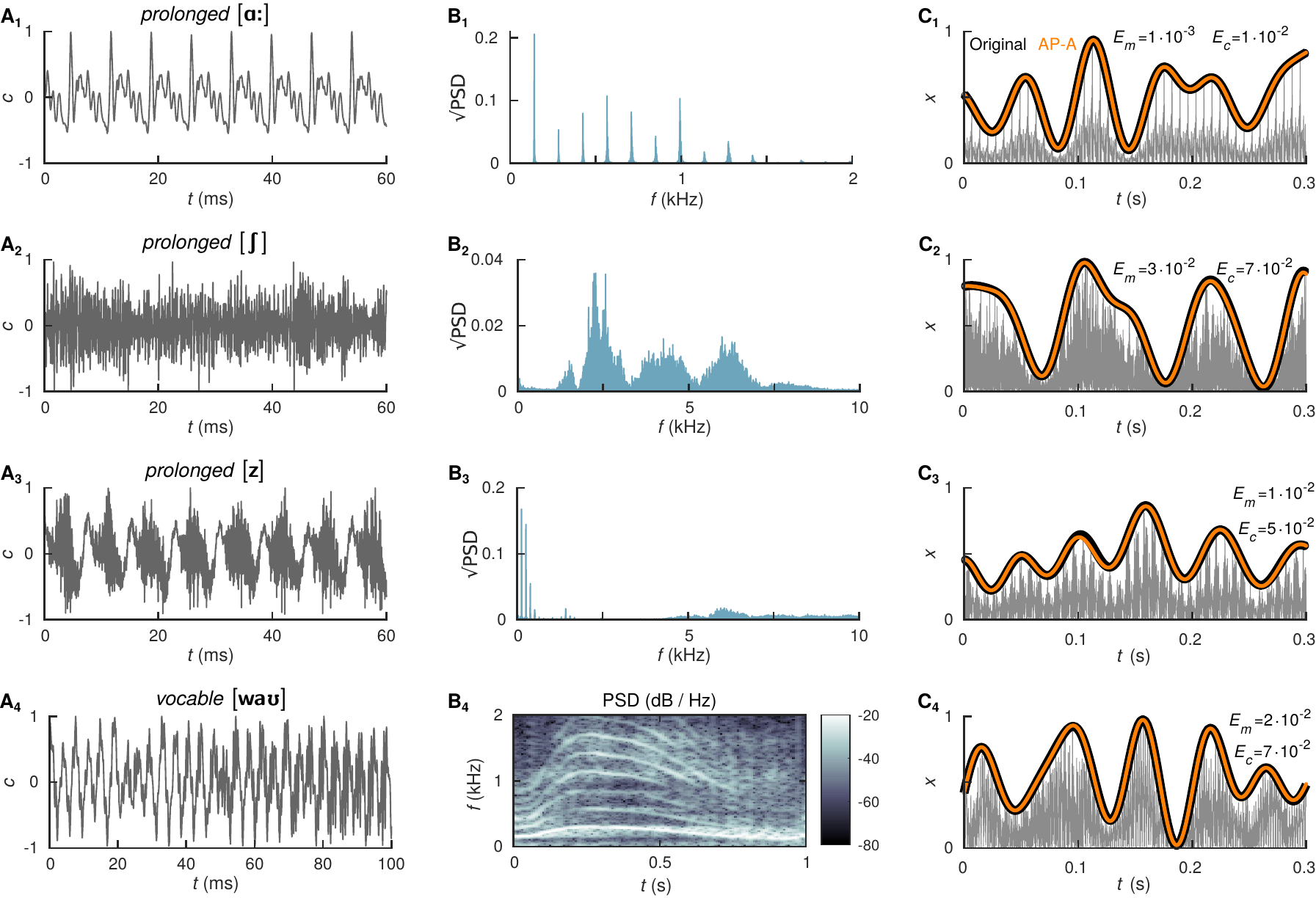}
\caption{\footnotesize Demodulation of signals built of synthetic modulators and natural speech carriers. $\mathbf{A_1}$--$\mathbf{A_4}$: fragments of carrier-signals of prolonged steady [$\upalpha$:] ($\mathrm{A_1}$), [{\scriptsize$\int$}] ($\mathrm{A_2}$), and [z] ($\mathrm{A_3}$), as well as an extended vocable [wa{\tiny$\mho$}] ($\mathrm{A_4}$). $\mathbf{B_1}$--$\mathbf{B_3}$: periodogram estimators of the power spectral densities of the carriers illustrated in $\mathrm{A_1}$--$\mathrm{A_3}$; $\mathbf{B_4}$: spectrogram of the vocable [wa{\tiny$\mho$}]. $\mathbf{C_1}$--$\mathbf{C_4}$: exemplary segments of amplitude-modulated carriers from panels $\mathrm{A_1}$--$\mathrm{A_4}$, their modulators (black), and modulator estimates obtained with the AP-A algorithm. Insets of panels $\mathrm{C_1}$--$\mathrm{C_4}$ display values of the modulator and carrier recovery errors.}
\label{fig:14}
\end{figure*}

To this end, we considered four speech-carriers generated by a male speaker uttering prolonged ($\sim1$\,s) [$\upalpha$:], [{\scriptsize$\int$}], and [z], as well as vocable [wa{\scriptsize$\mho$}] (see Fig.\,\ref{fig:14}\,$\mathrm{A_1}$\,--\,$\mathrm{A_4}$) at $f_0=120\,\mathrm{Hz}$ without noticeable amplitude modulation. The [$\upalpha$:] is predominantly harmonic, [{\scriptsize$\int$}] is quasi-random, [z] has a mixed wave-shape, and [wa{\scriptsize$\mho$}] is quasi-harmonic but with upward and downward frequency glides. These characteristics are further revealed by the power-spectral-density (PSD) plots for the [$\upalpha$:], [{\scriptsize$\int$}], and [z] (see Fig.\,\ref{fig:14}\,$\mathrm{B_1}$\,--\,$\mathrm{B_3}$), and a spectrogram for the [wa{\scriptsize$\mho$}] (Fig.\,\ref{fig:14}\,$\mathrm{B_4}$). We then created test signals as products of mentioned carriers and maximally-uniformly distributed synthetic modulators (see Section~\ref{sec:SMSynthetic}) with $\omega=25\,\mathrm{Hz}$ and $f_s=44.1\,\mathrm{kHz}$.

Demodulation of the considered test signals with the AP-A algorithm allowed us to achieve high-accuracy modulator recovery, as shown in Fig.\,\ref{fig:14}\,$\mathrm{C_1}$\,--\,$\mathrm{C_4}$ (note the insets with $E_m$ and $E_c$ values there). In more detail, we found that, for the [$\upalpha$:], [{\scriptsize$\int$}], [z], and [wa{\scriptsize$\mho$}], the distances between carrier points with absolute values of, respectively, $\geq 0.99$, $\geq 0.95$, $\geq 0.93$, and $\geq 0.93$ were below that needed to satisfy $\lceil n / d \rceil \geq 2\omega-1$ at $\omega=20\,\mathrm{Hz}$. For the [{\scriptsize$\int$}], [z], and [wa{\scriptsize$\mho$}], 80\,\% of the required points were $\geq 0.99$. The demodulation quality remained reasonably good when $\omega$ was increased to $50\,\mathrm{Hz}$, resulting in $E_m$ values of $1 \cdot 10^{-2}$, $5 \cdot 10^{-2}$, $5 \cdot 10^{-2}$, and $5 \cdot 10^{-2}$ for, respectively, [$\upalpha$:], [{\scriptsize$\int$}], [z], and [wa{\scriptsize$\mho$}] carriers.

\subsection*{Natural modulators \label{subsec:SMSpeech2}}

Next, we aimed to clarify whether the AP approach can properly separate $\mathbf{m}$ and $\mathbf{c}$ of speech signals using the dynamic range compression discussed in Section~VIII-B of the main text. For this purpose, we considered synthetic time series built of the four natural carriers [$\upalpha$:], [{\scriptsize$\int$}], [z], and [wa{\scriptsize$\mho$}] introduced above and modulator estimate $\mathbf{\hat{m}^*}$ of an utterance \textit{``\ldots protein which forms p \ldots''} from Section~VIII-B of the main text. We then demodulated the resulting signals by using the AP-A algorithm combined with the dynamic range compression. The obtained estimates $\mathbf{\Hat{\Hat{m}}^*}$ were in good agreement with the original, as shown in Fig.\,\ref{fig:17} (note the insets with $E_m$ and $E_c$ values there). The AP-B, AP-P, and LDC algorithms returned very similar modulator estimates but needed much longer computing times that mirror the performance analysis presented in Section~IV-C of the main text.

\begin{figure*}[ht]
\centering
\includegraphics[width=0.985\textwidth]{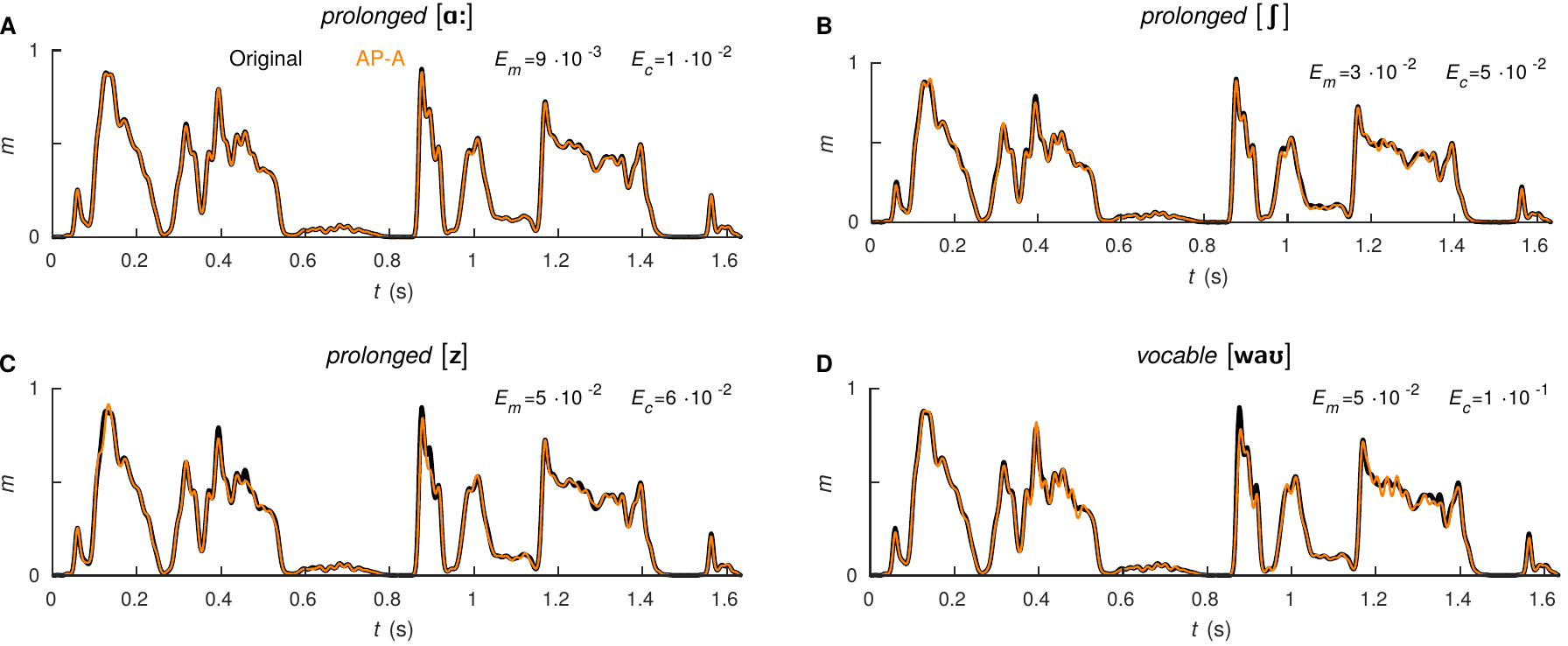}
\caption{\footnotesize Demodulation of signals built of natural speech modulators and carriers. $\mathbf{A}$--$\mathbf{D}$ compares original modulators $\mathbf{\hat{m}^*}$ (black) and their estimates $\mathbf{\Hat{\Hat{m}}^*}$ obtained by using the AP-A algorithm combined with the dynamic range compression (orange) for, respectively, [$\upalpha$:], [{\scriptsize$\int$}], [z], and [wa{\scriptsize$\mho$}] carriers. Insets display values of the modulator and carrier recovery errors.}
\label{fig:17}
\end{figure*}

\newpage
\phantomsection

\markboth{REFERENCES}%
{REFERENCES}

%
\bibliographystyle{IEEEtran}
\bibliography{supplement}
